\documentclass[]{article}
\usepackage{graphics} 
\usepackage{epsfig} 
\usepackage{amssymb}  %
\usepackage{upgreek}
\usepackage{forest} 
\usepackage{comment}
\usepackage{cancel}
\usepackage{url}
\usepackage{hyperref}
\usepackage{soul}
\usepackage{graphicx,psfrag, color}
\usepackage{latexsym}
\usepackage{subfigure}
\usepackage{bbm}
\usepackage{latexsym}
\usepackage{lettrine}%
\usepackage{dsfont}
\usepackage{array}
\usepackage{amsthm}

\usepackage[english]{babel}
\usepackage{graphicx}
\usepackage{pgf, tikz, pgflibraryshapes}
\usepackage{amssymb,bm}
\usepackage{colortbl}
\usepackage{multicol}
\usepackage{stackrel}
\newcommand{\mc}{\mathcal}
\newcommand{\e}{\mathrm{e}}
\newcommand{\E}{\mathcal{E}}

\newcommand{\R}{\mathbb{R}}
\newcommand{\N}{\mathbb{N}}
\newcommand{\argmax}[1]{\underset{#1}{\mathrm{argmax\,}}}

\newcommand{\mb}{\boldsymbol}
\newcommand{\se}{\text{ if }}
\newcommand{\ba}{\begin{array}}
\newcommand{\ea}{\end{array}}
\usepackage{color}
\usepackage[mathscr]{euscript}

\usepackage{amsfonts}

\newcommand{\Z}{\mathbb{Z}}  

\newcommand{\G}{\mathcal{G}} 
\newcommand{\V}{\mathcal{V}} 

\newcommand{\1}{\boldsymbol{1}} 
\newcommand{\ind}{\boldsymbol{1}} 

\newcommand{\equaref}[1]{(\ref{eq:#1})}

\newtheorem{corollary}{Corollary}

\newtheorem{theorem}{Theorem}
\newtheorem{lemma}{Lemma}
\newtheorem{assumption}{Assumption}
\newtheorem{proposition}{Proposition}
\newtheorem{definition}{Definition}
\newtheorem{remark}{Remark}

\newtheorem{example}{Example}

\newcommand{\sidebyside}[4]{
\begin{figure}[phtb]
\begin{minipage}[phtb]{6cm}
\includegraphics[width=0.98\textwidth]{#1.pdf}
\caption{#2\label{fig:#1}}
\end{minipage}
\hfill \hspace{0.1cm}
\begin{minipage}[phtb]{6cm}
\includegraphics[width=0.98\textwidth]{#3.pdf}
\caption{#4\label{fig:#3}}
\end{minipage}
\hfill
\vspace{-0mm}
\end{figure}
}

\usepackage{amsmath}
\usepackage{siunitx}
\usepackage{amsfonts}

\date{}
\title{Asynchronous semi-anonymous dynamics over large-scale networks} 
\author{ C. Ravazzi, G. Como, M. Garetto, E. Leonardi, A. Tarable}

\begin{document}
\maketitle

\begin{abstract}
We analyze a class of stochastic processes, referred to as {\em asynchronous} and {\em semi-anonymous dynamics} (ASD), 
over directed labeled random networks. 
These processes are a natural tool to describe general best-response and noisy best-response dynamics in network games where each agent, at random times governed by  
 independent Poisson clocks, can choose among a {\em finite set of actions}.  The payoff is determined by  the relative popularity of different 
actions among neighbors, while being {\em independent of the specific identities of neighbors}.

Using a mean-field approach, we prove that, under certain conditions on the network and initial node configuration, 
the evolution of ASD can be approximated, in the limit
of large network sizes, by the solution of a system of non-linear ordinary differential equations. Our framework is very general and applies to a large class of 
graph ensembles for which the typical random graph locally behaves like a tree. 
In particular, we will focus on labeled configuration-model random graphs, a generalization of the traditional configuration model
which allows different classes of nodes to be mixed together in the network, permitting us, for example, to incorporate a community
structure in the system. Our analysis also applies to configuration-model graphs having a power-law degree distribution, an essential feature of many real systems.  
To demonstrate the power and flexibility of our framework, we consider several
examples of dynamics belonging to our class of stochastic processes.
Moreover, we illustrate by simulation the applicability of our analysis to realistic scenarios by 
running our example dynamics over a real social network graph.   
\end{abstract}


\section{Introduction}

Many complex systems arising in different domains exhibit cascading phenomena that spread through networks of local interactions. Examples of such cascades include, but are not limited to, infrastructure failures \cite{Wang2013}, adoption of innovations, conventions and technologies \cite{UBHD2028615,Toole2012,Montanari20196}, diffusion of beliefs, opinions, fake news \cite{Petter2018}, memes, and the like \cite{Vosoughi1146}. These phenomena can have profound effects on politics \cite{10.5555/1805895}, social norms \cite{doi:10.1111/mafi.12051}, financial networks \cite{Granovetter78}, marketing campaigns \cite{Jackson2008}.

The standard mathematical approach to modeling cascading processes is to consider a graph (finite or infinite) in which nodes stand for individuals that can be in one of several (discrete or continuous) states, and edges (directed or undirected, possibly weighted) represent interactions with neighboring nodes. Individuals are supposed to repeatedly update their state over (discrete or continuous) time, depending on the current state of their neighbors \cite{schelling1978micromotives,BLUME1993387,RePEc:ecm:emetrp:v:61:y:1993:i:5:p:1047-71}. 

Simple epidemic models in which nodes can change their state as consequence of a single contact with a neighboring node \cite{Newman02} turn out to be too simplistic to describe systems in which individuals tend to react to the joint states of their neighbors. To represent such combined effect, one of the most commonly used models in the literature is the linear threshold model, originally introduced by Granovetter \cite{Granovetter78} and widely investigated in several variants \cite{Watts02,Zhong17}. The general idea behind such models is to assume that a node adopts a given state if the fraction of neighbors (possibly weighted by edges) currently adopting that state exceeds a certain threshold.          
More in general, researchers have considered so-called networked coordinated games, in which nodes adopt the best-response (according to some payoff matrix) in reaction to the strategies adopted by neighbors \cite{Broere2017}.     

Some fundamental distinctions in this wide class of models are the following. In progressive processes, state transitions are irreversible: once a node joins a given state, it keeps such state indefinitely, irrespective of what happens to neighbors \cite{Kempe03,RePEc:eee:phsmap:v:374:y:2007:i:1:p:449-456,6160999,Lim16,doi:10.1111/mafi.12051}. In non-progressive processes transitions are, instead, reversible, since nodes still remain under the influence of their neighbors after the adoption of a given state \cite{6426073,Morris2000Contagion,RCF18}. Another crucial distinction concerns the update rule of the nodes: does the future state of a node also depend on the current state of the node, in addition to the states of neighbors, or is it uniquely determined by the neighbors? Indeed the above distinctions, combined to the nature of edges (i.e., directed or undirected), lead to models of widely different nature and analytical tractability.

In this paper we analyze a class of cascading processes, referred to as Asynchronous and Semi-anonymous Dynamics (ASD), with the following characteristics: i) nodes update their (discrete) state over continuous time, according to independent Poisson clocks, 
following an arbitrary rule that depends on the number (or relative fraction) of neighbors in each possible state, but not on the specific identities of neighbors (which are order-independent); ii) edges are directed; iii) state transitions are reversible (non-progressive model).

In this paper we extends previous work in \cite{RCF18}, where authors analyze
a two-state, deterministic linear threshold model with synchronous node update, adopting
a mean-field approach. Here, we seek to understand how this 
approach can be pushed to its greatest generality, extending the class of 
networks and underlying node dynamics to which it can be applied.  
We mention, however, that in \cite{RCF18} authors consider also the progressive variant
of their model, while here we focus only on the non-progressive model. 
The interested reader can refer to \cite{DBLP:conf/sigecom/ErolPT20} for new methods to approximate the dynamics on large networks sampled from a graphon. Although this framework is flexible, it works well on dense network formation models, such as stochastic block models, but not on sparse networks.

One stream of related work \cite{Blume95, Morris2000Contagion, Lelarge12} analyzes the possible equilibria of 
a binary-decision game in the case of undirected graphs and  synchronous update.
Similarly to these works, we study a game where players'payoffs depend on the actions taken by their neighbors in the network 
but not on the specific identities of these neighbors. 
   
Another huge stream of related work is concerned with the algorithmic aspects of influence maximization \cite{Kempe03,Lynn2016}, where
the goal is to find the initial node configuration that maximizes the final size of the cascade. In contrast to such stream of work,
here we assume the initial node configuration to be randomly selected according to a given node statistics.    


\subsection{Overview of main results}In our analysis we will make use of a mean-field approximation. This is a powerful and standard tool 
for studying complex systems that are constituted by a large number of interacting units \cite{Bordenave2010APS}. 
The core idea in the mean-field theory is to describe a complex stochastic system 
using a simpler deterministic dynamical system \cite{Kurtz70,Bordenave2010APS}, assuming that each unit interacts 
with an average of other units. When the states can take values in a finite or 
countable space, the evolution of the fraction of units in a given state can be described by
the solution of a system of non-linear ordinary differential equations.
The main goal is then to quantify and estimate the accuracy of this approximation. 
We anticipate that our results guarantee that the discrepancy between the behavior of the actual system
and the solution given by the approximation can be made small when the network size grows and the time horizon scales only logarithmically with the network size. 

In Section \ref{sec:bp} we provide the complete analysis of ASD over 
a labeled branching process, i.e., an infinite ensemble of labeled graphs with a rooted 
tree structure. In this case it can be easily shown that the evolution of the fraction of nodes 
in a given state indeed corresponds to the solution of some ordinary differential equations (see Proposition \ref{prop:ODE}). 

In Section \ref{sec:rn} we turn our attention to general graph ensembles. 
The core message is the following: if the graph exhibits a local tree structure, then the analysis 
on a suitably chosen labeled branching process provides a 
good approximation of the expected fraction of nodes in a given state (see Proposition \ref{thm:ASD_over_gen_net}). A property of Local Weak Convergence is the key feature that provides this link and formalizes the idea that, 
for large $n$, the local structure of the graph near a vertex chosen uniformly at random 
is approximately a branching process.
Finally, the analysis of the concentration around the expectation allows us to derive 
the accuracy of the above approximation.

We will then focus in Section \ref{sec:labconfmod}
on the labeled configuration model, a general mix of heterogeneous nodes 
with class-specific node statistics, which includes the classic configuration model (CM) 
as special case. We will explore conditions for Local Weak Convergence (see Theorem \ref{thm:top_res_modgraphs}) and the concentration of the ASD evolution around the expectation (see Theorem \ref{thm:concentration_modgraphs}).
In Section \ref{sec:asdconfig} we will show that the sequence of node degrees is allowed to follow a 
power-law distribution scaling with the network size. 
This is particularly important for applications to social network graphs.

As a second example, in Section \ref{section:CBM2} we will consider a labeled configuration model with a community structure,
which is another fundamental feature found in many real systems.
Indeed, by considering as label of a node its membership to a given community,
we can represent graphs with  a general distribution of in/out degrees
among nodes belonging to the same or different communities. This allows us
to describe, for example, \lq\lq assortative"  graphs, in which
intra-community edges are denser than 
inter-community edges.

In Section \ref{sec:mf} and \ref{sec:experiments} we present numerical results of ASD previously introduced in large but finite networks, considering both synthetically generated graphs and real-world social networks. More precisely, we compare the solutions of the differential equations derived through mean field approximation with results obtained by running Monte-Carlo simulations.

\subsection{General notation}\label{sec1} 
Throughout this paper, we use the following notational conventions. 
Let $\mathbb{N}$, $\mathbb{Z}_+$, $\mathbb{ R}$ be the set of natural, non-negative integers and real numbers, respectively. Given $n\in\N$ we use the notation $[n]=\{1,\ldots, n\}$.The symbol $|\cdot|$ denotes the absolute value if applied to a scalar value and the cardinality if applied to a set. We denote the indicator function of set $A$ with the notation $\ind_A$.
Given a finite set $\mathcal{V}$, $\R^{\mathcal{V}}$ denotes the space of real vectors with components labelled  by elements of $\mathcal{V}$.
If $x\in\R^n$, we denote the $\ell$-th entry by $x_{\ell}$ and $x_{[\ell]}$ is the projection of $x$ on the sub-space generated by  the first $\ell$ elements, i.e. $x_{[\ell]}=(x_1,x_2,\ldots, x_{\ell})\in\R^{\ell}$. 

This paper makes frequent use of Landau symbols.  The notation ``$f(N) = O(g(N))$ when $N \rightarrow\infty$'' means that positive constants $c$ and $N_0$ exist, so that $f(N) \leq cg(N)$ for all $N > N_0$. The expression ``$f(N) = o(g(N))$ when  $N \rightarrow\infty$'' means that $\lim_{N\rightarrow \infty}\frac{f(N)}{g(N)}=0$.

A labeled directed multigraph is a $6$-ple $(\V,\E,\mathcal A,\lambda,\sigma,\tau)$,  where $\mc V$, $\mathcal E$, and $\mc A$ are the sets of nodes, links, and labels,  respectively, all finite; $\lambda:\mc V\to\mc A$ is the map giving the label $\lambda(v)$ of a node $v\in\mc V$; and $\sigma,\tau:\mc E\to\mc V$ are the maps giving the tail node $\sigma(e)$ and head node $\tau(e)$ of a link $e\in\mc E$ so that $e$ is directed from $\sigma(e)$ to $\tau(e)$. 
 The set of in-neighbors and out-neighbors of a node $v\in\V$ are defined as 
$\mathcal{N}^{-}_v=\{w\in\V\setminus\{v\}: (w,v)\in\E\}$ and $\mathcal{N}^{+}_v=\{w\in\V\setminus\{v\}: (v,w)\in\E\}$, respectively, and the corresponding in-degree and out-degree as $d_v=|\mathcal{N}^{-}_v|$ and $k_v=|\mathcal{N}^{+}_v|$. We define its out-degree vector  $\textbf{k}_v\in\mathbb{Z}_+^{\mathcal{A}}$ as the vector whose component $a\in \mathcal{A}$ 
represents the number of out-neighbors of $v$  belonging to class  
$a$. Similarly, we define for node $v$ the in-degree vector  $\textbf{d}_v\in\mathbb{Z}_+^{\mathcal{A}}$. 

A path from a vertex $u\in\mathcal{V}$ to a vertex $v\in\mathcal{V}$ (i.e. a path  $u\to v$) is a  finite  sequence of edges $(u_i, v_i)_{1\le i\le L}$ with $u_1=u$, $v_L=v$, $v_i=u_{i+1}$. If there is at least a path from $u$ to $v$, we say that $u$ is connected to $v$, and the graph distance from $u$ to $v$ is then defined as the minimum length of a path from $u$ to $v$.  If all (ordered) vertex pairs  are connected, the graph $\mathcal{G}$ is said strongly connected.  If, instead, for any pair $(u,v)$ either a path $u\to v$ or $v\to u$ exists, we say that the graph is weakly connected.
 We define a simple path as a path along which all vertices are distinct. A directed tree is a weakly connected graph  in which no more than one
  path exists between every  pair of vertices $(u,v)$.

\section{Asynchronous semi-anonymous dynamics}
\subsection{Mathematical model}
Let us consider a finite population of $n$ agents interacting in a connected network, which we map onto the nodes of a labeled directed multigraph $\mathcal{G}=(\mathcal{V}, \mathcal{E},\mc A,\sigma,\tau,\lambda)$, whereby a link $e\in\mc E$ represents a direct influence of its head node $\tau(e)$ on its tail node $\sigma(e)$ and each class of nodes $\mc V_a=\{v\in\mc V:\,\lambda(v)=a\}$, for $a\in\mc A$ may have a different behavior, thus allowing to account for heterogeneity.

Let each agent $v\in\mc V$ be endowed with a time-varying state $Z_v(t)$ taking values from a finite set $\mc X$  for every $t\geq0$. 
We shall denote the vector of all agents' states by $\textbf{Z}(t)= (Z_v(t))_{v \in\V }$ and refer to it as the network configuration at time $t$. 
We shall consider {\em asynchronous} and {\em semi-anonymous dynamics} (ASD) where the state of each agent is updated at random activation times by choosing a new state in response to the state of its out-neighbors, according to a conditional probability distribution
that is invariant with respect to permutations of such out-neighbors. 

Formally, let $\mb Z(t)$ be a continuous-time Markov chain with finite state space equal to the set of configurations $\mc Z=\mc X^{\mc V}$ and the structure illustrated below.  

\begin{definition}[Asynchronous semi-anonymous dynamics] \label{def:ASD}
Let $\mc P=\{\theta\in\R_+^{\mc X}:\,\1'\theta=1\}$ be the simplex of probability vectors over $\mc X$.  For every label  $a\in\mc A$, let 
\begin{equation}\label{eq:Theta}\Theta^{(a)}:\Z_+^{\mc A\times\mc X}\to\mc P\end{equation}
be a stochastic kernel, and, for every node $v\in\mc V$, let 
$$\Upsilon^{\mc G}_v:\mc Z\to\Z_+^{\mc A\times\mc X}$$
be defined by 
$$\left(\Upsilon^{\mc G}_v(\mb z)\right)_{ax}=\left|\left\{e\in\mc E:\,\sigma(e)=v,\lambda(\tau(e))=a,z_{\tau(e)}=x\right\}\right|\,,\quad a\in\mc A,\,x\in\mc X\,.$$
Then, $\mb Z(t)$ evolves as a continuous-time Markov chain on $\mc Z$ with transition rates: 
$$\Lambda_{\mb z,\mb z^+}=
\left\{\ba{lcl}
\textcolor{black}{\gamma}\, \Theta_{z^+_v}^{(\lambda(v))}(\Upsilon^{\mc G}_v(\mb z))&\se& \mb z\text{ and }\mb z^+\text{ differ in the }v\text{-th entry only}\\
0&\se&\mb z\text{ and }\mb z^+\text{ differ in more than one entry}\ea\right.$$
where $\gamma$ denotes the Poisson rate at which node $v$ updates its state.  
\end{definition}

The formulation above in Definition \ref{def:ASD} is very general. Some remarks are in order.
\begin{remark}
Classes can describe heterogeneous nodes in a variety of ways. In our examples, we will consider the following three cases: 
i) classes describing different update rules of the nodes;
ii) classes describing nodes with different 
degree distributions;
iii) classes describing node membership to different \lq communities'.
In the most general scenario, a class might represent nodes belonging to a 
specific community, with a given update rule and a particular 
degree distribution.
\end{remark}

\begin{remark}
We emphasize that the new state of an agent, when it gets updated, does not need
to be a deterministic function of its neighborhood. Indeed, we explicitly allow for a stochastic
rule of adopting a certain state.
This allows us to model noisy or mixed-strategy best-response dynamics 
in networked games. 
\end{remark}

Our main interest in this paper is to track the evolution of some macroscopic features, e.g., the evolution of the fraction of nodes belonging to a specific class that are in a given state at time $t$. 
We will demonstrate that a mean-field approximation can yield insight into this analysis for a large class of random networks.

\subsection{Examples of ASD dynamics}\label{Examples_ASD}
To clarify the general formulation introduced above, we provide three examples of ASD dynamics 
that will later be studied in more details in our numerical illustration section (see Section \ref{sec:experiments}).
Since in our examples the node update rule depends only on the total number of neighbors in a given state,
and not on their label, we simplify the general notation introduced before and define:
\begin{equation} \label{eq:xix}
\xi_x(\mb z) = \sum_{a \in \mc A} \left(\Upsilon^{\mc G}_v(\mb z)\right)_{ax} \,,\quad x\in\mc X\,.
\end{equation}
The explicit dependence on the Markov-chain state $\mb z$ will be omitted in the following, whenever possible.
 
\subsubsection{Ternary Linear Threshold Model (TLTM)}\label{ex:LTM}
Let $\mc X=\{-1,0,1\}$ be the set of admissible states and let $\G=(\V,\E,\mathcal{A},\lambda, \sigma,\tau)$ be a labeled multigraph. The label 
${a}_v = (a^+_v,a^-_v)$ of node $v$ determines two given 
(in general, asymmetric) thresholds $a^+_v$, $a^-_v$,
which trigger the transition to state 1 and $-1$, respectively. 
Specifically, when activated, the update of  node $v$  is given by
$$
Z_{v}(t)=\begin{cases}
1&\text{if }\sum_{j\in\mathcal{N}_v^{+}}Z_j(t^{-}) \geq a^+_v\\
0&\text{if }\sum_{j\in\mathcal{N}_v^{+}}Z_j(t^{-}) \in (-a_v^-,a_v^+)\\
-1&\text{if }\sum_{j\in\mathcal{N}_v^{+}}Z_j(t^{-})\leq -a^-_v
\end{cases}
$$
where $Z_v(t^{-}) = \lim_{x\uparrow t}Z_v(x).$
The above rule can be encoded in our general formulation by considering the functions: 
\begin{align*}
\Theta^{(a)}_1(\xi_1,\xi_{-1},\xi_0)&=\boldsymbol{1}_{\{\xi_1-\xi_{-1}\geq a^+\}}\\
\Theta^{(a)}_0(\xi_1,\xi_{-1},\xi_0)&=\boldsymbol{1}_{\{\xi_1-\xi_{-1}\in(-a^-,a^+)\}}\\
\Theta^{(a)}_{-1}(\xi_1,\xi_{-1},\xi_0)&=\boldsymbol{1}_{\{\xi_1-\xi_{-1}\leq- a^-\}}
\end{align*}
which depend only on the numbers $\xi_1$, $\xi_{-1}$ and $\xi_0$ of out-neighbors in state 1, $-1$ and $0$, respectively.

\subsubsection{Binary Response with Coordinating and Anti-coordinating agents (BRCA)}
Inspired by the model in \cite{Ramazi2016NetworksOC}, we consider a network game where each agent can choose between two actions in $\mc X=\{-1,1\}$. 
The network consists of two classes of nodes, i.e. $\mathcal{A}=\{+,-\}$ and $\V=\mathcal{V_+}\cup\V_-$. We assume that ${\V_+}$ and $\V_-$ represent agents 
following the majority (i.e., coordinating) or the minority (i.e., anti-coordinating)
of their out-neighbors, respectively. 

Specifically, an agent is updated according to the following rule: if $v\in\V_+$ then
$$
Z_{v}(t)=\begin{cases}
1&\text{if }\sum_{j\in\mathcal{N}_v^{+}}Z_j(t^{-})> 0\\
-1&\text{if }\sum_{j\in\mathcal{N}_v^{+}}Z_j(t^{-})<0\\
\pm1&\text{if }\sum_{j\in\mathcal{N}_v^{+}}Z_j(t^{-})=0
\end{cases}$$
and if $v\in\V_-$ then
$$ 
Z_{v}(t)=\begin{cases}
1&\text{if }\sum_{j\in\mathcal{N}_v^{+}}Z_j(t^{-})< 0\\
-1&\text{if }\sum_{j\in\mathcal{N}_v^{+}}Z_j(t^{-})>0\\
\pm1&\text{if }\sum_{j\in\mathcal{N}_v^{+}}Z_j(t^{-})=0.
\end{cases}
$$
In essence, when a node is updated, it counts the number of neighbors in state $-1$ and $1$, and 
adopts the state of the majority of its neighbors if $v\in\V_+$, or it adopts the state of the 
minority of its neighbors if $v\in\V_-$. In the case of a tie, it chooses
uniformly at random between states 1 and $-1$. 

The above rule corresponds 
in our general framework to the functions
\begin{align*}
\Theta^{(+)}_1(\xi_1, \xi_{-1})&=\boldsymbol{1}_{\{\xi_1>\xi_{-1} \}} + \frac1{2} \boldsymbol{1}_{\{\xi_1=\xi_{-1} \}} \\
\Theta^{(+)}_{-1}(\xi_1, \xi_{-1} )&= 1 - \Theta_1^{(+)}(\xi_1,\xi_{-1}) \\
\Theta^{(-)}_1(\xi_1, \xi_{-1})&=\boldsymbol{1}_{\{\xi_1< \xi_{-1}\}} + \frac1{2} \boldsymbol{1}_{\{\xi_1=\xi_{-1} \}}\\
\Theta^{(-)}_{-1}(\xi_1, \xi_{-1})&=1 - \Theta_1^{(-)}(\xi_1,\xi_{-1})
\end{align*}
which depend only on the number $\xi_1$ and $\xi_{-1}$ of neighboors in state 1 and $-1$ respectively.

\subsubsection{Evolutionary Roshambo Game (ERG)}
In this example, the best response of an agent follows the 
same rationale of the popular rock-paper-scissors game. 
Specifically, we assume that nodes have three possible states, i.e. $\mc X=\{R, P, S\}$. 
When an agent is updated, it performs the following computation:
\begin{itemize}
\item[i)] for each out-neighbor, it determines its best pairwise response
according to the two-player game 
\begin{center}
\begin{tabular}{c| c c c }
$\varphi$& R&P&S\\
\hline
R& $b$ & 0 & $c $\\ 
P& $c$ & $b$ & 0 \\  
S& 0 & $c$ & $b    $
\end{tabular}
\end{center}
with $c>b$, i.e. R wins over S, S wins over P, P wins over R. 
In Section \ref{sec:experiments} we will consider the case with $b=c/2$;
\item[ii)] the new state of the agent is selected so as to maximize 
the sum of payoffs provided by pairwise interactions: 
$$Z_v(t)=\argmax{\omega\in\mc X}\sum_{j\in\mathcal{N}_v^{+}}\varphi(\omega,Z_j(t^{-})).$$
When the maximizing set is not unique, the new state is selected uniformly at random among the 
maximizing alternatives.
\end{itemize}
In this case we have (ignoring ties for simplicity)
\begin{align*}
\Theta^{(R)}(\xi_R,\xi_{P},\xi_{S})&=\boldsymbol{1}_{\{b\xi_R+c\xi_S> \max\{c\xi_R+b\xi_{P},c\xi_P+b\xi_S\}\}} \\ 
\Theta^{(P)}(\xi_R,\xi_{P},\xi_{S})&=\boldsymbol{1}_{\{c\xi_R+b\xi_{P}> \max\{b\xi_R+c\xi_S,c\xi_P+b\xi_S\}\}} \\
\Theta^{(S)}(\xi_R,\xi_{P},\xi_{S})&=\boldsymbol{1}_{\{c\xi_P+b\xi_S> \max\{b\xi_R+c\xi_S,c\xi_R+b\xi_{P}\}\}}.
\end{align*}

\section{ASD on the labeled branching process}\label{sec:bp}
 
In this section we consider a labeled branching process, i.e. a particular ensemble of infinite labeled directed graphs with rooted tree structure, and then analyze ASD on it. 
As already said, the reason why we introduce this special graph is that the analysis of  ASD on it provides fundamental hints for the analysis
of ASD  on a general locally tree-like ensemble of graphs. 

More precisely, we will consider a labeled branching process completely described by probabilities distributions 
$p_{\textbf{k},a}= p_{\textbf{k}|a} p_a$ and $q_{\textbf{k}|a}^{b}$. 
The first is the joint probability distribution that characterizes the root, i.e. the probability that the root has label $a\in \mathcal{A}$ and out-degree vector $\textbf{k}\in\mathbb{Z}_+^{\mathcal{A}}$. We recall that the component $k_b$ represents 
the number of out-neighbors belonging to class $b\in \mathcal{A}$.
The latter is the vectorial out-degree distribution  for a non-root node with label $a$, whose parent has label $b$.
In next sections we will show that probability distributions $p_{\textbf{k},a}$ and $q_{\textbf{k}|a}^{b}$ specifying the
\lq\lq approximating"  labeled branching process, will be chosen so to exactly  match  statistics' of the  network under investigation. For this reason, we inform the reader that  the same notation will be adopted 
to denote  statistics on both the network and the  associated labelled branching process.

\subsection{Labeled branching process}\label{subsec:tau}
Recall that in our notation $\textbf{k}_v\in\mathbb{Z}_+^{\mathcal{A}}$ denotes the 
the out-degree vector of vertex $v$, whose component $a\in \mathcal{A}$ represents 
the number of out-neighbors of $v$ belonging to class  $a\in \mathcal{A}$.

We will call labeled branching process $\mathcal{T}$ with node 
set $\V = \{v_0, v_1, . . .\}$ and label set $\mathcal{A}$ the rooted tree 
built through the following procedure: 

\begin{itemize}
\item Step $0$: Start with a root node $v_0$ and assign to it a random label $A_0\in\mathcal{A}$ and a random out-degree vector $\textbf{K}^{(0)}\in\mathbb{Z}_+^{\mathcal{A}}$ with joint probability distribution $$\mathbb{P}(A_0=a,\textbf{K}^{(0)}=\textbf{k})=p_{\textbf{k},a}\,.$$  
For every $a\in \mathcal{A}$, add $\textbf{K}^{(0)}_a$  out-edges with label $(A_0,a)$ to the root  $v_0$ and declare all these edges active. Note that  an edge label is defined as  the ordered pair of the labels associated to adjacent nodes.

\end{itemize}
Then, for $h = 1,2,\dots$
\begin{itemize}
\item Step $h$: If there are no active edges, stop. Otherwise, take any active edge $e$, let $(a,b)$ be its label  and declare the edge inactive. 
Assign to  edge $e$ a head node $\tau(e)=v_h$ with label $\lambda(v_h)=b$ and generate a random vector  $\textbf{K}^{(h)}=\textbf{k}$ in $\mathbb{Z}_+^{\mathcal{A}}$ with conditional probability distribution 
$q_{\textbf{k}|b}^a$,  then for every label $c\in\mc A$ add $\textbf{K}^{(h)}_{c}$ new active outgoing edges to $v_h$ with label $(b,c)$.
\end{itemize}

\subsection{Ordinary differential equations of ASD}
Let us now consider the ASD process over the graph $\mathcal{T}$ built above.
In the following matrix notation, vectors are meant to be column vectors, unless otherwise specified.

\begin{proposition}\label{prop:ODE}
Let $\textbf{Z}(t)$, for $t \geq 0$, be the state vector of the ASD on $\mathcal{T}$. 
Then, for every fixed time $t \geq 0$, the following facts hold 
\begin{enumerate}
\item For every $i\in\V$, the states $\{Z_{\tau(e)}(t)|\,e\in\mc E\,:\,\sigma(e)=i\}$ of the offsprings $j$ of $i$  in $\mathcal{T}$ are independent and identically distributed random variables with $\zeta_{\omega|a,b}(t)=\mathbb{P}(Z_j(t)=\omega \mid A_j  =a, A_i = b)$, $\omega\in\mc X$, $a, b \in \mathcal{A}$ satisfying
\begin{equation}\label{eq:zeta}
\frac{\mathrm{d}\zeta_{\omega | a,b}(t)}{\mathrm{d}t} = \gamma \left( \phi_{\omega | a,b}(\boldsymbol{\zeta}(t))-\zeta_{\omega | a,b}(t) \right),
\end{equation}
where 
\begin{gather*}
\phi_{\omega | a,b}(\boldsymbol{\zeta})= \sum_{\textbf{k}\in\mathbb{Z}_+^{\mathcal{A}}} \varphi^{(\textbf{k},a)}_{\omega}(\boldsymbol{\zeta} )q_{\textbf{k}|a}^{b}
\end{gather*}
and 

$$\varphi^{(\textbf{k},a)}_{\omega}(\boldsymbol{\zeta}) = 
\sum_{\substack{\boldsymbol{\xi}\in\Z_+^{\mc A\times\mc X}:\\
\boldsymbol{\xi}\1={\color{black}\textbf{k}}}} \Theta_\omega^{(a)}(\boldsymbol{\xi})  {{\color{black}\textbf{k}} \choose \boldsymbol{\xi} } \prod_{c \in \mathcal{A}}\prod_{g\in\mc X}[\zeta_{g | c,a}]^{\xi_{cg}}.$$

where  
 $$ {{\color{black}\textbf{k}} \choose \boldsymbol{\xi} }= \prod_c  {k_c \choose \boldsymbol{\xi}_c }, $$ 
 $\boldsymbol{\xi}_c$ is the c-th row of matrix $\boldsymbol{\xi}$ and
 $ {\xi_{cg}}$ denotes the $(c,g)$-th  element  of  matrix $\boldsymbol{\xi}$.
\item The state $Z_{v_0}(t)$ of the root node $v_0$ is a random variable with $y_{\omega|a}(t)=\mathbb{P}(Z_{v_0}=\omega \mid A_0  =a)$ satisfying
\begin{equation}\label{eq:zeta2}
\frac{\mathrm{d}y_{\omega|a}(t)}{\mathrm{d}t}=\gamma \left(\psi_{\omega|a}(\boldsymbol{\zeta}(t))-y_{\omega|a}(t) \right),
\end{equation}
with \begin{gather*}
\psi_{\omega|a}(\boldsymbol{\zeta})=  \sum_{\textbf{k}\in \mathbb{Z}_+^{\mathcal{A}}} \varphi^{(\textbf{k},a)}_{\omega}(\boldsymbol{\zeta})p_{\textbf{k}\mid a} 
\end{gather*}
\end{enumerate}
\end{proposition}

\begin{proof}
\begin{enumerate}
\item Let $v_0$ be the root of $\mathcal{T}$. Then, For every $i \in\V$, the states ${Z_j(t)}:(i,j)\in\E$ of the offsprings of $v_i$ in $\mathcal{T}$ are independent and identically distributed Bernoulli random variables. 
Define $\zeta_{\omega|a,b}(t)=\mathbb{P}[Z_j(t)=\omega \mid A_j  =a, A_i = b]$, $j\in\V\setminus\{v_0\}$, where $v_i$ is the father of $v_j$, we have
\begin{align*}
&\zeta_{\omega | a,b}(t+\Delta t)\\
&\ =\e^{-\gamma\Delta t}\zeta_{\omega| a,b}(t)+\left(1-\e^{-\gamma\Delta t}\right)\times\\
&\ \times \sum_{\textbf{k}\in {\mathbb{Z}_+^{\mathcal{A}}}} \mathbb{P}\left[Z_j(t)=\omega|K_j=\textbf{k},A_j=a, A_i = b\right]\mathbb{P}\left[K_j=\textbf{k}|A_j=a, A_i = b\right]
\\
&\ =\left({\gamma\Delta t}+o(\Delta t)\right)\sum_{\textbf{k}\in\mathbb{Z}_+^{{\mathcal{A}}}} \varphi_{\omega}^{(\textbf{k}, a)}(\boldsymbol{\zeta}(t))q_{\textbf{k}|a}^{b} +\left(1-{\gamma\Delta t}+o(\Delta t)\right)\zeta_{\omega |a,b }(t)+o(\Delta t)\\
&\ =\left({\gamma \Delta t}+o(\Delta t)\right)\phi_{\omega|a,b}(\boldsymbol{\zeta}(t))+\left(1-{\gamma \Delta t}+o(\Delta t)\right)\zeta_{\omega|a,b}(t),
\end{align*}
from which we conclude
\begin{align*}
\frac{\mathrm{d}\zeta_{\omega|a,b}(t)}{\mathrm{d}t}=\lim_{\Delta t\rightarrow0}\frac{\zeta_{\omega|a,b}(t+\Delta t)-\zeta_{\omega|a,b}(t)}{{\Delta t}}&=\gamma(\phi_{\omega|a,b}(\boldsymbol{\zeta}(t))-\zeta_{\omega|a,b}(t)).
\end{align*}

\item Define $y(t)=\mathbb{P}[Z_{v_0}(t)=\omega\mid A_0  =a]$, then with the same arguments we have
$$
\frac{\mathrm{d}y_{\omega|a}}{\mathrm{d}t}=\gamma(\psi_{\omega|a}(\boldsymbol{\zeta}(t))-y_{\omega|a}(t)).
$$
with $$\psi_{\omega|a}(\boldsymbol{\zeta}(t))=  \sum_{\textbf{k}\in\mathbb{Z}_+^{\mathcal{A}}} \varphi^{(\textbf{k},a)}_{\omega}(\boldsymbol{\zeta}(t))p_{\textbf{k}\mid a}.$$
\end{enumerate}

\end{proof}

\begin{remark}
Whenever labels of neighbor nodes are independent, $p_{\textbf{k}\mid a}$ and 
$q_{\textbf{k}\mid a}^{b}$ depend on $a,b$ only through $k = \sum_i k_i$, and things become simpler. Indeed, if we define $\zeta_{\omega}(t)=\mathbb{P}[Z_j(t)=\omega ] = \sum_{a,b \in \mathcal{A}} \zeta_{\omega|a,b}(t) p_a p_{b}$, then, it can be easily shown that, similarly to \eqref{eq:zeta}, we can derive an ODE for $ \zeta_{\omega}(t)$ in the form
\begin{equation}\label{eq:zeta_bis}
\frac{\mathrm{d}\zeta_{\omega }(t)}{\mathrm{d}t} = \gamma \left(\phi_{\omega}(\boldsymbol{\zeta}(t))-\zeta_{\omega }(t) \right),
\end{equation}
where
\begin{gather*}
\phi_{\omega}(\boldsymbol{\zeta}(t))=\sum_{a,b\in\mathcal{A}}\sum_{k\in\mathbb{Z}_+} \varphi^{(k,a)}_{\omega}(\boldsymbol{\zeta}(t))q_{k\mid a,b} p_a p_{b},\\
q_{k\mid a,b}=\sum_{\substack{\textbf{k}\in  \mathbb{Z}_+^{\mathcal{A}}:\\ {\color{black}\textbf{k}}^T\1=k}} q_{\textbf{k}|a}^{b},
\end{gather*}
and $$\varphi^{(k,a)}_{\omega}(\boldsymbol{\zeta}(t))=\sum_{\substack{\boldsymbol{\xi} \in  \mathbb{Z}_+^{\mathcal{X}} :\\ 
\boldsymbol{\xi}^T\1=k}}
\Theta_\omega^{(a)}(\boldsymbol{\xi}){k\choose \boldsymbol{\xi}}\prod_{g\in\mc X}[\zeta_{g}(t)]^{\xi_{g}}.$$

Analogously, defining $y_{\omega}(t)=\mathbb{P}[Z_0(t)=\omega ] = \sum_{a \in \mathcal{A}} y_{\omega|a}(t) p_a$, the ODE replacing \eqref{eq:zeta2} can be written as
\begin{equation}\label{eq:zeta2_bis}
\frac{\mathrm{d}y_{\omega}(t)}{\mathrm{d}t}=\gamma \left(\psi_{\omega}(\boldsymbol{\zeta}(t))-y_{\omega}(t) \right),
\end{equation}
where
\begin{gather*}
\psi_{\omega}(\boldsymbol{\zeta}(t))=\sum_{a\in\mathcal{A}} \sum_{k\in\mathbb{Z}_+}\varphi^{(k,a)}_{\omega}(\boldsymbol{\zeta}(t))p_{k|a}p_a,\\ 
p_{k\mid a}=\sum_{\substack{\textbf{k}\in  \mathbb{Z}_+^{\mathcal{A}}:\\ {\color{black}\textbf{k}}^T\1=k}}p_{\textbf{k}|a}.
\end{gather*}
\end{remark}

\begin{remark}
The above  analysis of  ASD over the ensemble $\mathcal{T}$ can be easily extended to the case in which the rate of activation of a vertex depends on  its label   ${ a} \in\mathcal{A}$. Without loss of generality we will assume $\gamma=1$ in the following.

\end{remark}
\section{ASD on labeled random networks}\label{sec:rn}
In this section we consider the evolution of ASD process over a multigraph $\mathcal{G}$ taken from a general ensemble of labeled directed graphs $\mathfrak{E}^{(n)}$ of size $n$. In particular, we show that, under certain conditions on the ensemble and on the initial node configuration, the ASD process over    
$\mathcal{G}$ can be well approximated by the same process over a labeled 
branching process $\mathcal{T}$.
The ensemble $\mathfrak{E}^{(n)}$ is described by the \lq node statistics' 
$p_{\textbf{d},\textbf{k},a,s}$, which provides the probability that a node picked at random has in-degree vector $\textbf{d}$, out-degree vector $\textbf{k}$, 
label  $a\in\mathcal{A}$ and initial state $s \in \mc X$. We shall assume that $p_{\textbf{d},\textbf{k},a,s}$ factorizes as
 $$p_{\textbf{d},\textbf{k},a,s}=p_{\textbf{d},\textbf{k},a}p_{s|a}.$$
The above node statistics clearly provides all information
needed to compute any marginal or conditional distribution 
we might be interested in. For example, 
$p_{\textbf{d},\textbf{k},a} = \sum_{s \in \mc X} p_{\textbf{d},\textbf{k},a,s}$
is the distribution of in-degree vector, out-degree vector and label
of a generic node. As another example, $p_{\textbf{k},a} = \sum_{\textbf{d}}
p_{\textbf{d},\textbf{k},a}$ provides the distribution of out-degree vector 
and label of a generic node. We denote $p_a = \sum_{\textbf{k}}
p_{\textbf{k},a}$ the probability for a  node to be associated with label $a$.  
With  intuitive notation, $p_{\textbf{k}|a} = p_{\textbf{k},a}/p_a$ denotes the
distribution of out-degree vector of a node with label $a$, and so on.  

As it always happens in graphs with heterogeneous degrees, we will need to distinguish the probability law of $\textbf{k}_v$ for a generic node $v$ picked uniformly at random, and the probability law of $\textbf{k}_v$ for a node $v$ reached by traversing an edge. This because, in general, we could have correlation between in-degree and out-degree. Moreover, when we reach a node by following a certain edge, it is also important to distinguish the label of the node originating the traversed edge  picked uniformly at random. To account for the above generality, we need to introduce some additional notation.
Specifically, we define:
\begin{equation}\label{eq:qnew}\
q^{a}_{{\textbf{d}},{\textbf{k}}|b}= \frac{d_a p_{{\textbf{d}},{\textbf{k}},b}}{ \sum_{{\textbf{d}},{\textbf{k}}}d_a p_{{\textbf{d}},{\textbf{k}},b}},
\end{equation}
which is the distribution of in-degree vector ${\textbf{d}}$ and out-degree vector ${\textbf{k}}$
of a node with label $b$, reached by traversing an edge from a 
node with label $a$. 
Similarly, $q^{a}_{{\textbf{k}}|b} = \sum_{\textbf{d}} q^{a}_{{\textbf{d}},{\textbf{k}}|b}$
is the marginal distribution of out-degree vector 
of a node with label $b$, reached by traversing an edge from a 
node with label $a$.

\subsection{Relevant neighborhood at time $t$}
We first  observe that, since the process evolves through local interactions,
the state of a generic node $v$ on a multigraph $\mathcal{G}=(\mathcal{V}, \mathcal{E},\mc A,\lambda,\sigma,\tau)$ at time $t$
is determined  only by the structure and state of a relatively small neighborhood around $v$. 
Given a generic node $v\in \mathcal{V}$, we define the {\em relevant neighborhood} $\mathcal{N}_{t}$ 
of $v$ as the subgraph induced by the  set of all nodes in $\mathcal{V}$ having an impact on $Z_v(t)$, i.e., 
on the state of $v$ at time $t$. 
Similarly, we  define the {\em relevant neighborhood} $\mathcal{T}_t$ 
as the subtree induced by the set of all nodes in $\mathcal{T}$ having an impact on $Z_{v_0}(t)$, where $v_0$ is the root node
of $\mathcal{T}$.

The relevant neighborhood can be built by looking backward in time, identifying dependencies 
between neighboring nodes. First, observe that the state of $v$ at time $t$ depends on its out-neighbors $v'$ (one-hop away nodes) 
if and only if $v$ has updated its state in $[0,t]$ at least once, i.e., we can find an update time of $v$, 
$\vartheta_v(t)\le t$ . The state of node $v$ depends on a two-hop away node 
$v''$,  if and only if we can find a common neighbor $v'$ of $v$ and $v''$, 
such that $\vartheta_{v'}(t)< \vartheta_{v}(t)\le t$.  
Similarly the state of $v$ depends on a three-hops away node $v'''$ only if we can find 
two nodes $v'$ and $v''$ along a directed path from $v$ to $v'''$ such that 
$\vartheta_{v''}(t) <\vartheta_{v'}(t)< \vartheta_v(t)\le t$, and so on.

Due to the fact that update times of each node form independent Poisson processes with rate $\gamma=1$, 
we can exploit well-known properties of the Poisson process (time-reversibility, memoryless
property) to obtain $\mathcal{N}_{t}$ (or $\mathcal{T}_t$)
as the result of a process evolving forward in time, and exploring progressively the neighborhood 
of $v$ by adding an exponentially distributed delay (of mean 1) on each explored node, up to time $t$.
 
More precisely, the relevant neighborhood of $v$ is obtained by the following process. 
Vertices can be active, neutral or inactive. Initially, the relevant neighborhood is  empty.
\begin{enumerate}
\item The process starts by activating node $v$ at time $t=0$. All of the other nodes are set neutral.
\item Upon the activation of a node, a random timer is associated to it, taken from an exponential distribution
of mean $\gamma=1$. Moreover, the node is added to the relevant neighborhood, together with the outgoing edges. 
\item Upon expiration of its associated timer: i) an active node is set inactive; ii) all of its
neutral out-neighbors are set active and added to the relevant neighborhood, together with outgoing edges. 
\end{enumerate} 

For $t\geq0$, we can stop the above exploration process at time $t$ (i.e., we no longer add
nodes to the relevant neighborhood after time $t$), obtaining a truncated version $\mathcal{N}_{t}$
of $\mathcal{G}$, composed of all the nodes that have been activated. Similarly, we obtain a truncated version $\mathcal{T}_t$ 
of $\mathcal{T}$. 

\subsection{Approximation result}

Let $\mathcal{G}=(\V,\E,\mathcal A,\lambda,\sigma,\tau)$ be a multigraph \textcolor{black}{sampled from a given labeled network ensemble $\mathfrak{E}^{(n)}$} of size $n$. 
For $t\geq0$, let $\mathcal{N}_{t}$ be the relevant neighborhood at time $t$ of a node $v$ chosen uniformly at random from $\mathcal{V}$, and let $\mu_{\mathcal{N}_{t}}$ be its distribution on the multigraph space. Let $\mathcal{T}_t$ be a labeled branching process
as defined in Sec. \ref{subsec:tau}, truncated at time $t$, and let 
$\mu_{\mathcal{T}_t}$ be its distribution.
 
Proposition \ref{thm:ASD_over_gen_net} identifies some sufficient conditions to guarantee that the ASD process over a network is well approximated by the solution of the differential equation in \eqref{eq:zeta2}.

\begin{proposition}\label{thm:ASD_over_gen_net}

For $t\geq0$, let $Z(t)$ be the state vector of the ASD at time $t$ on $\mathcal{G}$. Let $z_{\omega}(t) =\frac{1}{n}|\{v\in\V:Z_v(t)=\omega\}|$ be the 
fraction of state-$\omega$ adopters at time $t$, and $\overline{z}_{\omega}(t) = \mathbb{E}[z_{\omega}(t)]$ be 
its expectation over the ensemble. For any $
\epsilon>0$, 
$$\mathbb{P}(|z_{\omega}(t)-y_{\omega}(t)|\geq \epsilon)\leq \mathbb{P}(|z_{\omega}(t)-\overline{z}_{\omega}(t)|\geq \epsilon-\|\mu_{\mathcal{N}_{t}}-
\mu_{\mathcal{T}_t}\|_{\mathrm{TV}} )
$$
where $y_{\omega}(t)$ is the solution of \eqref{eq:zeta2}.
\end{proposition}

\begin{proof}Notice that
\begin{align}
&\mathbb{P}(|z_{\omega}(t)-y_{\omega}(t)|\geq \epsilon)\leq \mathbb{P}(|z_{\omega}(t)-\overline{z}_{\omega}(t)|+|\overline{z}_{\omega}(t)-y_{\omega}(t)|\geq \epsilon) \label{eq:approximation}
\end{align}
We prove now that 
$|\overline{z}_{\omega}(t)-y_{\omega}(t)|\leq \|\mu_{\mathcal{N}_{t}}-\mu_{\mathcal{T}_t}\|_{\mathrm{TV}}$ from which we get the result.

Observe that, by definition, the state $Z_v(t)$ of node $v$  depends exclusively on the initial states $Z_j(0)=\sigma_j$ of the agents belonging to the relevant neighborhood $\mathcal{N}_{t}$ of node $v$ at time $t$, i.e., $\mathbb{P}(Z_v(t)=\omega)=\chi_{\omega}(\mathcal{N}_{t})$ 
where $\chi_{\omega}$ is a function in the range [0,1].

We thus have
\begin{align*}
\overline{z}_{\omega}(t) &= \mathbb{E}[z_{\omega}(t)]=\frac{1}{n}\sum_{v\in\V}\mathbb{P}(Z_v(t)=\omega)=\int\chi_{\omega}(g)\mathrm{d}\mu_{\mathcal{N}_{t}}(g).
\end{align*}
On the other hand, considering the state of the root in the labeled branching process $\mathcal{T}$,  
the output of the ODE \eqref{eq:zeta2} satisfies
\begin{align*}
y_{\omega}(t)=\int\chi_{\omega}(g)\mathrm{d}\mu_{\mathcal{T}_t}(g).
\end{align*}
It then follows
\begin{align*}
|\overline{z}_{\omega}(t)-y_{\omega}(t)|
&\leq\left|\int\chi_{\omega}(g)\mathrm{d}\mu_{\mathcal{N}_{t}}(g)-\int\chi_{\omega}(g)\mathrm{d}\mu_{\mathcal{T}_t}(g)\right|\\
&\leq\left|\int\left(\chi_{\omega}(g)-\frac{1}{2}\right)\mathrm{d}\mu_{\mathcal{N}_{t}}(g)-\int\left(\chi_{\omega}(g)-\frac{1}{2}\right)\mathrm{d}
\mu_{\mathcal{T}_t}(g)\right|\\
&\leq\|\mu_{\mathcal{N}_{t}}-\mu_{\mathcal{T}_t}\|_{\text{TV}}.
\end{align*}

\end{proof}

From Proposition \ref{thm:ASD_over_gen_net} we deduce that the evolution of the ASD process is well approximated by the solution of the differential equation in \eqref{eq:zeta2} for graph ensembles enjoying the following two fundamental properties:
\begin{itemize}
\item[(a)] Topological Property: Local Weak Convergence is required, in the sense that  $\|\mu_{\mathcal{N}_{t}}-\mu_{\mathcal{T}_t}\|_{\text{TV}}$ can be 
made arbitrarily small by increasing the graph size.
\item[(b)] Concentration Property: for large graph size, the fraction of state-$\omega$ adopters in the ASD 
process must concentrate around its expectation with probability close to one.
\end{itemize}

\section{ASD over labeled configuration model}\label{sec:labconfmod}
Considering the ensable $\mathfrak{E}^{(n)}$ of all labeled networks with given size $n$ and 
statistics $p_{\textbf{d},\textbf{k},a,s}$, we define the corresponding 
{\em labeled configuration model} ensamble $\mathfrak{C}_{n,p}$, on which we will restrict
our investigation in the rest of the paper. 
In particular, we provide general bounds for ASD evolution over $\mathfrak{C}_{n,p}$.

Next we consider two specific examples of labeled configuration model,
which we believe are particularly interesting, and apply to them the 
general bounds above, showing asymptotic convergence to the ODE solution
as the network size grows large.

\subsection{Labeled configuration model}
We first explicitly describe the construnction of the labeled configuration model $\mathfrak{C}_{n,p}$.
For each $v\in\mathcal{V}$, denote with $\boldsymbol{\kappa}_v=(\kappa_v^a)_{a\in\mathcal{A}}$ and $\boldsymbol{\delta}_v=(\delta_v^a)_{a\in\mathcal{A}}$ the out-degree and in-degree vectors, respectively, such that there is exactly a fraction $p_{{\textbf{d}},{\textbf{k}},a}$ of nodes $v\in\mathcal{V}$ with $(\boldsymbol{\delta}_v,\boldsymbol{\kappa}_v,a_v)=({\textbf{d}},{\textbf{k}},a)$. Denote with $\mathcal{L}_{a,a'}$ a set of stubs, and define arbitrary maps $\nu_{a,a'},\gamma_{a,a'}:\mathcal{L}_{a,a'}\rightarrow\mathcal{V}$, satisfying the property: $|\nu_{a,a'}^{-1}(v)| =\delta_{v}^{a}$ for nodes $v$ with label $\lambda(v)=a'$ and $|\gamma_{a,a'}^{-1}(v)|=\kappa_{v}^{a'}$ with $\lambda(v)=a$. For all $a, a' \in \mathcal{A}$, let $\pi_{a,a'}$ be chosen uniformly at random among all permutations of $\mathcal{L}_{a,a'}$ and define multigraph $\mathcal{G} = (\mathcal{V},\mathcal{E},\mathcal{A}, \lambda,  \sigma, \tau)$ with set of nodes $\mathcal{V}$ and $\mathcal{E} =\bigcup_{(a,a')\in\mathcal{A}\times\mathcal{A}} \mathcal{E}_{a,a'}$, where $\mathcal{E}_{a,a'}={(\gamma_{a,a'}(h), \nu_{a,a'}(\pi_{a,a'}(h)):  h\in \mathcal{L}_{a,a'}}$, and $\sigma (\gamma_{a,a'}(h), \nu_{a,a'}(\pi_{a,a'}(h)))=  \gamma_{a,a'}(h)$ and $\tau (\gamma_{a,a'}(h), \nu_{a,a'}(\pi_{a,a'}(h)))=  \nu_{a,a'}(\pi_{a,a'}(h))$.

Denote with $l_{a,a'} = |\mathcal{L}_{a,a'}|$ the total number of edges incoming to nodes with label $a'$, originating from nodes with label $a$, so that:  
$$
l_{a,a'}= n\sum_{{\textbf{d}},{\textbf{k}}}d_ap_{{\textbf{d}},{\textbf{k}},a'} = n\sum_{{\textbf{d}},{\textbf{k}}}k_{a'}p_{{\textbf{d}},{\textbf{k}},a}
$$
The total number of edges in the graph is $l = \sum_{a,a'} l_{a,a'}$.
The average in-degree of a node, which is equal to the average out-degree, 
will be denoted by $\bar{d} = l/n$.

We repeat here for readers' ease the expression of the fraction of nodes with label 
$a'$, reached from a node with label $a$, having 
in-degree vector ${\textbf{d}}$ and out-degree vector ${\textbf{k}}$: 
\begin{equation}\label{eq:statistics_edges}
q^{a}_{{\textbf{d}},{\textbf{k}}|a'}=\frac{d_a p_{{\textbf{d}},{\textbf{k}},a'} }{\sum_{{\textbf{d}},{\textbf{k}}} d_a p_{{\textbf{d}},{\textbf{k}},a'}  }
\end{equation}
We summarize the notations in Table \ref{tab:notations}. In Figure \ref{fig:Notation2}, we show an example of notation use for a simple case in which $\mathcal{A}=\{+,-\}$.

\begin{center}
\begin{table}[h!]
\boxed{\begin{tabular}{ l  }
$\mathcal{A}$ ordered set of labels\\
$ a_v = \lambda(v) $ label of node $v$  \\ 
$\mathcal{V}_a=\{v\in\V:a_v=a\}$\\
$ \boldsymbol{\delta}_v=(\delta_v^a)_{a\in\mathcal{A}}$ in-degree of node $v$  \\ 
$ \boldsymbol{\kappa}_v=(\kappa_v^a)_{a\in\mathcal{A}}$ out-degree of node $v$  \\ 
$\mathcal{L}_{a,a'}$ set of stubs from nodes with label $a$ 
to nodes with label $a'$  \\
$l_{a,a'} = |\mathcal{L}_{a,a'}|$ number of edges from nodes with label $a$ 
to nodes with label $a'$\\
$\pi_{a,a'}$ permutation of $\mathcal{L}_{a,a'}$\\
$\nu_{a,a'}:\ \mathcal{L}_{a,a'}\rightarrow \mathcal{V}$ map with the property $|\nu^{-1}_{a,a'}(v)|=\delta_v^a$ for all $v\in\mathcal{V}_{a'}$ \\
$\gamma_{a,a'}:\ \mathcal{L}_{a,a'}\rightarrow \mathcal{V}$  map with the property $|\gamma^{-1}_{a,a'}(v)|=\kappa_v^{a'}$ for all $v\in\mathcal{V}_a$ \\
$\mathcal{E}_{a,a'}={(\gamma_{a,a'}(h), \nu_{a,a'}(\pi_{a,a'}(h)):  h\in \mathcal{L}_{a,a'}}$\} set of edges from nodes\\$\qquad$ with label $a$ to nodes with label $a'$\\
$\mathcal{E}=\bigcup_{a,a'}\mathcal{E}_{a,a'}$
\end{tabular}}\caption{Notations}\label{tab:notations}
\end{table}
\end{center}

\begin{figure}[h!]
\begin{center}\begin{tikzpicture} [every node/.style={circle,fill=gray!20,inner sep=4pt}]
  \node (n1) at (3,8) {$+$};
  \node (n2) at (4,8)   {$+$};
  \node (n3) at (5,8) {$+$};
   \node (n4) at (6,8){$+$};
  \node (n5) at (7,8)   {$+$};
     \node (n6) at (9,8) {$-$};
  \node (n7) at (10,8)   {$-$};
  \node (n8) at (11,8) {$-$};

  \node (n9) at (6.5,5) {$v$};
  
    \node (n11) at (3,2) {$+$};
  \node (n12) at (4,2)   {$+$};
  \node (n13) at (5,2) {$+$};
   \node (n14) at (7,2){$-$};
  \node (n15) at (8,2)   {$-$};
     \node (n16) at (9,2) {$-$};
  \node (n17) at (10,2)   {$-$};
  \node (n18) at (11,2) {$-$};

 \tikzset{every node/.style={}} 
  \node (n10) at (10,5)   {$\begin{array}{c}a_v=+\\
  \boldsymbol{\kappa}_v=(\kappa_v^{+},\kappa_v^{-})=(3,5)\end{array}$};
  \node (n10) at (4,5)   {$\boldsymbol{\delta}_v=(\delta_v^{+},\delta_v^{-})=(5,3)$};
\foreach \from/\to in {n1/n9,n2/n9,n3/n9,n4/n9,n5/n9}
\path (\from) edge[->,bend right=3] (\to);
\foreach \from/\to in {n6/n9,n7/n9,n8/n9}
\path (\from) edge[->,bend right=3] (\to);
\tikzset{mystyle/.style={->}} 
\tikzset{every node/.style={fill=cyan!10}} 

\foreach \from/\to in {n9/n11,n9/n12,n9/n13,n9/n14,n9/n15,n9/n16,n9/n17,n9/n18}
\path (\from) edge[->,bend right=3] (\to);
\tikzset{mystyle/.style={->}} 
\tikzset{every node/.style={fill=cyan!10}} 
\tikzset{every node/.style={fill=cyan!10}} 
      \path (n3) edge[->,bend right=3] node { $\ \quad\nu_{+,+}^{-1}(v)\quad\ $}(n9);
      
      \tikzset{every node/.style={fill=cyan!10}} 
      \path (n7) edge[->,bend right=3] node { $\nu_{-,+}^{-1}(v) $}(n9);

      \tikzset{every node/.style={fill=cyan!10}} 
      \path (n9) edge[->,bend right=3] node { $\gamma_{+,+}^{-1}(v) $}(n12);
          \tikzset{every node/.style={fill=cyan!10}} 
  \path (n9) edge[->,bend right=3] node { $\ \ \  \gamma_{+,-}^{-1}(v)\  \quad$}(n16);

\end{tikzpicture}
\end{center}\caption{Labeled configuration model with two classes $\mathcal{A}=\{+,-\}$}\label{fig:Notation2}

\end{figure}
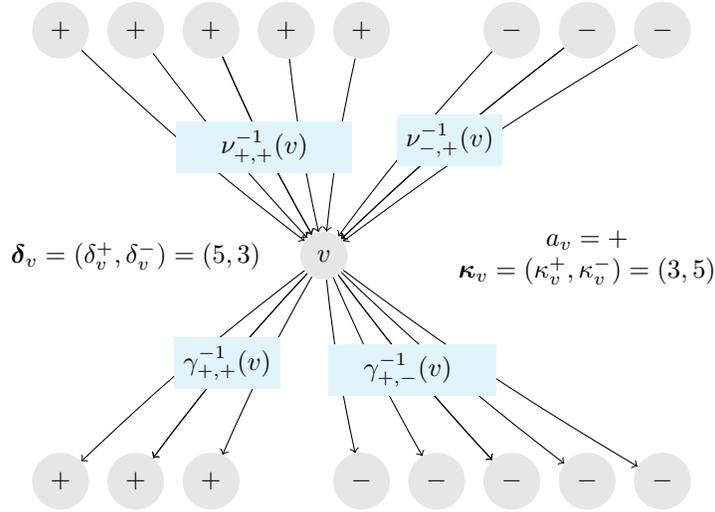
We emphasize that the graph ensemble defined above  extends the classical configuration model, which can be recovered 
as a particular case by setting  $|\mathcal{A}|=1.$

\subsection{Bounds to ASD dynamics}\label{Sec:generalbounds}
In order to apply the general approximation result stated in Proposition
\ref{thm:ASD_over_gen_net} to the labeled configuration model defined above,
we need to prove both Local Weak Convergence and Concentration Property
of ASD.
The following theorems actually provide the main results of our paper:   
 
\begin{itemize}
\item Theorem \ref{thm:top_res_modgraphs} is a topological result and is related exclusively to the properties of the labeled configuration model. More precisely, it provides a useful bound on  the total variation distance between the relevant neighborhood of a graph drawn uniformly at random from the labeled configuration model ensemble and the labeled branching process described in Section \ref{sec:bp}. 
The proof is rather technical and is postponed to Appendix \ref{app:Topological_result}.
\item Theorem \ref{thm:concentration_modgraphs} is related to the ASD evolution and 
to the specific properties of the labeled configuration model. 
It provides a bound on the distance between the fraction of state-$\omega$ adopters and 
its expectation. The proof can be found in Appendix \ref{appA_concentration}.
\end{itemize}

Let $\mathcal{G}=(\V,\E,\mathcal{A}, \lambda,\sigma, \tau)$ be a multigraph sampled from the ensemble $\mathfrak{C}_{n,p}$. For $t\geq 0$, let $\mathcal{N}_{t}$ be the relevant neighborhood at time $t$ of a node $v$ chosen uniformly at random from $\mathcal{V}$, and let $\mu_{\mathcal{N}_{t}}$ be its distribution. 
Let $\mathcal{T}_{t}$ be the labeled branching process truncated at time $t$  and let 
$\mu_{\mathcal{T}_t}$ be its distribution. Moreover, let  $W^{b,a}_t$ be the number of edges in $\mathcal{T}_t$ from nodes with label $b\in\mathcal{A}$ to nodes with label $a\in\mathcal{A}$ and let $\mathsf{F}_{W^{b,a}_t}(x_{b,a}) =\mathbb{P}(W^{b,a}_t>x_{b,a})$.

 \begin{theorem}[Topological Property]\label{thm:top_res_modgraphs} 
We have
\begin{align}\label{eq:bound_topologico}
&\|\mu_{\mathcal{T}_{t}}-\mu_{\mathcal{N}_{t}}\|_{\mathrm{TV}} \le \mathbb{P}(\mathcal{T}_{t}\neq \mathcal{N}_{t})\nonumber\\
&\quad\leq \inf_{\textbf{X} \in(\R^+)^{\mathcal{A}\times\mathcal{A}}}\sum_{a,b\in\mathcal{A}}\left[ \mathsf{F}_{W^{b,a}_t}(x_{b,a}) 
+
\frac{x_{b,a}(x_{b,a}+1)}{2} \frac{\sum_{{\textbf{d}},{\textbf{k}}}d_bq^{b}_{{\textbf{d}},{\textbf{k}}|a}}{l_{b,a}}  \right.\\
&\left.+
\sum_{b'\neq b} x_{b,a}x_{b',a} \frac{\sum_{{\textbf{d}},{\textbf{k}}}d_bq^{b'}_{{\textbf{d}},{\textbf{k}}|a} }{l_{b,a}}\right] \end{align}
 \end{theorem}

\begin{example}[Topological Property for the classical configuration model]\label{es:clssical_cm}
If $|\mathcal{A}|=1$ then the model ensemble $\mathfrak{C}_{n,p}$ boils down to the classical configuration model and the bound derived in \eqref{eq:bound_topologico} reduces to
 \begin{align}\label{Top_res_conf_model}
 \|\mu_{\mathcal{T}_{t}}-\mu_{\mathcal{N}_{t}}\|_{\mathrm{TV}}&\leq  \mathbb{P}(\mathcal{T}_{t}\neq \mathcal{N}_{t})\leq
  \inf_{x>0} \left[\mathsf{F}_{\widetilde{W}_t}(x) + \frac{{\sum_{d,k}dq_{d,k}}x(x+1)}{2n\overline{d}} \right]
 \end{align}
where $\widetilde{W}_t$ is the number of nodes in $\mathcal{T}_t$ and $\mathsf{F}_{\widetilde{W}_t}(x)=\mathbb{P}(\widetilde{W}_t> x)$.

\end{example}

We next introduce the concentration result that allows us to 
estimate to what extent the fraction  of state-$\omega$ adopters in the ASD process concentrates around its expectation. 

\begin{theorem}[Concentration Property]\label{thm:concentration_modgraphs}
	Let $\mathcal{G}$ be a multigraph sampled from the ensemble  $\mathfrak{C}_{n,p}$.
	  We denote with $\mathcal{N}^{v}_t$ the relevant neighborhood at time $t$ of a node $v$, sampled with a probability proportional to its in-degree, and with $V_t^{v}$ the number of nodes in it.
	For $t \geq 0$, let $Z(t)$ be the state vector of the ASD dynamics on $\mathcal{G}$, $b(t) =|\{v\in\V:Z_v(t)=\omega\}|$ be the number of state-$\omega$ adopters at time $t$ conditioned to $\mathcal{G}$.  
	For any $\epsilon>0,\eta>0,x>0,s\geq 1$ we have 
	{\color{black}{\begin{align*}
&	\mathbb{P}(|b(t)-\mathbb{E}[{b}(t)]|> \eta n)\\
	&\qquad\leq
4\e^{-\frac{\eta^2 n}{1152(1+\epsilon)tx^2}} +\left(1+\frac{12}{\eta}\right)(1+\epsilon){\color{black}{tn}}\frac{2^s\mathbb{E}_v\left[|V^{v}_t|^s\right]}{x^s}+2\e^{-\frac{nt\epsilon^2}{2(1+\epsilon)}}+2\e^{-\frac{\eta^2 n}{288(t+\eta/12)}}\\
&\qquad+\left(1+\frac{4}{\eta}\right)\frac{{{2^s}}}{x^s}\sum_{w\in\mathcal{V}}|\boldsymbol{\delta}_w|{\left[\mathbb{E}[|V_t^{w}|^s\right]}
+2\e^{-\frac{\eta^2 n}{128\overline{d}x^2}}
	\end{align*}}}
\end{theorem}

\begin{remark} 
We emphasize that the bounds presented in Theorem \ref{thm:top_res_modgraphs} and  Theorem \ref{thm:concentration_modgraphs}
represent an important step forward with respect to 
results already known in literature.
In particular, Lemma 5 and Proposition 2 in \cite{RCF18} states a similar result 
for a different, simplified version of our system dynamics in which:
i) node updates are synchronized, being triggered by a common discrete time step;
ii) the update rule is restricted to be the deterministic linear threshold model;
iii) both the maximum in-degree $d_{\max}$ and the maximum out-degree $k_{\max}$
of nodes are supposed to be finite.
In contrast to \cite{RCF18}, we introduce a much more general result 
along three different directions:
\begin{enumerate}
\item We consider asynchronous dynamics (each node is updated by an independent Poisson clock). 
Hence the neighborhood exploration process in the proof has to take into account this new source of randomness. Specifically, the estimation of the total variation \eqref{eq:bound_topologico} 
is split into two terms, which are obtained by conditioning on the number of nodes in $\mathcal{T}_t$. The necessity of this refined analysis will be clear in the next section.
\item We consider arbitrary semi-anonymous dynamics, with possible random (noisy) response to the state of neighbors. Moreover, we define our dynamics on a much more general ensemble of labelled 
random graphs, which allows us to differentiate the distribution of incoming/outgoing edges
for each pair of classes.  
\item We allow the maximum in- and out-degree of nodes to possibly scale 
with $n$ (under some technical constraints). This is crucial for applications to social networks
and many other complex systems in which the degree distribution has often been 
observed to follow a power law. But notice that even in the case of the classic 
Erd\"os-R\'enyi random graph $G(n,p)$, $d_{\max}$ or $k_{\max}$ of course  
are not  independent of $n$.
Note that, in the case of finite $d_{\max}, k_{\max}$, by taking $x=k_{\max}^t$ our bound in \eqref{Top_res_conf_model} leads to 
$$
\|\mu_{\mathcal{N}_t}-\mu_{\mathcal{T}_t}\|_{\text{TV}}\leq \frac{d_{\max}k_{\max}(k_{\max}^t+1)}{ 2n\overline{d}}.
$$
recovering the result in \cite{RCF18} (Lemma 5).
\end{enumerate}
\end{remark}

\begin{remark}
We emphasize that our results could be also extended to weighted directed networks, 
under the assumption that weights on edges are described by i.i.d. random variables. 
For example, one could consider the relative weighted popularity of a given state among the neighbors, to determine the next state of an agent. It is obvious that in this case the dynamics are not purely semi-anonymous since  different neighbors  exert different influence strengths. However, the randomness of the weights and i.i.d. assumption make our approach still valid.
\end{remark} 

In the next section, we show how the bounds derived above for Local Weak Convergence and concentration property
can be used to study asymptotic behavior of ASD on a labeled configuration model with power-law degree distribution.
More precisely, we will consider a sequence of labeled graphs with size $n$ and described by distributions $p_{\textbf{k}|a}^{(n)},q_{\textbf{k}|a}^{(n)}$ such that 
$$
p_{\textbf{k}|a}^{(n)}\stackrel{n\rightarrow\infty}{\longrightarrow}{p_{\textbf{k}|a}}, \ q_{\textbf{k}|a}^{(n)}\stackrel{n\rightarrow\infty}{\longrightarrow}{q_{\textbf{k}|a}}.
$$
Then we will consider the labeled branching process obtained by the construction above with asymptotic distributions. The following proposition quantifies the distance between the solution corresponding to the differential equation with distribution $p_{\textbf{k}|a}^{(n)},q_{\textbf{k}|a}^{(n)}$ and the solution corresponding to the differential equation with asymptotic distribution $p_{\textbf{k}|a},q_{\textbf{k}|a}$. 
\begin{proposition}\label{prop:asympdist1}
Let 
\begin{itemize} 
\item $\boldsymbol{\zeta}^{(n)}(t)$ be the solution of \eqref{eq:zeta} with $q_{\textbf{k}|a}^{(n)}$ and initial condition $\boldsymbol{\zeta}^{(n)}_0$;
\item $\boldsymbol{\zeta}(t)$ be the solution of \eqref{eq:zeta} with $q_{\textbf{k}|a}$ and initial condition $\boldsymbol{\zeta}_0$.
\item $\boldsymbol{y}^{(n)}(t)$ be the solution of \eqref{eq:zeta} with $p_{\textbf{k}|a}^{(n)}$ and initial condition $ \boldsymbol{y}^{(n)}_0$;
\item $\boldsymbol{y}(t)$ be the solution of \eqref{eq:zeta} with $p_{\textbf{k}|a}$ and initial condition $ \boldsymbol{y}_0$.
\end{itemize}
In addition let $\boldsymbol \phi( {\textbf{z}})-{\textbf{z}} $ and $\boldsymbol{\psi}{(\textbf{z}})$ be Lipschitz continuous in $[0,1]^{|\mathcal A|} $, and let $L$ and $M>0$ be the Lipschitz constants corresponding to infinity norm. Then
\begin{align*}
\sup_{t\in[0,m\Delta]}\|\boldsymbol{\zeta}^{(n)}(t)-\boldsymbol{\zeta}(t)\|_{\infty}\leq \frac{\|\boldsymbol{\zeta}^{(n)}_0-\boldsymbol{\zeta}_0\|_{\infty}}{(1-\Delta L)^m}
+\frac{1}{ L}\left(\frac{1}{(1-\Delta L)^m}-1\right)\|q^{(n)}_{\textbf{k}|a}-q_{\textbf{k}|a}\|_{\mathrm{TV}}\end{align*}
with $\Delta<1/L$ and 
\begin{align*}
&\sup_{t\in[0,m\Delta]}\|\boldsymbol{y}^{(n)}(t)-\boldsymbol{y}(t)\|_{\infty}\leq \frac{\|\boldsymbol{y}^{(n)}_0-\boldsymbol{y}_0\|_{\infty}}{(1-\Delta )^m}+\\
&
   \left(\frac{1}{(1-\Delta )^m}-1\right)  \left[ M\sup_{t\in[0,m\Delta]} {\|\boldsymbol{\zeta}^{(n)}(t)-\boldsymbol{\zeta}(t)\|_{\infty}}+ 
{\|p^{(n)}_{\textbf{k}|a}-p_{\textbf{k}|a}\|_{\mathrm{TV}}} \right].
\end{align*}
\end{proposition}

The proof can be found in Appendix \ref{Convergence_ODE}. 

In particular, if $p^{(n)}_{\textbf{k}|a}$ is a truncated version of $p_{\textbf{k}|a}$, we can apply to the previous bound  the following statement of immediate verification:
\begin{proposition} \label{prop:asympdist2} Consider a generic distribution 	$p_{\textbf{k}|a}$ and its 
truncated version  	$p^{(n)}_{\textbf{k}|a}$, i.e.  	$p^{(n)}_{\textbf{k}|a}= \frac{p_{\textbf{k}|a}\ind_{{{\textbf{k}} \in {\cal B}_{n}  }}}  {\sum_{\textbf{k}\in{\cal B}_{n} }  p_{\textbf{k}|a} }$ for a generic compact set ${\cal B}_n\in \N^{|\mathcal A|}$, then  we have: 	
\[
\|p^{(n)}_{\textbf{k}|a}-p_{\textbf{k}|a}\|_{\mathrm{TV}}= 1- \sum_{\textbf{k}\in{\cal B}_{n} }  p_{\textbf{k}|a} .
\]
\end{proposition}

\subsection{Asymptotic behavior on labeled configuration model with power-law degree distribution}
\label{sec:asdconfig}

In this section we consider the classical configuration model with a {\emph{truncated power-law degree distribution}, which is a particular case of labeled configuration model with $|\mathcal{A}|=1$. We simplify the notation:
let $p^{(n)}_{d,k}$ be the fraction of nodes with in-degree $d$ and out-degree $k$, where we have highlighted the number of nodes $n$, and let $\overline{d}=\sum_{d,k}dp^{(n)}_{d,k}$ be the average degree. 

\begin{center}
\begin{table}[h!]\centering
\boxed{\begin{tabular}{c|c}
 $\sum_{d,\textcolor{black}{k}}{d}q^{(n)}_{d,k} =\Theta(n^{\delta})$ & $0 \leq \delta<1/2$ \\
 $q_{k}^{(n)}\propto k^{-\beta}$ & $\beta>2$ \\
 $k_{\max}=\Theta(n^{\zeta})$ & $\zeta<\min\left\{\frac{1-\delta}{2},\frac{1}{\beta-1} \right\}$  \\
\end{tabular}}
\caption{Assumptions on power-law degree distribution}\label{tab:ass}
\end{table}
\end{center}
 
\begin{assumption}\label{ass1}
Let us assume that $$\sum_{d,\textcolor{black}{k}}{d}q^{(n)}_{d,k} =\sum_{d,\textcolor{black}{k}}\frac{d^2}{\overline d}p^{(n)}_{d,k} =\Theta(n^\delta)$$
with $0 \leq \delta < 1/2$. This means that we allow the average in-degree of a node, reached by  an edge selected uniformly at random,
to possibly scale with $n$. Moreover, we will assume that 
$$ q^{(n)}_{k}=\frac{1}{\overline{d}}\sum_{d}dp^{(n)}_{d,k} =O(k^{-\beta})$$
follows a power-law of exponent $\beta > 2$ and maximum value $k_{\max} = \Theta(n^{\zeta})$
with $\zeta<\min\left\{\frac{1-\delta}{2},\frac{1}{\beta-1} \right\}.$
\end{assumption}

Let $\mu_s $ be the $s$-th moment of $q = \{q^{(n)}_k\}_{k\geq0}$.
From Assumption \ref{ass1} we have
\begin{equation}\label{eq:mus1}
 \mu_s = \begin{cases} 
\Theta(1) &\mbox{if } \beta > s+1  \\
\Theta(n^{\zeta (s+1-\beta)}) &\mbox{if } 2 < \beta < s+1 .
\end{cases} 
\end{equation}
Notice that, being $\beta>2$,  $\mu_1$ is always finite and, therefore,  does not scale with $n.$
\medskip

In order to guarantee that the ASD over a network drawn uniformly at random from the configuration model ensemble is well approximated by the solution of ODE, it is sufficient that the terms in the upper bounds derived in \eqref{Top_res_conf_model} (see Example \ref{es:clssical_cm}), in Theorem \ref{thm:concentration_modgraphs}, and in Proposition \ref{prop:asympdist1} go to zero when $n\rightarrow\infty$. 
In the following, let $\mathcal{N}^{(n)}=(\V^{(n)},\E^{(n)})$ be a sequence of networks,  each one sampled from the corresponding model ensemble $\mathfrak{C}_{n,p^{(n)}}$, 
where $\{p^{(n)}\}_n$ is a sequence of truncated versions of a power law distribution of $p$ satisfying Assumption \ref{ass1}. 
For $t\geq0$, let $\mathcal{N}^{(n)}_{t}$ be the relevant neighborhood of a node $v$ chosen uniformly at random from the node set $\mathcal{V}^{(n)}$. Moreover, let $\mathcal{T}^{(n)}_{t}$ be the sequence of  truncated  Galton-Watson (GW)  processes  (see \cite{durrett_2010}) for which the root offspring  follows distribution $p^{(n)}$, while the degree of non-root nodes follow law $q^{(n)}$. Finally, let $p^{(n)} \stackrel{n \to \infty}{\to} p$ and $q^{(n)} \stackrel{n \to \infty}{\to} q$. We summarize the main assumptions and notations in Table \ref{tab:ass}.

Before presenting the topological result for the configuration model with power-law degree
distribution, we present two technical results, whose proofs are postponed to Appendix \ref{app:B}.

\begin{lemma}[Bound on the number of nodes/edges in $\mathcal{T}_t$] \label{lemma:F_branching}
Consider the GW process $\mathcal{T}$  in which the offspring distribution of the root follows law $p$, while the degree of remaining nodes follows law $q$. Let  $\mathcal{T}_t$ be  the corresponding random 
tree obtained by truncating $\mathcal{T}$ at time $t$, and $\widetilde{W}_t$ be the number of nodes in $\mathcal{T}_t$.  Let $h_n=c\log n$ for some $c>0$ and $t=o(h_n)$ as $n\rightarrow\infty$, then for any $s>0$
we have 
$$
\mathsf{F}_{\widetilde{W}_t}(x_n)\leq \frac{\mathbb{E}[N_{h_n}^s]}{x_n^s}+o(1/n)\qquad n\rightarrow\infty
$$
where  $\{N_{h}\}_{h\in\N}$  is  the number of nodes in a truncated   version of $\mathcal{T}$
with maximal width $h$. 
\end{lemma}

\begin{lemma}\label{lemma:powerlaw}
 Let  $\{N_h\}_{h\geq0}$ be a supercritical GW process, in which the offspring distribution of the root follows law $p$, while the degree of remaining nodes follows law $q$.
We have:
$\mathbb{E}[N_h^s] = O(\mu_s \cdot \mu_1^{s(h-1)})$, $\forall \beta > 1$, 
where $\mu_j$ is the $j$-th moment of $q$. 
\end{lemma}   
 
\begin{theorem}[Topological result for configuration model with power-law degree distribution]\label{top_res_pl}
With the above definitions, let $\mu_{\mathcal{N}^{(n)}_{t}}$ and $\mu_{\mathcal{T}^{(n)}_{t}}$ be the distributions of $\mathcal{N}^{(n)}_{t}$ and 
$\mathcal{T}^{(n)}_{t}$, respectively.  
Under Assumption \ref{ass1}, for $t=o(\log n)$, we have
$\|\mu_{\mathcal{T}^{(n)}_{t}}-\mu_{\mathcal{N}^{(n)}_{t}}\|_{\mathrm{TV}}\le  \mathbb{P}(\mathcal{T}_{t}\neq \mathcal{N}_{t}) =o(1)$
 when $n\rightarrow \infty$.
\end{theorem}

\begin{proof}
From inequality \eqref{Top_res_conf_model}, we have
 \begin{align}\label{Top_res_conf_model1}
 \|\mu_{\mathcal{T}_{t}}-\mu_{\mathcal{N}_{t}}\|_{\mathrm{TV}}&\leq  \mathbb{P}(\mathcal{T}_{t}\neq \mathcal{N}_{t})\leq
  \inf_{x>0} \left[\mathsf{F}_{\widetilde{W}_t}(x) + \frac{{\sum_{d,k}dq_{d,k}}x(x+1)}{2n\overline{d}} \right]
 \end{align}
where $\widetilde{W}_t$ is the number of nodes in $\mathcal{T}_t$.
Let $x_n=n^{(1-\delta)/2-\gamma}$ for some $\gamma>0$ and $h_n=c\log n$ for some $c>0$ then
\begin{align*}
  \|\mu_{\mathcal{T}_{t}}-\mu_{\mathcal{N}_{t}}\|_{\mathrm{TV}}&\leq \frac{{\sum_{d}dq_{d,k}}x_n(x_n+1)}{2n\overline{d}} +\mathsf{F}_{\widetilde{W}_t}(x_n) \qquad 
  n\rightarrow\infty\\
  &\leq \frac{{\sum_{d}dq_{d,k}}x_n(x_n+1)}{2n\overline{d}} +\frac{\mathbb{E}[N_{h_n}^s]}{x_n^s} +o(1/n)\qquad 
  \end{align*}
where the last inequality holds for any $s>1$ and  $\{N_{h}\}_{h\in\N}$  is  (the number of nodes of) a truncated  GW process of maximal width $h$, in which the offspring distribution of the root follows law $p$, while the degree of remaining nodes follow law $q$ (see Lemma \ref{lemma:F_branching}).

  Under Assumption \ref{ass1} we prove that there exists $s$ such that $\frac{\mathbb{E}[N_{h_n}^s]}{x_n^s}=o(1/n)$ as $n\rightarrow\infty$ and we conclude that for some $\gamma>0$
   \begin{align*}
  \|\mu_{\mathcal{T}_{t}}-\mu_{\mathcal{N}_{t}}\|_{\mathrm{TV}}&\leq\frac{1}{2\overline{d}n^{2\gamma}}+o(1/n)=o(1)\qquad n\rightarrow\infty.
 \end{align*}
To find a suitable value of $s$, we distinguish two cases.
\begin{itemize}
\item[(i)] 
If $\beta > \lfloor \frac{2}{1-\delta} \rfloor + 2$ we can simply choose $s=\lfloor \frac{2}{1-\delta} \rfloor + 1$. By so doing, we
stay in the case $\beta > s+1$, and from \eqref{eq:mus1} and 
Lemma \ref{lemma:powerlaw} we get: 
$ \mathbb{E}[N_{h_n}^s]=\Theta(n^{{cs}\log\mu_1})$ and
\begin{align*}
\frac{\mathbb{E}[N_{h_n}^s]}{x_n^s}=\Theta(n^{-s\left(\frac{1-\delta}{2}-\gamma-c\log\mu_1\right)})=o(1/n)\qquad n\rightarrow\infty
\end{align*}

\item[(ii)] If $\beta \leq \lfloor \frac{2}{1-\delta} \rfloor + 2$, we choose instead a sufficiently large value of $s$,
falling in the case $s > \beta - 1$ in which $\mu_s$ scales with $n$ as in \eqref{eq:mus1}.
In particular, from Lemma \ref{lemma:powerlaw} we have 
$$
\mathbb{E}[N_{h_n}^s]=\Theta(n^{\zeta(s+1-\beta)+{cs}\log\mu_1}) .
$$
Thus
$$
\mathbb{E}[N_{h_n}^s]/x_n^s=\Theta\left(n^{\zeta(s+1-\beta)+{cs}\log\mu_1-s\left(\frac{1-\delta}{2}-\gamma\right)}\right) .
$$
We now observe that if there exists $s\in\N$ such that
\begin{equation} \label{eq:vincolo3}
s \left( \frac{1-\delta}{2} - \zeta - c \log{\mu_1}-\gamma \right) > 1 - \zeta(\beta-1) 
\end{equation}
then $\mathbb{E}[N_{h_n}^s]/x_n^s=o(1/n)$ for $n\rightarrow\infty$.
Since $\zeta < 1/(\beta-1)$, the right hand side in \equaref{vincolo3} is positive,
and since $\zeta < \min\{\frac{1-\delta}{2}\}$, we can always find two sufficiently small constants
$\gamma$ and $c$ such that $\left( \frac{1-\delta}{2} - \zeta - c \log{\mu_1}-\gamma \right)$ is also
positive. Therefore, there exists an integer $s$ large enough such that both $s > \beta - 1$ and \equaref{vincolo3}
are satisfied. 
\end{itemize}
\end{proof}
\begin{corollary}\label{THM:FINAL}
For $t \geq 0$, let $\boldsymbol{Z}(t)$ be the state vector of the ASD dynamics on $\mathcal{N}^{(n)}$ and $z{\color{black}{_{\omega}^{(n)}}}(t) =\frac{1}{n}|\{v\in\V:Z_v(t)=\omega\}|$ be the fraction of state-$\omega$ adopters at time $t$. Under Assumptions \ref{ass1} for $t=o(\log n)$, for any $\eta>0$ 
\begin{align*}
\mathbb{P}(|z{\color{black}{_{\omega}^{(n)}}}(t)-y_{\omega}(t)|>\eta)= o(1)\qquad \text{for } n\rightarrow\infty
\end{align*}
where $y_{\omega}(t)$ is the solution of \eqref{eq:zeta2}  over a GW tree  $\mathcal{T}_{t}$ with the 
asymptotic degree statistics $p$  and $q$.
\end{corollary}

\begin{proof}
Let $\mathcal{T}^{(n)}_{t}$ be the continuous-time branching process truncated up to time $t$, $\mu_{\mathcal{T}^{(n)}_{t}}$ be its distribution. Denote by $\boldsymbol{y}^{(n)}(t)$ the solution of \eqref{eq:zeta2} with $p_{\textbf{k}|a}^{(n)}$ and initial condition $ \boldsymbol{y}^{(n)}_0$.
We have {\color{black}\begin{align}\begin{split}\label{eq:conver}
\mathbb{P}(|z^{(n)}_{\omega}(t)-y_{\omega}(t)|\geq \eta)\leq\mathbb{P}(|z^{(n)}_{\omega}(t)-y^{(n)}_{\omega}(t)|\geq \eta/2)+\mathbb{P}(|y^{(n)}_{\omega}(t)-y_{\omega}(t)|\geq \eta/2)
\end{split}\end{align}
From Theorem \ref{top_res_pl}  we obtain that if  $t=o(\log n)$ as $n\rightarrow\infty$ we have $\|\mu_{\mathcal{T}^{(n)}_{t}}-\mu_{\mathcal{N}^{(n)}_{t}}\|_{\mathrm{TV}}=o (1)$ when $n\rightarrow \infty$. We conclude that for any 
$\eta>0$ there exists a sufficiently large $n_0$ such that, if $n\geq n_0$,  then $\|\mu_{\mathcal{N}^{(n)}_{t}} -\mu_{\mathcal{T}^{(n)}_{t}}\|_{\mathrm{TV}}\leq  \mathbb{P}(\mathcal{T}_{t}\neq \mathcal{N}_{t})   \leq {\eta}.$
 Using Proposition \ref{thm:ASD_over_gen_net}, it follows that for any $\eta>0$ and large enough $n$:
 \begin{align}\label{eq:conver}
\mathbb{P}(|z^{(n)}_{\omega}(t)-y_{\omega}(t)|\geq \eta)\leq\mathbb{P}(|z^{(n)}_{\omega}(t)-\overline{z}^{(n)}_{\omega}(t)|\geq \eta/4)+\mathbb{P}(|y^{(n)}_{\omega}(t)-y_{\omega}(t)|\geq \eta/2)\nonumber
\end{align}
From Theorem \ref{thm:concentration_modgraphs}, by choosing $x=n^{4/9}$, $s=3$, and $\epsilon>0$,
  we get
 \begin{align}
 &\mathbb{P}(|z^{(n)}_{\omega}(t)-y_{\omega}(t)|\geq \eta)\leq \mathbb{P}(|y^{(n)}_{\omega}(t)-y_{\omega}(t)|\geq \eta/2)+o(1)\nonumber
\end{align}
where, we have  applying jointly Corollary \ref{cor:mom} and
Corollary \ref{coro:powerlaw} in Appendix~\ref{app:B} to bound the third moment of $V_t^v$. 
Finally  Propositions \ref{prop:asympdist1} and \ref{prop:asympdist2} guarantee that  $|y^{(n)}_{\omega}(t)-y_{\omega}(t)| \to 0 $ and $\mathbb{P}(|y^{(n)}_{\omega}(t)-y_{\omega}(t)|\geq \eta/2)=0$ for large enough $n$.}
 \end{proof}

\begin{remark}[Relation to Theorem 1 in \cite{RCF18}] In \cite{RCF18} a similar approximation result was proved in the specific case of binary synchronous LTM. More precisely Theorem 1 in \cite{RCF18} implies that, for sequences of networks whose network statistics converge to a given limit as the network size grows large, the concentration result is guaranteed on the configuration model for finite values of $t$, provided that the maximum in- and out-degrees remain bounded.  It is worth remarking that Theorem \ref{THM:FINAL}, proved in this paper, is much more general and the result applies to:
(a) asynchronous semi-anonymous dynamics;
(b) networks with maximum in- and out-degrees growing as a function of network size $n$;
(c) networks with power-law degree distribution.
\end{remark}

\begin{remark}[Transient behavior versus Asymptotic behavior]
Theorem \ref{THM:FINAL} guarantees the approximation result for values of $t$ growing at most as $o(\log n)$. However, using techniques devised in \cite{2011arXiv1111.5710B} for the exchange of limits in $t $ and $n$, it can be shown that with high probability, as the network size grows large, the asymptotic fraction of $\omega$-adopters concentrates on the set of all stationary points of the ODE in \eqref{prop:ODE}. This implies that, when the ODE in \eqref{prop:ODE} admits multiple stationary points, then from our result it is not possible to predict the asymptotic limit of the system. However, if the system in \eqref{prop:ODE} admits a unique (globally attractive) stationary point, convergence is guaranteed in that point for all initial conditions. 
\end{remark}
}

\subsection{Asymptotic behavior on Configuration Block Model}\label{section:CBM2}
{
In this section, we apply the bounds derived for Local Weak Convergence and Concentration Property in Section \ref{sec:rn}
to study the asymptotic ASD on a labeled configuration model
with community structure, which is a key feature of many real systems. 
In particular, we consider a Configuration Block Model (CBM) 
with $K$ communities of sizes $\{n_i\}_{i=1}^K$, which are mapped into corresponding classes with labels $\{a_i\}_{i=1}^K$.

When the maximum in/out degree of nodes is finite both Local Weak Convergence 
and Concentration property can be easily proven by taking the simple worst-case
in which all nodes have in/out degree equal to the maximum in/out degree, so we
will consider here a more challenging case, where in/out degree of nodes is allowed to scale with $n$. 

However, to simplify the analysis, we assume no correlation between in-degree and out-degree of a node.
As a consequence, the law of $p$ is the same as the law of $q$,
and $\zeta_{\omega|a,a'}(t) = y_{\omega|a}(t)$ (see Proposition \ref{prop:ODE}).

Moreover, we assume that the number of edges established from a node of community $i$ towards
nodes of community $j$ is independent for any pair $(i,j)$, 
including the special case $i = j$, i.e., intra-community edges.
Therefore, $p_{\textbf{d},\textbf{k}|a}$ factorizes into:
$$p_{\textbf{d},\textbf{k}|a} = \prod_{i \in \mathcal{A}} p^{\texttt{in}}_{i,a}[d_i] \prod_{j \in \mathcal{A}} p^{\texttt{out}}_{a,j}[k_j]$$
We will require that in/out degree sequences of the nodes, although possibly dependent on the network size $n$,
generate empirical distributions $p^{\texttt{in}}_{i,a}[d]$ and $p^{\texttt{out}}_{a,j}[k]$ with a light tail, 
for any pair $(i,a)$ or $(a,j)$, thus having finite moments of any order.   

In order to guarantee that the ASD over a network drawn uniformly at random from the CBM ensemble is well approximated by the solution of the ODE, it is sufficient that the terms in the upper bound derived in \equaref{bound_topologico} and in Theorem \ref{thm:concentration_modgraphs} 
go to zero when $n\rightarrow\infty$. In the following, let $\mathcal{N}^{(n)}=(\V^{(n)},\E^{(n)})$ be a sequence of networks,  each one sampled from the corresponding CBM ensemble  \\ $G^{(n)}=G(\{n_i\}_{i=1}^{K},p^{(n)}_{\textbf{d},\textbf{k},a})$,
and, for $t\geq0$, let $\mathcal{N}^ {(n)}_{t}$ be the  relevant neighborhood of a node $v$ chosen uniformly at random from the 
node set $\mathcal{V}^{(n)}$. Moreover, let $\mathcal{T}^{(n)}_{t}$ be the sequence of GW processes with offspring 
distribution following law $p^{(n)}$. Finally, let $p^{(n)} \stackrel{n \to \infty}{\to} p$.
\medskip

\begin{theorem}\label{top_res_sbm}
Under above definitions and assumptions on the CBM ensemble, 
let $\mu_{\mathcal{N}^{(n)}_{t}}$ and $\mu_{\mathcal{T}^{(n)}_{t}}$ be the distributions of $\mathcal{N}^{(n)}_{t}$ 
and $\mathcal{T}^{(n)}_{t}$, respectively.
For $t=o(\log n)$, we have $\|\mu_{\mathcal{T}^{(n)}_{t}}-\mu_{\mathcal{N}^{(n)}_{t}}\|_{\mathrm{TV}}=o(1)$ 
when $n\rightarrow \infty$.
\end{theorem}
The proof is postponed to Appendix \ref{app:B}.

\begin{corollary}\label{THM:FINAL-CBM}
Under  above definitions and assumptions on the CBM ensemble, 
for $t \geq 0$, let $\boldsymbol{Z}(t)$ be the state vector of the ASD dynamics on $\mathcal{N}^{(n)}$ and $z_{\omega}^{(n)}(t) =\frac{1}{n}|\{v\in\V:Z_v(t)=\omega\}|$ be the fraction of state-$\omega$ adopters at time $t$. For $t=o(\log n)$ and any $\eta>0$:
\begin{align*}
\mathbb{P}(|z_{\omega}^{(n)}(t)-y_{\omega}(t)|>\eta)= o(1)\qquad \text{for } n\rightarrow\infty
\end{align*}
where $y_{\omega}(t)$ is the solution of \eqref{eq:zeta2}  over a GW tree  $\mathcal{T}_{t}$ with the 
asymptotic degree statistics $p$  and $q$.
\end{corollary}
The proof follows exactly the same lines of   Corollary \ref{THM:FINAL}.
{\color{black}
\begin{remark}
Note that  the approach  followed in the proof of Theorem \ref{top_res_sbm} can be extended  to a 
significantly more general class of labeled configuration graphs $\mathfrak{C}_{n,p}$.
The key step of the approach pursued in Theorem \ref{top_res_sbm} is to  find a   distribution meeting the following two constraints: 
 i) it  stochastically  dominates  the out-degree distribution of every class; ii)  its tail  is not too heavy, so that we can  effectively bound the moments on the number of nodes in the corresponding GW truncated tree by exploiting \eqref{eq:mic1} (see Appendix \ref{app:B}).
 Under such conditions we can conclude that  $\|\mu_{\mathcal{T}_{t}}-\mu_{\mathcal{N}_{t}}\|_{\mathrm{TV}}=o (1)$ 
 when $n\rightarrow \infty$. In particular, the previous considerations apply every time the dominating distribution is a 
 power law meeting the constraints in Assumption \ref{ass1}.
\end{remark}

} 
} 
\section{Analysis of mean-field ODE on regular random graphs}\label{sec:mf}
 
In the case of regular random graphs, where all nodes have the same out-degree, we can analytically derive some interesting properties
of the ODEs describing the temporal evolution of the system (see Proposition ~\ref{prop:ODE}), for each of the three examples of ASD introduced in
Section \ref{Examples_ASD}. In particular, we can characterize the equilibrium points of the system and their stability.      
 
\subsection{Ternary Linear Threshold Model (TLTM)}
We assume that all agents have the same out-degree $k$ and symmetric thresholds, i.e., $k_v=k$ and $a_v^{\pm}=r$ for all $v\in\V$. 
In this case, there is a single class $a=r$ for which $p_{k|r}= q_{k|r} = 1$, and 
\begin{align*}
\Theta^{(r)}_1(\xi_1,\xi_{-1},k-\xi_1-\xi_{-1})&=\boldsymbol{1}_{\{\xi_1-\xi_{-1}\geq r\}}\\
\Theta^{(r)}_0(\xi_1,\xi_{-1},k-\xi_1-\xi_{-1})&=\boldsymbol{1}_{\{\xi_1-\xi_{-1}\in(-r,r)\}}\\
\Theta^{(r)}_{-1}(\xi_1,\xi_{-1},k-\xi_1-\xi_{-1})&=\boldsymbol{1}_{\{\xi_1-\xi_{-1}\leq- r\}}
\end{align*}
From Proposition \ref{prop:ODE} we have that, given an initial distribution $(y_1(0),y_{-1}(0))$, the dynamics over the 
continuous-time branching process is described by 
\begin{equation}\begin{cases}
\frac{\mathrm{d}y_{1}}{\mathrm{d}t}=\phi^{(k,r)}_{+}(y_1(t),y_{-1}(t))-y_{1}(t),\\
\frac{\mathrm{d}y_{-1}}{\mathrm{d}t}=\phi^{(k,r)}_{-}(y_1(t),y_{-1}(t))-y_{-1}(t),\\
\end{cases}\label{ODE:TLTM}
\end{equation}
and $y_{0}=1-y_{1}-y_{-1}$, where
\begin{equation}\begin{split}\label{eq:out_reg_phi}
\phi_{+}^{(k,r)}(x,z) &= \sum_{u=r}^{k} \sum_{v=0}^{(k-u) \wedge (u-r)}{k \choose u}{k-u \choose v} x^{u} z^v (1-x-z)^{k-u-v}\\
&=\sum_{v=0}^{\lfloor\frac{k-r}{2}\rfloor}\sum_{u=v+r}^{k-v }{k \choose u}{k-u \choose v} x^{u} z^v (1-x-z)^{k-u-v}
\end{split}
\end{equation}
and $\phi_{-}^{(k,r)}(x,z) = \phi_{+}^{(k,r)}(z,x)$.

Some analytical properties of the dynamical system can be deduced from the analysis of $\phi_+^{(k,r)}(x,z)$ and $\phi_-^{(k,r)}(x,z)$. 
In particular, we are interested in finding stationary points, i.e. those points satisfying the following equations
$$
\begin{cases}\phi_{+}^{(k,r)}(x,z)=x\\
\phi_{-}^{(k,r)}(x,z)=z,         \qquad  x=z\le 1
\end{cases}
$$
and analysing their stability properties. 
Before presenting the main result we give some preliminary lemmas. 

\begin{lemma}\label{lemma:prop_phi}
Let $\phi_{+}^{(k,r)}(x,z)$ be as defined in \eqref{eq:out_reg_phi}. Then the following properties hold:
\begin{enumerate}
\item $\phi^{(k,r)}_{+}(x,z)$ is non decreasing in $x$ and strictly increasing if $0<r\leq k$;
\item $\phi^{(k,r)}_{+}(x,z)$ is non increasing in $z$ and strictly decreasing if $0\leq r< k$;
\item \label{lemma:prop_phi3}$\phi^{(k,r)}_{+}(0,z)=0$ for $0<r\leq k$, $\phi^{(k,0)}_{+}(x,0)=1$, $\phi^{(k,r)}_{+}(1,0)=1$, for $0\leq r\leq k$;
\item $\nabla_x\phi^{(k,r)}_{+}=\sum_{v=0}^{\lfloor{\frac{k-r}{2}}\rfloor}\chi(v+r-1,v)x^{v+r-1} z^v (1-x-z)^{k-2v-r}$ with
$$\chi(u,v)={k\choose u}{k-u\choose v}(k-u-v)$$
\item $\nabla_z\phi^{(k,r)}_{+}=-\sum_{u=r}^k\chi(u,v') x^{u} z^{v'} (1-x-z)^{k-u-v'-1}$ with $$v'=(k-u)\wedge (u-r).$$
\end{enumerate}
\end{lemma}
The proof is trivial and we omit it for brevity. 

\begin{proposition}If $2 \leq r < k$, then
\begin{enumerate}
\item   the equation $\phi^{(k,r)}_+(x,0) = x$ has exactly three solutions $\{0,x^{\star},1\}$ with $x^{\star}\in(0,1)$;
\item  for every $x\in(x^{\star},\bar{x})$ with $0<x^{\star}< \bar{x}\leq1$ there exists a unique value   $z(x)$ such that
$
\phi^{(k,r)}_{+}(x,z(x))=x.
$
\item  the equation $\phi^{(k,r)}_-(0,z) = z$ has exactly three solutions $\{0,z^{\star},1\}$ with $z^{\star}\in(0,1)$;

\item for every $z \in (z^{\star}, \bar{z})$ with $0<z^{\star}< \bar{z}\leq1$,  
there exists a unique value  $x(z)$ such that
$
\phi^{(k,r)}_{-}(x(z), z) = z.
$
\end{enumerate}
\end{proposition}
\begin{proof}
 From definition \eqref{eq:out_reg_phi} we have
\begin{align*}
\phi^{(k,r)}_+(x,0)=\sum_{u=r}^{k}{k\choose u}x^u(1-x)^{k-u}
\end{align*}
If $2\leq r< k$, the function $\phi^{(k,r)}_+(x,0)$ has a lazy-S-shaped graph, i.e., it is increasing, with a unique inflection point at $\widetilde{x}= (r - 1)/(k - 1)$, it is convex on the left-hand side of $\widetilde{x}$ and concave on the right-hand side of $\widetilde{x}$ (see Lemma 4 in \cite{RCF18}). From this fact and the observation that $\phi^{(k,r)}_+(0,0)=0$ and $\phi^{(k,r)}_+(1,0)=1$ we get the assertion at Point 1. It can also be proved that $x^{\star}\in[(r-1)/k,r/k]$.
Denoting $F(x,z)=\phi^{(k,r)}_+(x,z)-x$ and observing that $F(x^{\star},0)=\phi^{(k,r)}_+(x^{\star},0)-x^{\star}=0$ and $\nabla_z F(x^{\star},0) \neq 0$ (see expression in Point \ref{lemma:prop_phi} of Lemma \ref{lemma:prop_phi}), the statement in Point 2. is obtained by the implicit function theorem \cite{Apostol}. Point 3. and 4. are straightforward from the relation $\phi_-^{(k,r)}(x,z)=\phi_+^{(k,r)}(z,x)$.
\end{proof}
In Figure \ref{fig:ContorniALTM} the functions $z(x)$ such that
$
\phi^{(k,r)}_{+}(x,z(x))=x
$ and $x(z)$ such that $\phi^{(k,r)}_{-}(x(z),z)=z$ are depicted for threshold values  $r=2$ (left) and $r=3$ (right) and degree $k=10$. In addition to stationary points $$\{(0,0),(x^{\star},0),(1,0),(0,z^{\star}),(0,1))\},$$ as derived in Lemma \ref{lemma:prop_phi}, extra stationary points are placed at  intersections between curves $z(x)$
and $x(z)$.
In the specific case with $k=10$ it can be noticed that if $r=2$ then there are two additional stationary points.
\begin{figure}[!ht]\centering
\includegraphics[width=0.45\columnwidth,trim = 20mm 10mm 20mm 10mm, clip]{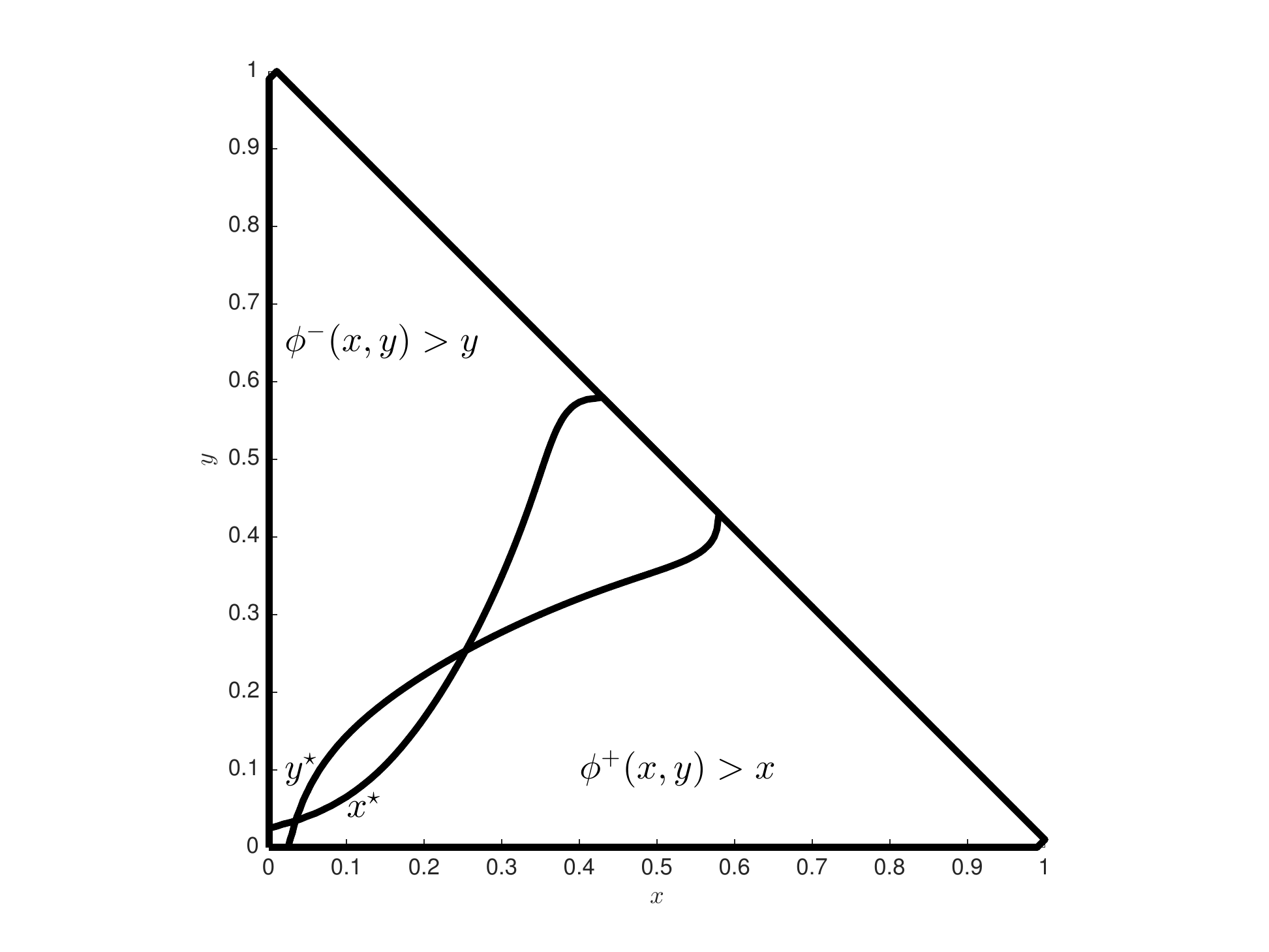}
\includegraphics[width=0.45\columnwidth,trim = 20mm 10mm 20mm 10mm, clip]{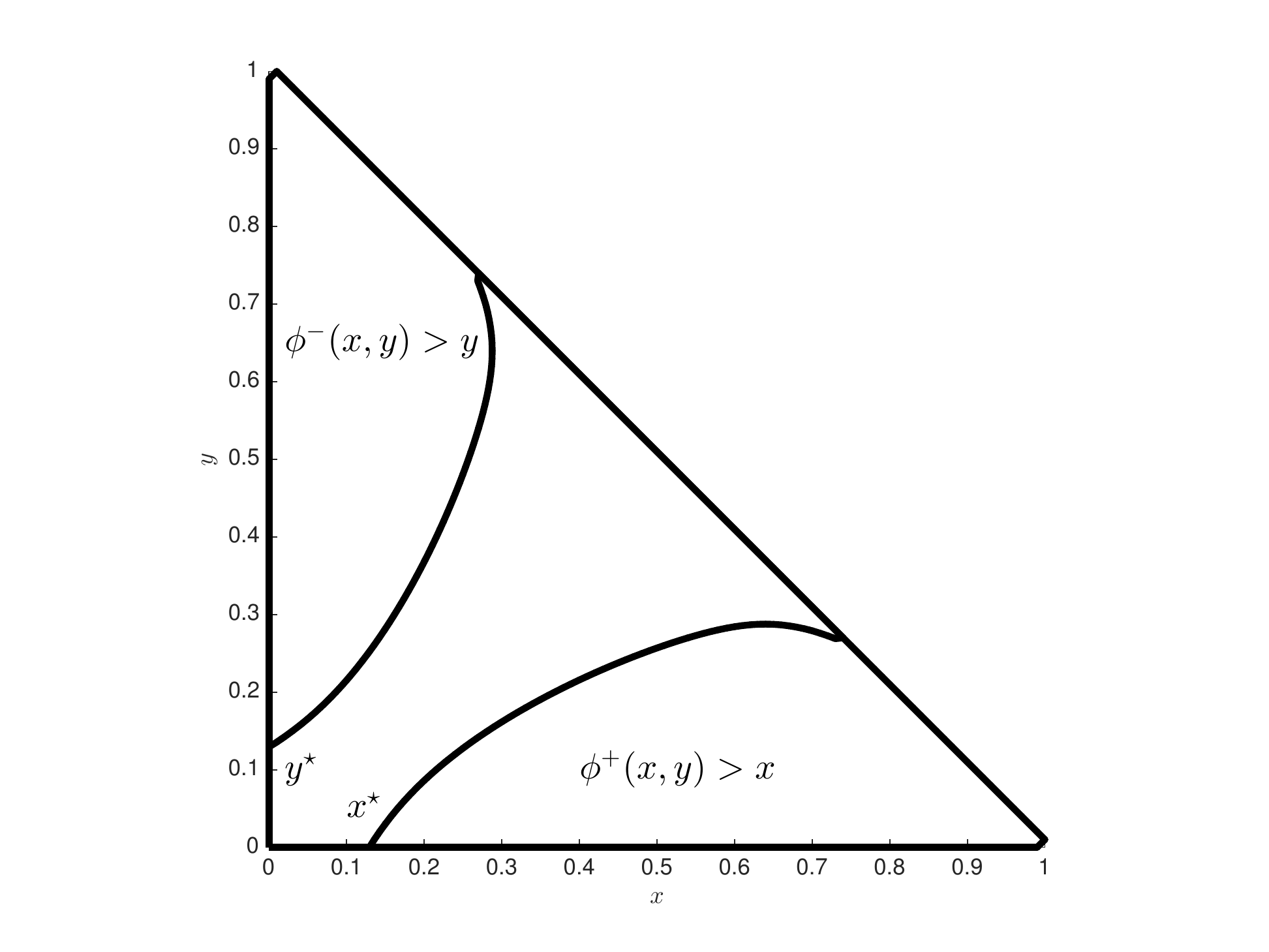}
\caption{$k=10,r=2$ (left), $k=10,r=3$ (right).}\label{fig:ContorniALTM}
\end{figure}
The following proposition gives sufficient conditions guaranteeing that the set of stationary points only contains the trivial points $\{(0,0),(x^{\star},0),(1,0),(0,z^{\star}),(0,1))\}$.
\begin{proposition} If $r\geq (k+1)/2$, 
$\{(0,0),(x^{\star},0),(1,0),(0,z^{\star}),(0,1))\}$ are the only fixed points of the system. 
\end{proposition}
\begin{proof}
Notice that
\begin{align*}
\phi_{+}^{(k,r)}(x,x)=\phi_{-}^{(k,r)}(x,x)&=\sum_{u=r}^{k}\sum_{v=0}^{(k-u)\wedge (u-r)}{k \choose u}{k-u \choose v}x^{u+v}(1-2x)^{k-u-v}\\
&\leq \sum_{u=r}^{k}{k \choose u}x^u(1-x)^{k-u}=\varphi^{(k,r)}(x).
\end{align*}
If $r\geq (k+1)/2$, we have $\varphi^{(k,r)}(x)< x$ for all $x<1/2$ and, consequently, $\phi_{+}^{(k,r)}(x,x)< x$ from which we conclude the assertion.
\end{proof}

The following corollary can be proved by linearization.
\begin{corollary}\label{corol:Local_stability_analysis}The following properties hold
\begin{enumerate}
\item $(0,0)$ is a locally stable stationary point for $2\leq r\leq k$, unstable if $r=1$.
\item $(0,1)$ and $(1,0)$ are locally stable stationary points for $1\leq r< k$, unstable otherwise. 
\item $(x^{\star},0)$ and $(0,z^{\star})$ are unstable stationary points.
\item The set of points $\{(x,z)\in \R^{2}: x=z\}$ is invariant.
\end{enumerate}
\end{corollary}

The basins of attraction for the ODE in \eqref{ODE:TLTM} are shown in Figure \ref{fig:univerise} for  degree $k  =10$ and threshold $r\in\{2,3\}$. 
Basins are  evaluated numerically, by solving the ODE system for a wide set of initial conditions.
More specifically, for each initial condition $x_0,z_0$ the color in the picture represents the asymptotically stable equilibrium point (yellow for (0,1), green for (0,0) and blue for state $(1,0)$) to which the trajectory tends. As to be expected by Corollary \ref{corol:Local_stability_analysis}, we have that points $(0,0)$, $(0,1)$ and $(1,0)$ are locally stable stationary points. 

\begin{figure}[!ht]
\centering
\includegraphics[width=0.47\columnwidth,trim = 20mm 10mm 20mm 10mm, clip]{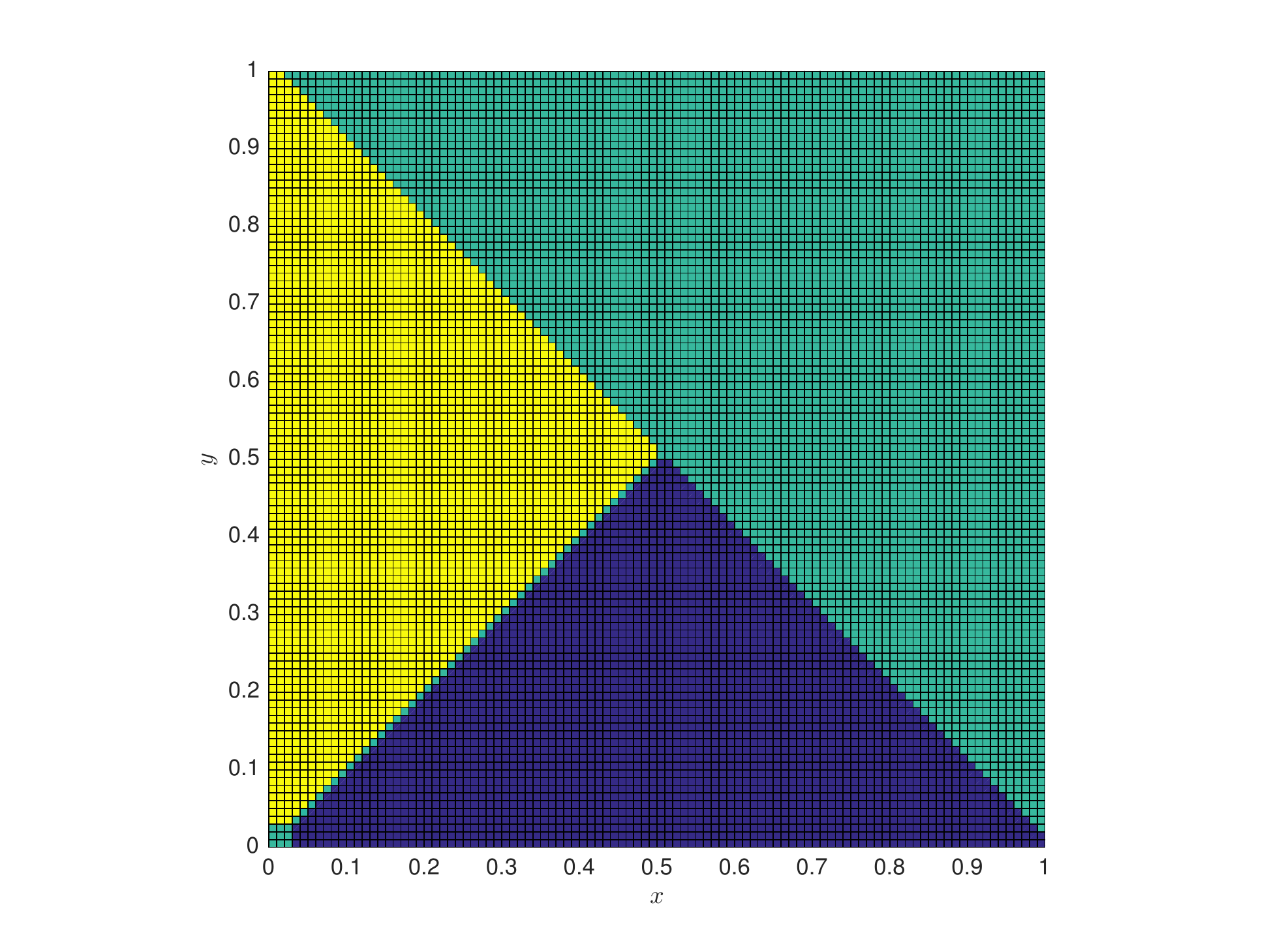}
\includegraphics[width=0.47\columnwidth,trim = 20mm 10mm 20mm 10mm, clip]{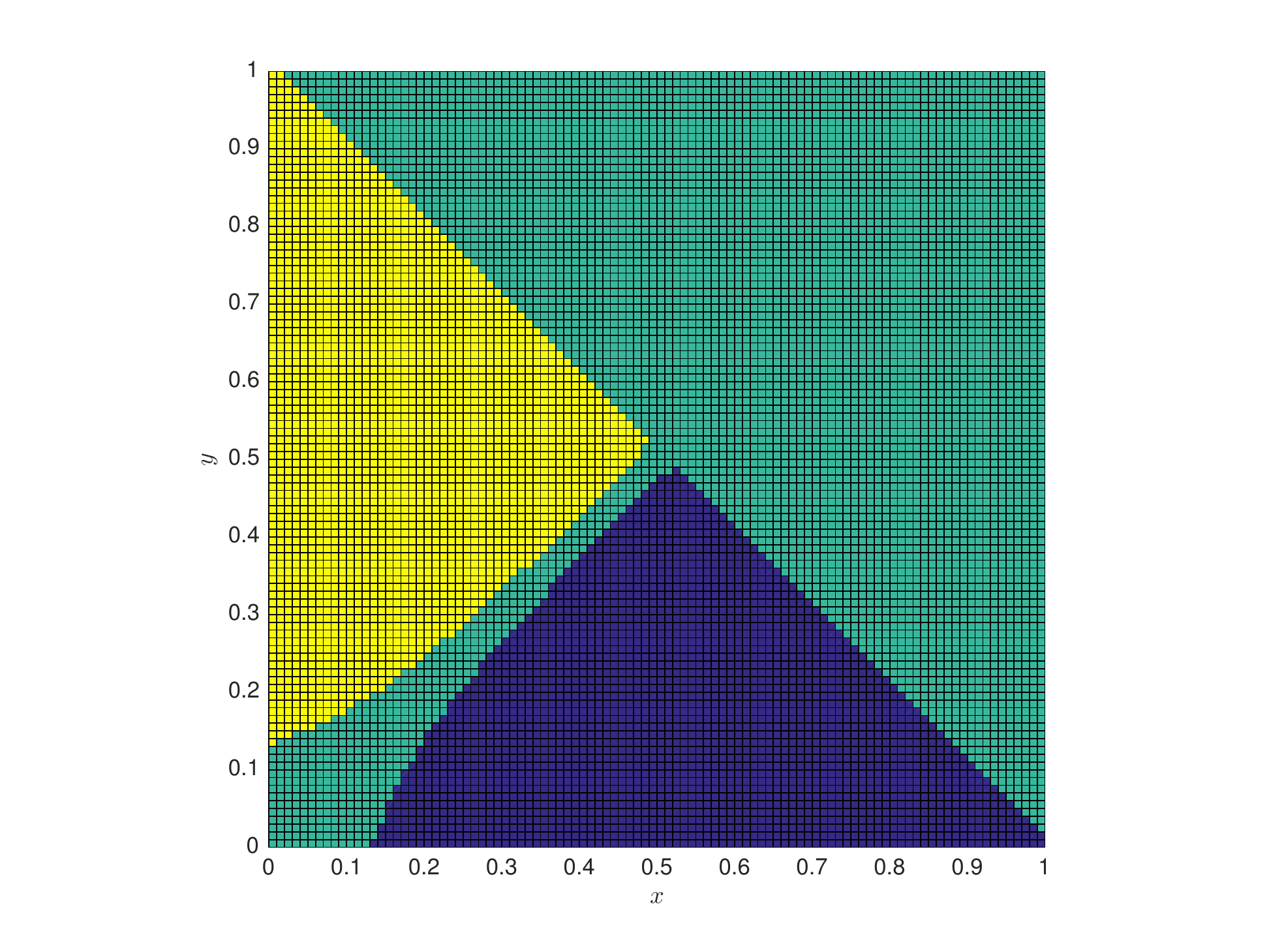}
\caption{$k=10$ and $r=2$ (left), $r=3$ (right)}
\label{fig:univerise}
\end{figure}

\subsection{Binary Response with Coordinating and Anti-coordinating agents (BRCA)}
We again consider the homogenous case in which all nodes have out-degree $k$. For simplicity, in the following, 
we will suppose that $k$ is odd, in order to avoid ties, although the extension to the case of even $k$ is straightforward. 
Let $\alpha$ be the fraction of coordinating nodes. From Proposition \ref{prop:ODE}, the evolution of the node states is governed by the following ODE:
\begin{equation} \label{eq:BRCA_ODE}
\frac{\mathrm{d}y_{1}(t)}{\mathrm{d}t} = \alpha \phi_1 \left(y_{1}(t) \right) + (1-\alpha) \left(1-\phi_1 \left(y_{1}(t) \right) \right) -y_{1}(t), 
\end{equation}
where $\phi_1 \left(y_{1}(t) \right)$ is the probability with which a coordinating node enters  state 1, given by
\begin{equation} \label{eq:BRCA_phi}
\phi_1 \left(y_{1} \right) = \sum_{k_1 = \lceil k/2 \rceil}^k {k \choose k_1} y_1^{k_1} (1-y_1)^{k-k_1}. 
\end{equation}
The above ODE derives from the fact that the probability of stepping to state 1 for an anti-coordinating node is equal to the probability of stepping to state -1 for a coordinating node.
Let us define $\ell_1\left(y_{1}(t) \right)$ the RHS of \eqref{eq:BRCA_ODE}. It is easily seen that, since  $\phi_1 \left(1/2\right) = 1/2$, $\ell_1\left(1/2 \right) = 0$, so that $y_1 = 1/2$ is a stationary point for every $\alpha$. 

The derivative of $\ell_1$ with respect to $y_{1}$ is given by:
\begin{eqnarray}
\frac{\mathrm{d}\ell_1}{\mathrm{d}y_{1}} &=& (2 \alpha-1) \frac{\mathrm{d}\phi_1}{\mathrm{d}y_{1}} - 1 \\
& = & (2 \alpha-1) k {k-1 \choose \lfloor \frac{k}{2} \rfloor} \left( y_{1}  (1-y_1) \right)^{\lfloor \frac{k}{2} \rfloor} -1
\end{eqnarray}

From the above equation, it turns out that we have two distinct regimes, according to whether $\alpha \leq \alpha_{\mathrm{th}}$ or $\alpha > \alpha_{\mathrm{th}}$, with
\begin{equation} \label{eq:BRCA_phi_2}
\alpha_{\mathrm{th}} = \frac1{2} \left( 1 + \frac{2^{k-1}}{k {k-1 \choose \lfloor \frac{k}{2} \rfloor}}\right) 
\end{equation}  

\begin{itemize}
\item If $\alpha \leq \alpha_{\mathrm{th}}$, then $\frac{\mathrm{d}\ell_1}{\mathrm{d}y_{1}} \leq 0$ for $0 < y_{1} < 1$. It then turns out that the only stationary point is $y_{1} = 1/2$, which is stable and has a basin of attraction equal to $[0,1]$.

\item If $\alpha > \alpha_{\mathrm{th}}$, then $\frac{\mathrm{d}\ell_1}{\mathrm{d}y_{1}}$ has two zeros, symmetric with respect to $z_1 = 1/2$, in the solutions of the quadratic equation
\[
y_{1}  (1-y_1) = \frac1{\sqrt[\lfloor \frac{k}{2} \rfloor]{(2 \alpha-1) k {k-1 \choose \lfloor \frac{k}{2} \rfloor}}}
\]
Because of that, there are three stationary points, out of which the one in $z_1 = 1/2$ becomes unstable. The other two are symmetric with respect to $z_1 = 1/2$, i.e., $y^{\pm} = 1/2 \pm \epsilon$, and stable, with basins of attraction $[0,1/2)$ and $(1/2,1]$.  Figure \ref{fig:brca} shows the value of $\epsilon$ as a function of $k$ and $\alpha$. 
\end{itemize}

\begin{figure}[htb]
\centering
\includegraphics[width=0.7\columnwidth]{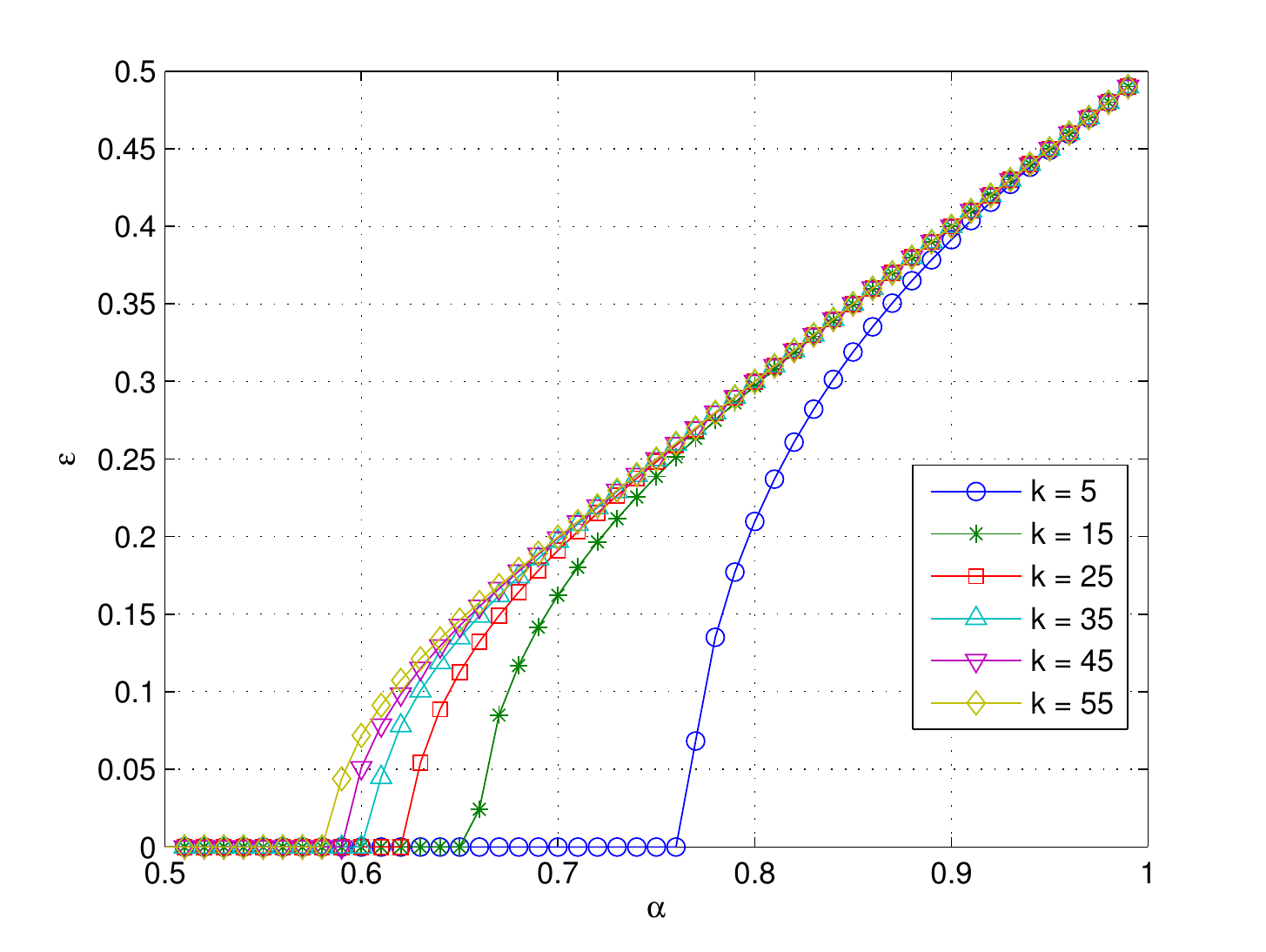}
\caption{Position of the stable stationary points as a function of the fraction of cooperative nodes and the common node degree.\label{fig:brca}}
\end{figure}

\subsection{Evolutionary Roshambo Game (ERG)} \label{sec:ERG_analytic}
In this section, we consider the ERG dynamics on a regular graph with $k_v = k$. 

Let $y_{\omega}$, $\omega \in \{R, P, S\}$ be the fraction
of nodes in state $\omega$. Moreover, let $\pi_{\omega}$, $\omega \in \{R, P, S\}$ be the probability that a given node, when activated, switches to state $\omega$ thanks to its neighbors' states. More precisely, for $\omega = R$:
\begin{equation}
\begin{split}
\pi_R =& \mathbb{P} \left\{ \xi_S > \frac{k}{3},  \,\,\,\xi_P < \frac{k}{3} \right\}  + \frac1{2} \mathbb{P} \left\{ \xi_S = \frac{k}{3},  \,\,\,\xi_P < \frac{k}{3} \right\} \\
& + \frac1{2} \mathbb{P} \left\{ \xi_S > \frac{k}{3},  \,\,\,\xi_P = \frac{k}{3} \right\}  + \frac1{3} \mathbb{P} \left\{ \xi_S = \frac{k}{3},  \,\,\,\xi_P = \frac{k}{3} \right\}
\end{split}
\end{equation}
and analogously for $\pi_S$ and $\pi_P$.

When it is activated, a node in state $\omega$ changes state with probability $1 - \pi_{\omega}$, while a node not in state $\omega$ switches to that state with probability $\pi_{\omega}$, by definition.  We thus have:
\begin{equation}
\frac{d y_{\omega}}{dt }  =  (1-y_{\omega}) \pi_{\omega}- y_{\omega} (1-\pi_{\omega}), \,\,\,\omega \in \{R, P, S\}  
\end{equation}
At an equilibrium point, $\frac{d y_{\omega}}{dt }  = 0$, so that $y_{\omega} = \pi_{\omega}$, $\omega \in \{R, P, S\}$. 

\begin{proposition} 
The only equilibrium point for the ERG dynamics on a regular graph is $(y_R, y_P, y_S) = (1/3,1/3,1/3)$.
\end{proposition}
\begin{proof}
We limit the proof to the case in which $k$ is not divisible by 3, to deal with a simpler expression for $\pi_{\omega}$. The extension to the general case is straightforward.

The point $(y_R, y_P, y_S) = (1/3,1/3,1/3)$ is clearly of equilibrium since $y_{\omega} = \pi_{\omega} = 1/3$, $\omega \in \{R, P, S\}$.

We prove that there cannot be other equilibria by first considering the region for which $y_S < y_P < y_R$ and showing that, in this region, $\pi_P > \pi_R$. By writing down the expression of the probabilities,  we have
\begin{equation}
\pi_P - \pi_R = \sum_{u = \lceil k/3 \rceil}^k \sum_{v = 0}^{\lfloor k/3 \rfloor} {k \choose u,v} \left( y_R^u y_S^v y_P^{k-u-v} - y_S^u y_P^v y_R^{k-u-v}\right)
\end{equation}
Out of the terms in the above sum, those for which $ k-u-v \leq u $ are positive because of the rearrangement inequality. Instead, those for which $ k-u-v > u $ can be negative. However, for a given $v$, the sum of the pair of terms in which $k - u -v$ and $u$ take the pair of values $u$ and $w$ is given by
$$
{k \choose u,v} \left[ y_R^u \left(y_S^v y_P^w -  y_S^w y_P^v\right) +  y_R^w \left(y_S^v y_P^u -  y_S^u y_P^v\right)\right]
$$
which is positive, again for the rearrangement inequality. Thus, for $y_S < y_P < y_R$, $\pi_P > \pi_R$ and there cannot be any equilibrium point in that region.

Analogously, we can show that in the region for which $y_S > y_P > y_R$, $\pi_P < \pi_R$ and there cannot be any equilibrium point in that region either. A similar argument shows the same thing in the regions characterized by a pair of equal coordinates, such as $y_S = y_P < y_R$, and so on.   

Finally, by symmetry, we can show that for any other ordering of the coordinates, there cannot be an equilibrium point. Thus, the only point of equibrium is $(y_R, y_P, y_S) = (1/3,1/3,1/3)$.
\end{proof}

\section{Numerical experiments}\label{sec:experiments}
In this section we present a few interesting cases in which we run our examples of ASD
in large but finite networks, considering both synthetically generated graphs and a real-world social network.
Results are obtained either by numerically solving the differential equation in \eqref{eq:zeta2},
or by running detailed Monte-Carlo simulations using an ad-hoc event-driven simulator.   

In particular, we will focus on the online social network Epinions, for which a popular snapshot is publicly
available at the Stanford Large Network Dataset Collection \cite{snapnets}.
The online social network Epinions.com is a consumer review website where the users can review different
kind of items with the purpose of rating hundred thousand products and ranking the reviewers to be trusted.
The available dataset contains the who-trust-whom relationships of all the members, operating from 1999 until 2014. 
The network consists of $|\V|=75879 $ nodes and $|\E|=508837$ directed edges, it is highly connected and contains cycles.
The average clustering coefficient is 0.1378.
The maximum in-degree is 3035, maximum out-degree is 1801, the average in/out-degree is 6.7.
In and out degrees follow an approximate power law distribution with exponent 1.7.

\subsection{TLTM on the Epinions graph}
We assume that all nodes have symmetric threshold $a^{\pm}_v = 2$.
In Figure \ref{fig:Epinions_contour}, we show on the left the loci of the stationary solution of \eqref{eq:zeta} in the plane $(\zeta_{-1}, \zeta_{1})$, and on the right an arrow plot representing the gradient of the system of the two ODEs in each possible point for which $\zeta_{-1} + \zeta_{1} \leq 1$. As it can be seen from the combinations of the two plots, there are three stable stationary points, which are located at $(\zeta_{-1} , \zeta_{1}) = (0,0)$, $(\zeta_{-1} , \zeta_{1}) = (\bar{\zeta},0)$ and $(\zeta_{-1} , \zeta_{1}) = (0,\bar{\zeta})$, with $\bar{\zeta} \simeq 0.91$ while there are four \emph{unstable} stationary points at values $(\zeta_{-1} , \zeta_{1}) = (\widetilde{\zeta},0)$, $(\zeta_{-1} , \zeta_{1}) = (0,\widetilde{\zeta})$, $(\zeta_{-1} , \zeta_{1}) = (\widetilde{\zeta_1},\widetilde{\zeta_1})$ and $(\zeta_{-1} , \zeta_{1}) = (\widetilde{\zeta_2},\widetilde{\zeta_2})$, where $\widetilde{\zeta} \simeq 1.02 \times 10^{-4}$, $\widetilde{\zeta}_1 \simeq 1.066 \times 10^{-4}$ and $\widetilde{\zeta}_2 \simeq 0.341$. The arrow plot allows also to verify that the boundary of the two basins of attraction for the stationary points coincides with the line $\zeta_{-1} = \zeta_{1}$ for $\widetilde{\zeta_1} \leq \zeta_{1} \leq  1/2$.  

\begin{figure}[htb]
\centering
\includegraphics[width=0.48\columnwidth,height=0.46\columnwidth]{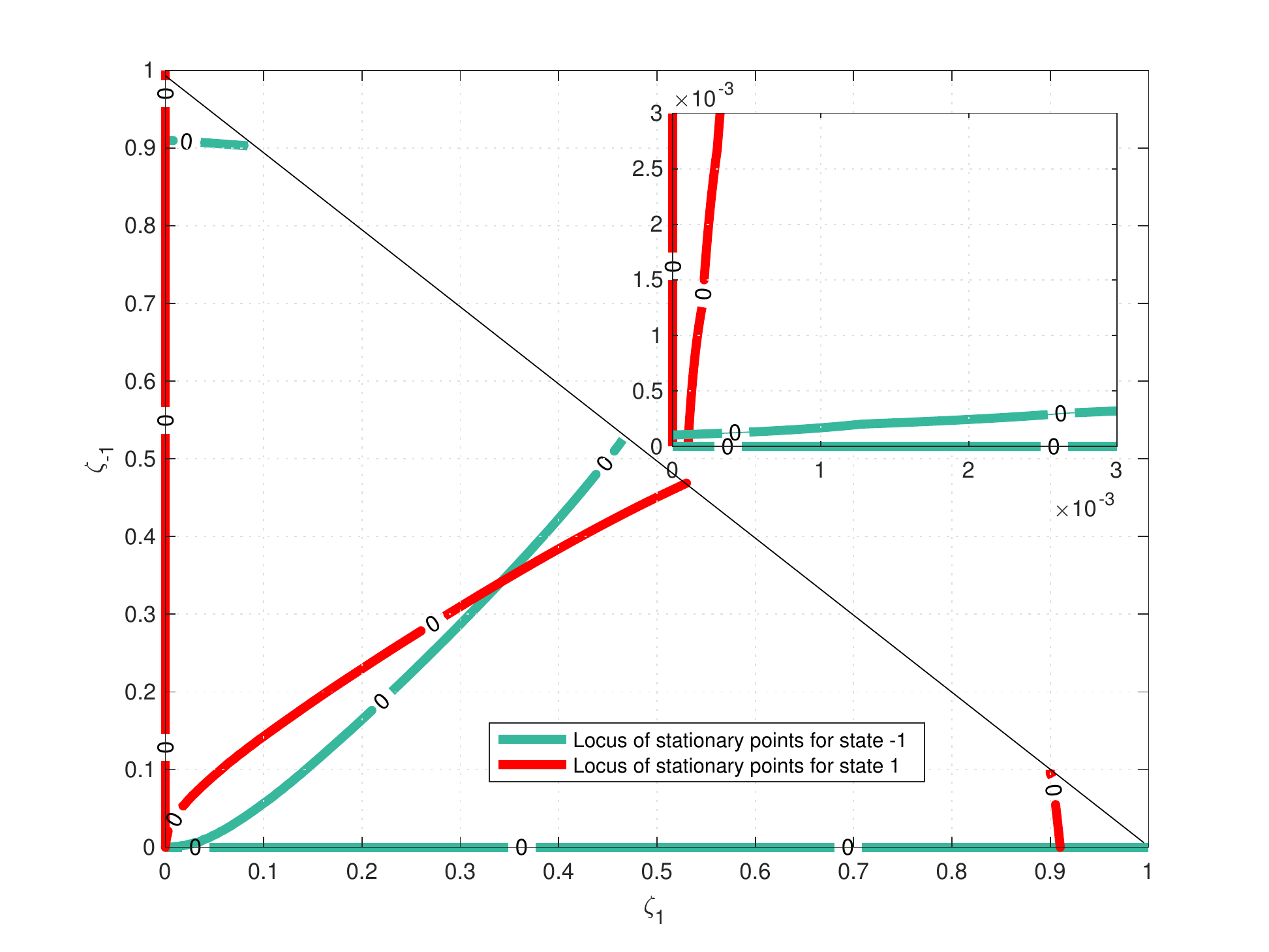}
\includegraphics[width=0.48\columnwidth,height=0.46\columnwidth]{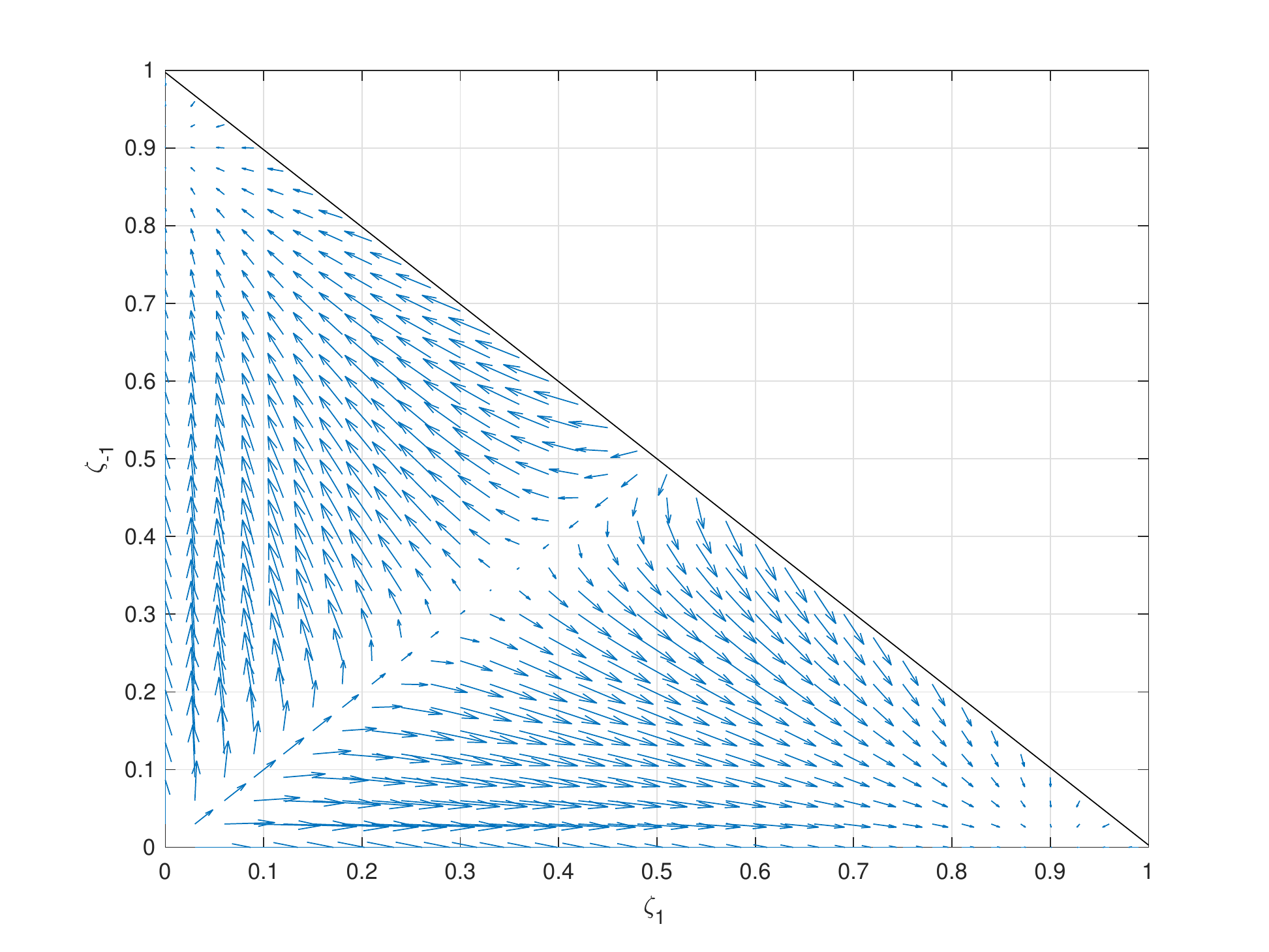}
\caption{TLTM in the Epinions social network with symmetric thresholds $r^{\pm} = 2$.\label{fig:Epinions_contour}: (a) Stationary solution of ODEs in \eqref{eq:zeta}. (b) Gradient of ODE system in \eqref{eq:zeta}}.
\end{figure} 

In Figure \ref{fig:Epinions_trans}, we show the evolution over time of the variables $\zeta_{\omega}(t)$ and $y_{\omega}(t)$, $\omega \in \{-1,0,1\}$, obtained through the numerical solution of ODEs in equations \eqref{eq:zeta}-\eqref{eq:zeta2}. In particular, $y_{\omega}(t)$ represents the fraction of nodes in state $\omega$ at time $t$, while $\zeta_{\omega}(t)$ represents the fraction of edges connected to a node in state $\omega$ at time $t$. At time $t = 0$, the fraction of nodes of any degree in states $-1,0,1$ is $0.3, 0.5, 0.2$, respectively. As it can be seen, the fraction of nodes in state $1$ decreases exponentially with time (the curve of $\zeta_1(t)$ is superimposed on that of $y_1(t)$). Instead, the curve for $\zeta_{-1}(t)$ reaches the fixed point $\overline{\zeta}  =0.91$, as predicted by Figure~\ref{fig:Epinions_contour}. The fixed point for $y_{-1}(t)$ is lower, at about $0.39$, implying that the fraction of nodes with higher degree that asymptotically reach state $-1$ is larger than for nodes with lower degree.

\begin{figure}[htb]
\centering
\includegraphics[width=0.48\columnwidth,height=0.46\columnwidth]{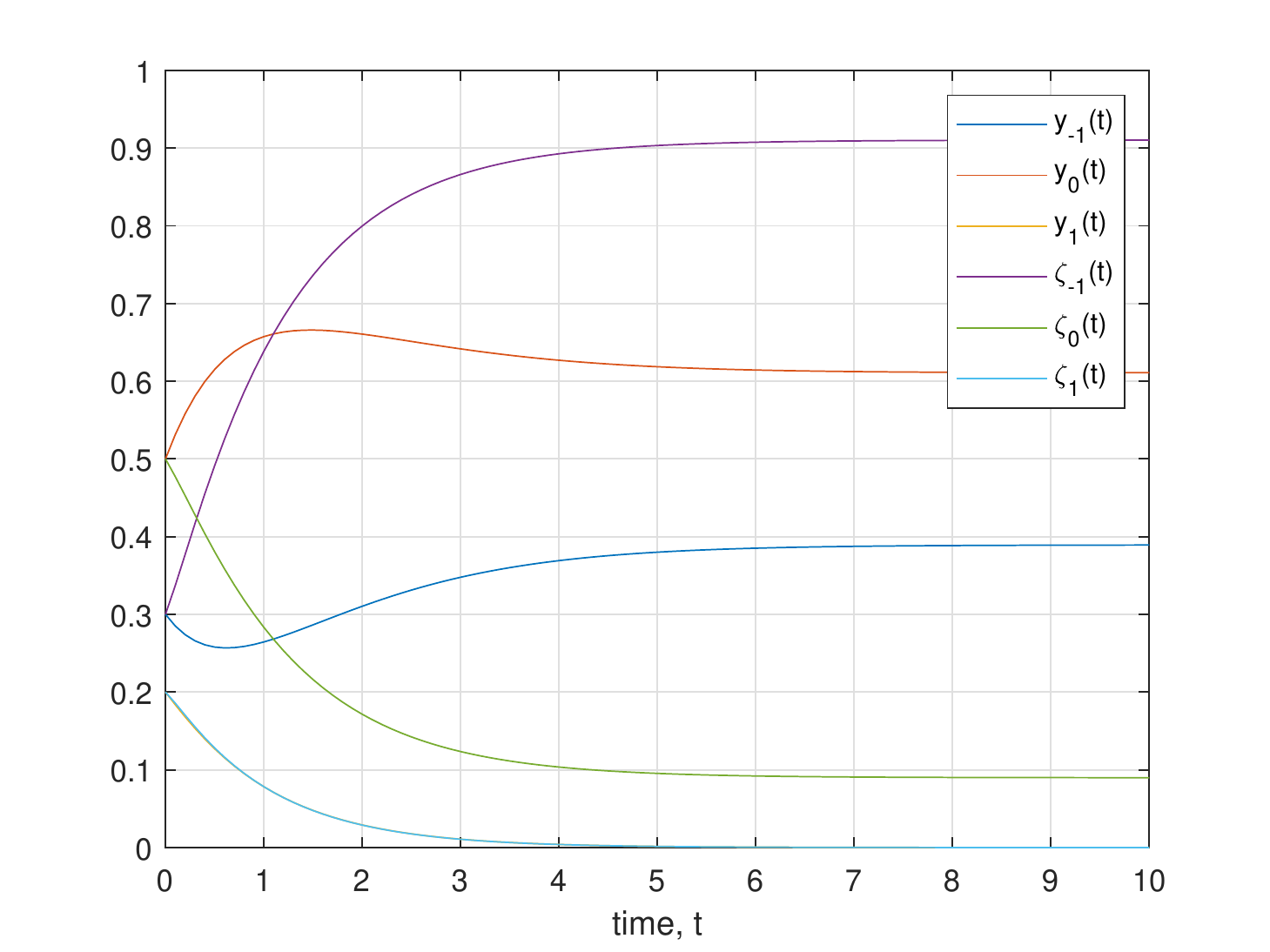}
\caption{TLTM in the Epinions social network with symmetric thresholds $a^{\pm}_v = 2$ \label{fig:Epinions_trans}: Transient solution of ODEs in \eqref{eq:zeta}-\eqref{eq:zeta2}}. 
\end{figure}

\subsection{TLTM on the Configuration Block Model}
In this subsection, we show results for a network with $n = 10^7$ nodes divided into two equal-size communities (classes), with size
$n/2 = 5 \cdot 10^6$. We consider the TLTM with symmetric thresholds
$a_v^{\pm} = 2$. We put $10^5$ seeds in state 1 in community 1, and $10^5$ seeds in state -1 in community 2.
The out-degree distribution of the first class is given by $p_{(k_1,k_2)|1} = p_{11}(k_1) p_{12}(k_2)$ where 
$$p_{11}(k_1) = {n/2-1 \choose k_1} p_A^{k_1} (1-p_A)^{1-k_1}$$ and $$p_{12}(k_2) = {n/2 \choose k_2} p_B^{k_2} (1-p_B)^{1-k_2}$$ $p_A$ and $p_B$ being chosen so that the average degree toward community 1 is 20, while the average degree toward community 2 is 6.
Analogously, the out-degree distribution of the second class is given by $p_{(k_1,k_2)|2} = p_{21}(k_1) p_{22}(k_2)$ where 
$$p_{21}(k_1) = {n/2 \choose k_1} p_C^{k_1} (1-p_C)^{1-k_1}$$ and $$p_{22}(k_2) = {n/2-1 \choose k_2} p_A^{k_2} (1-p_A)^{1-k_2}$$ $p_C$ being chosen so that the average degree toward community 1 is 5.
The network shows a slight asymmetry, since nodes in community 1 have slightly more  
edges directed towards nodes of community 2 than viceversa. 

Figure \ref{fig:sbm} compares the fraction of nodes in state 1 and -1 in either community,  
averaged across 100 simulation runs, against analytic results. We observe very good agreement between 
analysis (thin curves) and simulation (thick curves). Interestingly, in the beginning we have two 
weakly interfering percolation processes in the two communities, producing a significant 
increase of nodes in state -1 in community 2, and a significant increase of nodes in state 1 in community 1.
However, the percolation process in community 2 grows faster, because nodes in community 2 receive less
influence from nodes in community 1 than viceversa. As a consequence, the percolation process of nodes
in state -1 eventually invades also community 1, while nodes in state 1 vanish to zero throughout the network. 

\begin{figure}[htb]
\centering
\includegraphics[width=0.7\columnwidth]{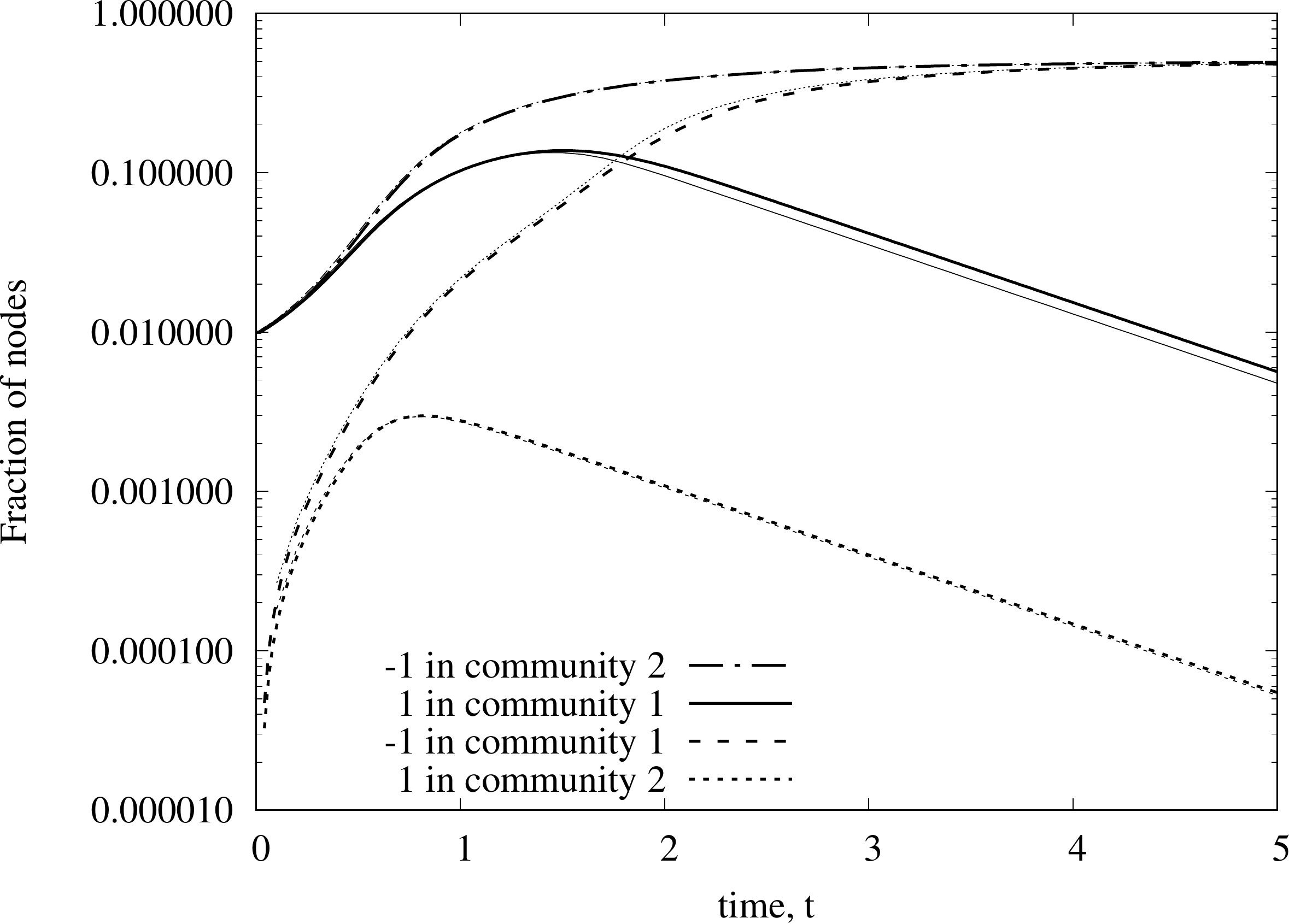}
\caption{Evolution over time of the fraction of nodes in states 1 and -1 in the two communities,
according to analysis (thin curves) and simulation (average of 100 runs) (thick curves).\label{fig:sbm}}
\end{figure}

\subsection{BRCA on the regular graph}
Here we consider a simple regular graph where all nodes have fixed out-degree $k = 21$ and fixed in-degree $d = 21$.
Note that, since the out-degree is odd, the best response is always deterministic (i.e., there are no ties).
We perform a single simulation run with $n=10^5$, with the following initial configuration: 30,000 coordinating nodes
in state 1, 10,000 coordinating nodes in state -1, 40,000 anti-coordinating nodes
in state 1, 20,000 anti-coordinating nodes in state -1.

\sidebyside{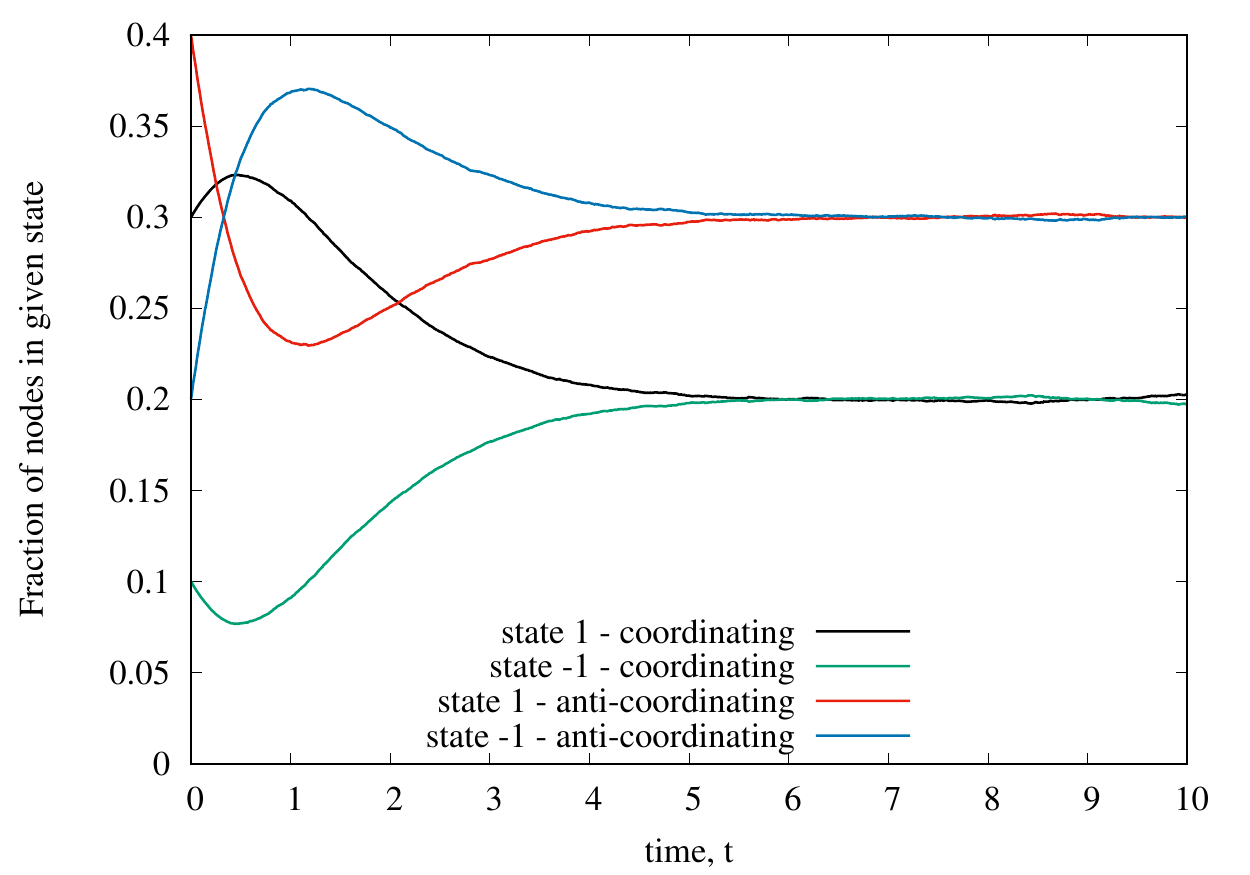}{Evolution over time of the fraction of nodes in the ABRD game, according to a single simulation run.}
{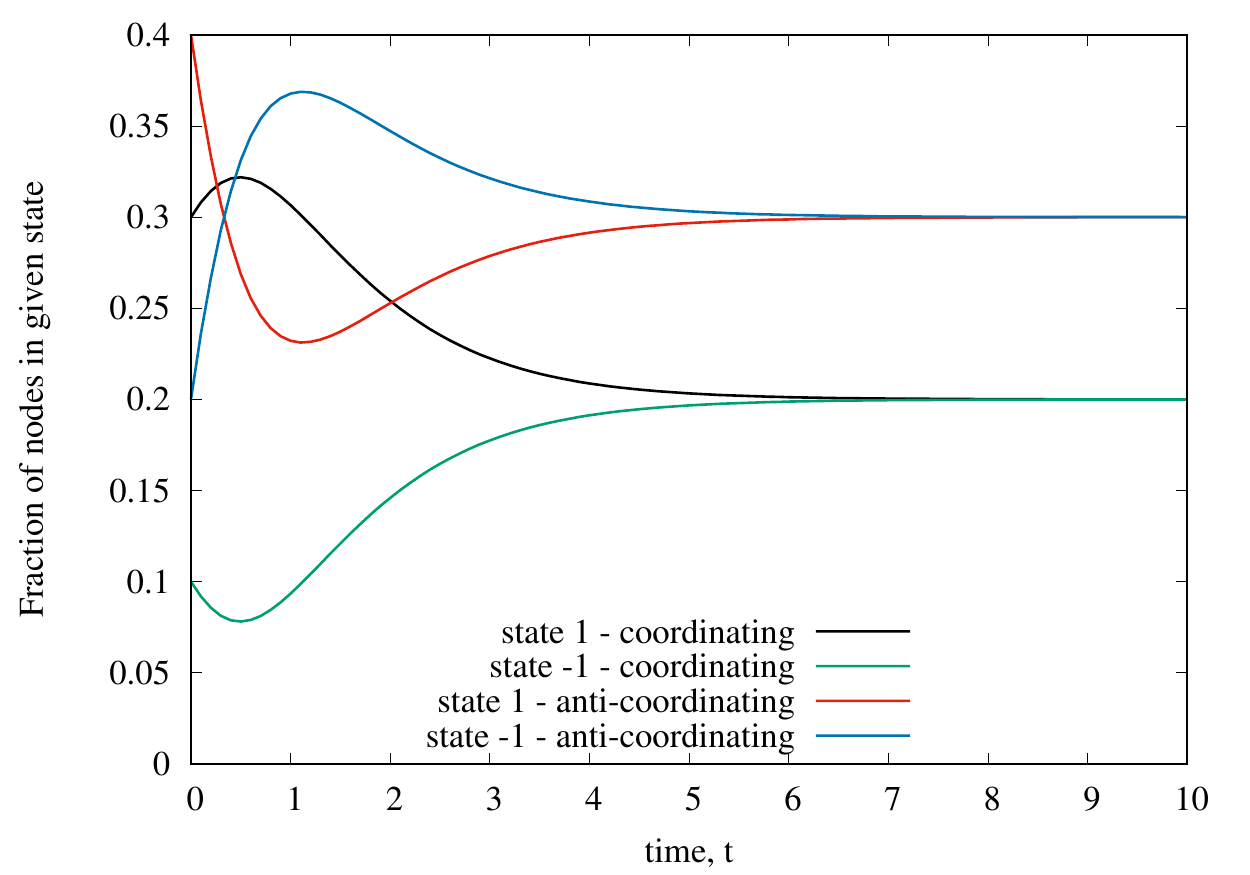}{Evolution over time of the fraction of nodes in the ABRD game, according to analysis.}

\begin{figure}[!ht]
\centering
\includegraphics[width=0.95\columnwidth]{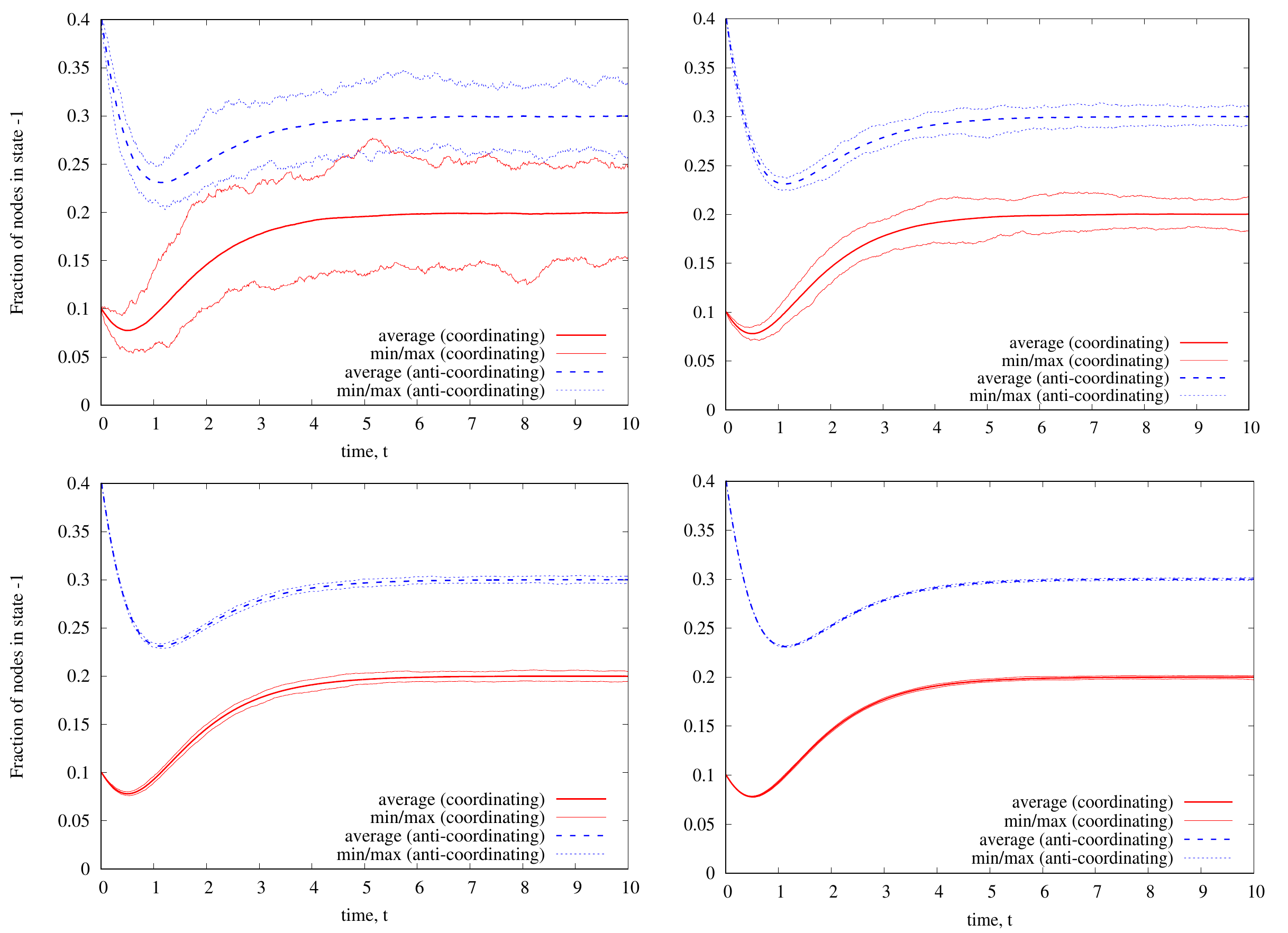}
\caption{Evolution over time of the fraction of nodes in state -1 in the ABRD game, 
across 400 simulation runs.\label{fig:combine}}
\end{figure}

In Figures \ref{fig:Figures/ABRD_sim} and \ref{fig:Figures/ABRD_mod} we show the fraction of nodes in each of the possible states
as function of time, according to simulation and analysis, respectively. We notice a perfect agreement between analytical 
prediction and simulation. Small fluctuations around the equlibrium configuration appear on the (single) simulation sample path.

To assess the degree of concentration of the process around its average, we carried out the following experiment: 
we performed 400 runs of the system, where the variability across runs is due to multiple reasons: i) the network topology generated by the configuration model; ii) the initial selection of nodes in the various states; iii) the temporal dynamics of the process (Poisson clocks). We then sampled the system with time granularity $\Delta t = 0.01$, and at each time instant we evaluated the average, the minimum, and the maximum fraction of nodes in each state, across the 400 runs.
In Figure \ref{fig:combine}
we show the results of the above experiment for the fraction of nodes in state -1 (either coordinating or anti-coordinating), with $n=10^3$ (top-left), $n=10^4$ (top-right), $n=10^5$ (bottom-left), $n=10^6$ (bottom-right).

Thin curves above and below the thicker line (denoting the average), correspond to maximum and minimum values.
Results for the fraction of nodes in state 1 are not shown, and they exhibit a similar variability.
We observe that results become more concentrated passing from $n=10^3$ to $n=10^6$ nodes. 
 
\subsection{BRCA on the Epinions graph}

We now investigate BRCA dynamics on the Epinions graph with a fraction of coordinating nodes equal to $\alpha = 0.7$, 
evenly distributed among nodes of any degree. Our main goal in this section will be to understand
better the origin of possible discrepancies between analysis and simulation. 

A numerical analysis of \eqref{eq:zeta}-\eqref{eq:zeta2} shows that, similarly to the regular case, the stationary points for the Epinions degree distribution are three, out of which the one in $\zeta_1 = 1/2$ is unstable. The other two are positioned at $\zeta_1 \simeq 0.33$ and $\zeta_1 \simeq 0.67$ and are stable. Such stationary points correspond to fractions of nodes in state 1 given by $y_1 \simeq 0.426 $ and $y_1 \simeq 0.574 $, respectively. 

Let us now move to the time evolution of fractions of nodes in a given state.
We consider the following initial condition: 42\% of coordinating nodes in state 1, 28\% of coordinating nodes in state -1,
10\% of anti-coordinating nodes in state 1, 20\% of anti-coordinating nodes in state -1.
\begin{figure}[htb]
\centering
\includegraphics[width=1.0\columnwidth]{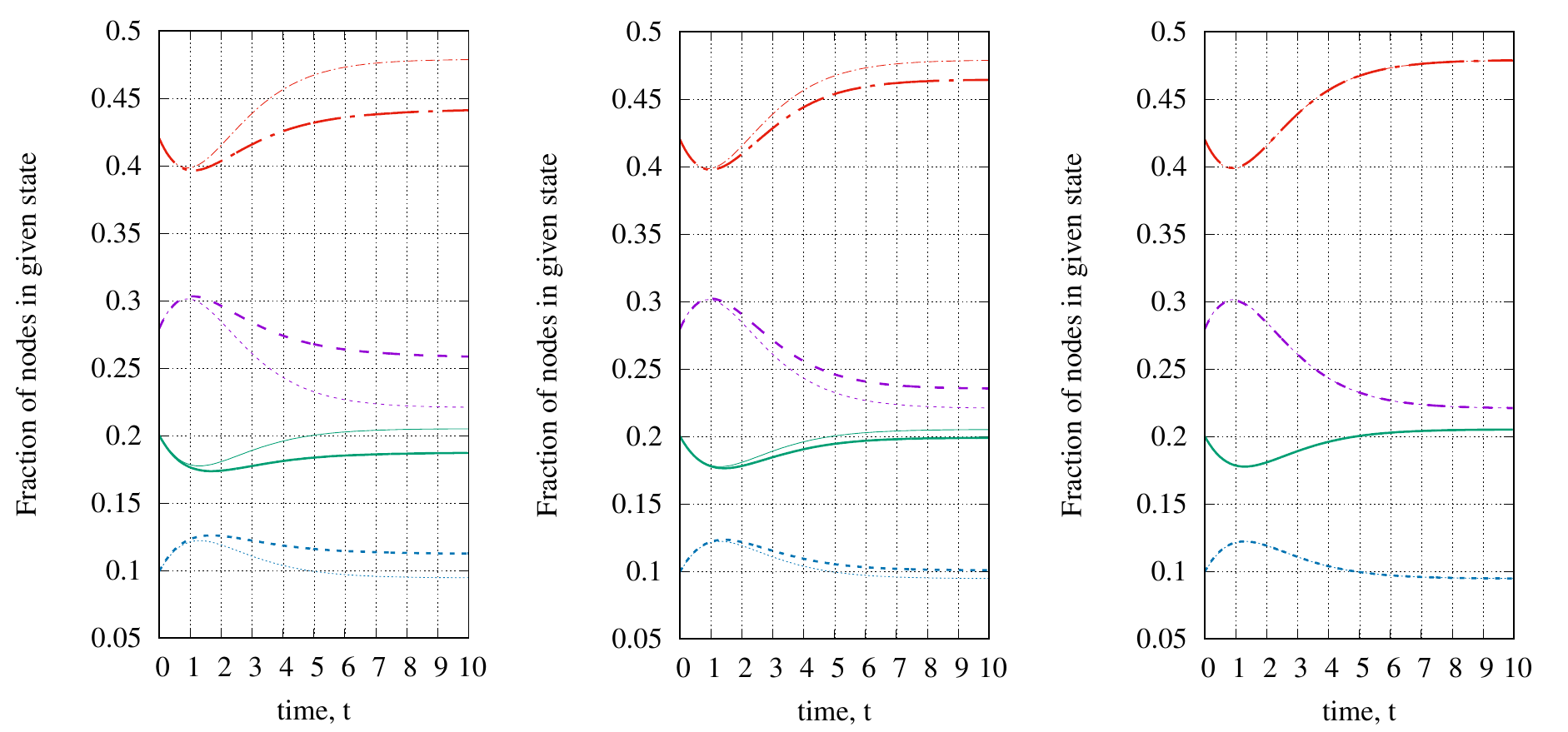}
\caption{BRCA dynamics: evolution over time of the fraction of nodes 
in each possible state, according to simulation (thick curves) and analysis (thin curves).
From top curve to bottom curve: coordinating nodes in state 1, coordinating nodes in state -1,
anti-coordinating nodes in state 1, anti-coordinating nodes in state -1.
From left to right plot: original Epinions graph (left); 
configuration model matched to the Epinions graph (middle); 
configuration model matched to the Epinions graph, with $10 n$ nodes (right).\label{fig:brca_transient}}
\end{figure}

The left plot in Figure \ref{fig:brca_transient} shows the fraction of nodes in each possible state as function of time, comparing
simulation (thick curves) and analysis (thin curves). For simulations, we have plotted the average of 400 runs, 
where in each run we randomly select a different seed set. 
We observe that, after a very similar initial behavior, simulations results tend to a different equilibrium point
with respect to analysis, as one can see by looking at the fraction of nodes at time $t=10$.  

We have identified two main reasons for the observed discrepancies: i) the first one due to the fact that 
the structure of the Epinions graph is not captured by the configuration model; ii) the second one 
due to the fact that the network size is not large enough to converge to a unique equilibrium point
across different runs, due to random effects.

To separate out the impact of the above two reasons, we have performed the following experiments. First,
we have run simulations in which, in each run, we randomly reshuffle the edges while maintaining the same
node statistics. Note that, by so doing, we generate graph according to the configuration model matched to the
Epinions graph. The results of this experiment are shown in the middle plot of Fig. \ref{fig:brca_transient}. As expected, now we observe 
a much better agreement between analysis and simulation. Still, there are non negligible discrepancies in the
final fraction of nodes in each possible state.

\begin{figure}[]
\centering
\includegraphics[width=1.0\columnwidth]{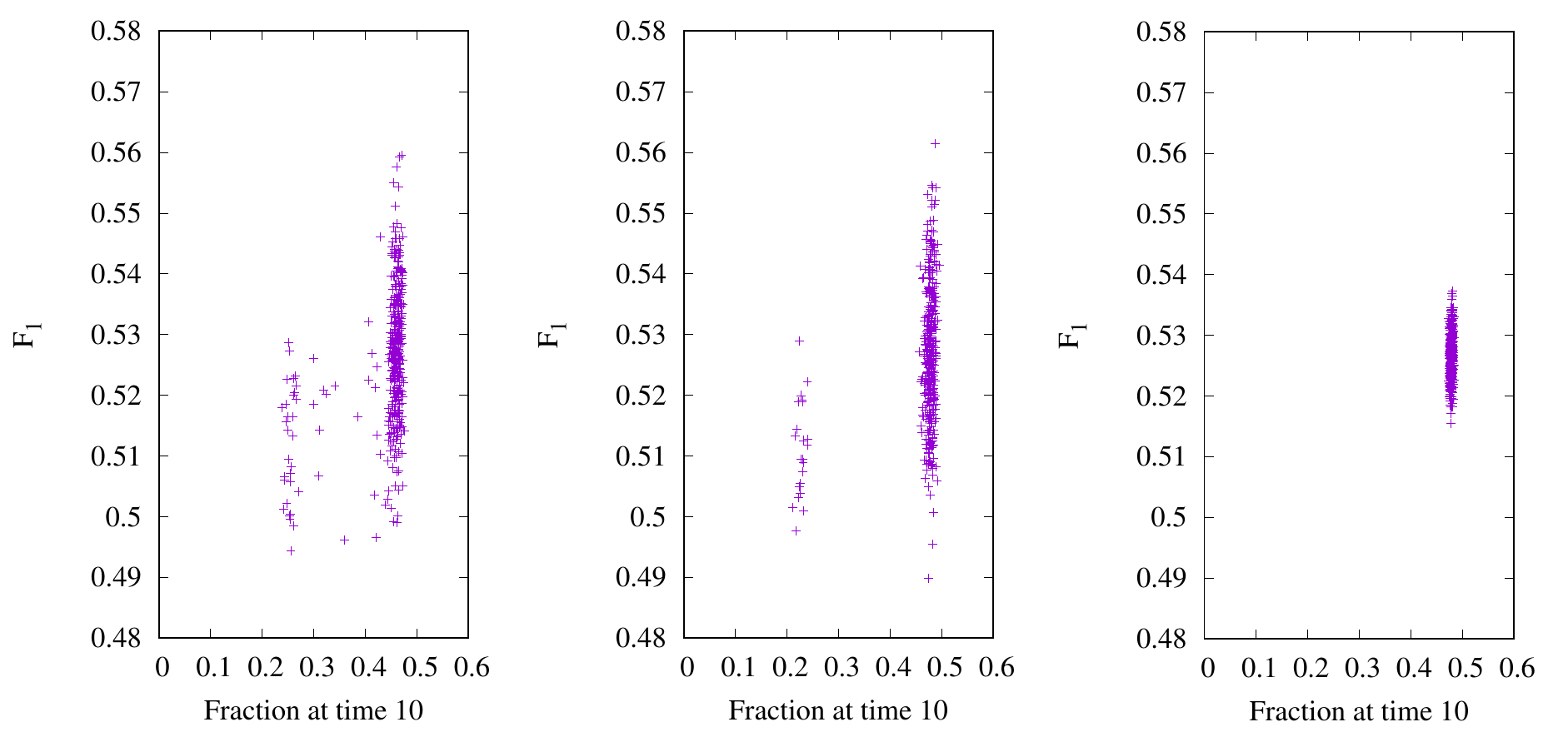}
\caption{$F_1$ metric vs fraction of coordinating node in state 1 (sampled at time $t=10$):
original Epinions graph (left); configuration model matched to the Epinions graph (middle); 
configuration model matched to the Epinions graph, with $10 n$ nodes (right).
\label{fig:Figures/brcaepi}}
\end{figure}
              An in-depth inspection of simulation results revealed that about 5\% of simulation runs
tend to a completely different equilibrium than the remaining 95\% of the runs. This fact is illustrated by
the middle plot in Figure \ref{fig:Figures/brcaepi}, where we have put a mark for each of the 400 runs, 
showing on the $x$ axes the fraction of coordinating nodes in state 1, sampled at time $t=10$.
For each run, we also computed the fraction of coordinating nodes that would transit to state $1$
if their clock would fire at time $t=0$. This fraction, denoted as $F_1$, is plotted on the $y$ axes,
and it is meant to capture a possible initial bias towards entering state 1, due to the initial 
network condition, which is especially dependent on the random seed allocation.

We can observe that simulations converge
to two main equilibria, and that simulations runs where the final fraction of coordinating
nodes in state 1 is smaller (around 0.2) have smaller values of the $F_1$ metric specified above, suggesting
that the initial seed allocation is the main responsible for driving the system into a different
configuration. We emphasize that a similar behavior can be observed also on the original 
Epinions graph, for which an analogous investigation of single simulation runs produced
the left plot in Fig. \ref{fig:Figures/brcaepi}. 
   
To remove the bi-stable outcome of simulations, we performed the following additional experiment: we considered
again the node statistics of the original Epinions graph, but this time we generated graphs
of size ten times larger than the original one (i.e., with $n = 758790$ nodes), using the configuration model.

The right plot in Fig. \ref{fig:Figures/brcaepi} shows that simulation results are 
now much more concentrated around a unique equilibrium. Moreover, 
in the right plot of Fig. \ref{fig:brca_transient} we observe an almost perfect agreement
between analysis and simulation for the evolution over time of the 
fraction of nodes in each possible state (in this plot thick and thin curves are 
essentially overlapped and thus indistinguishable).

We conclude that, when our analytic approach is used to predict 
the behavior of ASD dynamics on realistic (finite) graph, one must be aware of two 
main sources of errors: one due to the fact that real graphs are not completely
described by the configuration model; the other due to the fact that, in finite    
graphs, randomness can possibly drive the system to different equilibria, especially when initial 
conditions are close to the border of the attraction basin of the expected equilibrium.
However, our experiments with the Epinions graph suggest that our approach has remarkable
accuracy even in realistic graphs. Moreover, when the number of nodes exceeds, say, one million, 
results are sufficiently concentrated around their average to justify our mean-field approach, 
at least for the types of ASD dynamics that we have examined so far.

\subsection{ERG on the regular graph}
Here, we show numerical results for the ERG dynamics on a simple regular graph where all nodes have fixed out-degree $k = 50$ and fixed in-degree 
$d = 50$, and all nodes starts in the same state (the rock state). In simulation, we take 1 million nodes, and perform a single run of the system. In Figure \ref{fig:Figures/morra} we show the fraction of nodes in each of the three states as function of time, according to simulation.
In Figure~\ref{fig:Figures/morra_con}, in the plane representing the fraction of nodes in the paper and rock states, we show the loci of the stationary solutions for  \eqref{eq:zeta2} with $\omega \in \{R,P\}$. As shown in Section \ref{sec:ERG_analytic}, the only stationary solution for 
the Roshambo game is $y_{\omega} = 1/3$, $\omega \in \{R,P,S\}$, which is the only point lying in the intersection of the curves shown in Figure~\ref{fig:Figures/morra_con}.  Starting from $y_{R} = 1$, as it is shown in Figure~\ref{fig:Figures/morra}, there is a rapid increase in the fraction 
of nodes in the paper state, since it is the best response. When the node populations in the rock and paper states become 
of equal size, then nodes in the scissors state start to appear (in the plane, the trajectory deviates from the boundary $y_{R} + y_{P}= 1$). The three populations then tend to the equilibrium point in which they have the same size.

\sidebyside{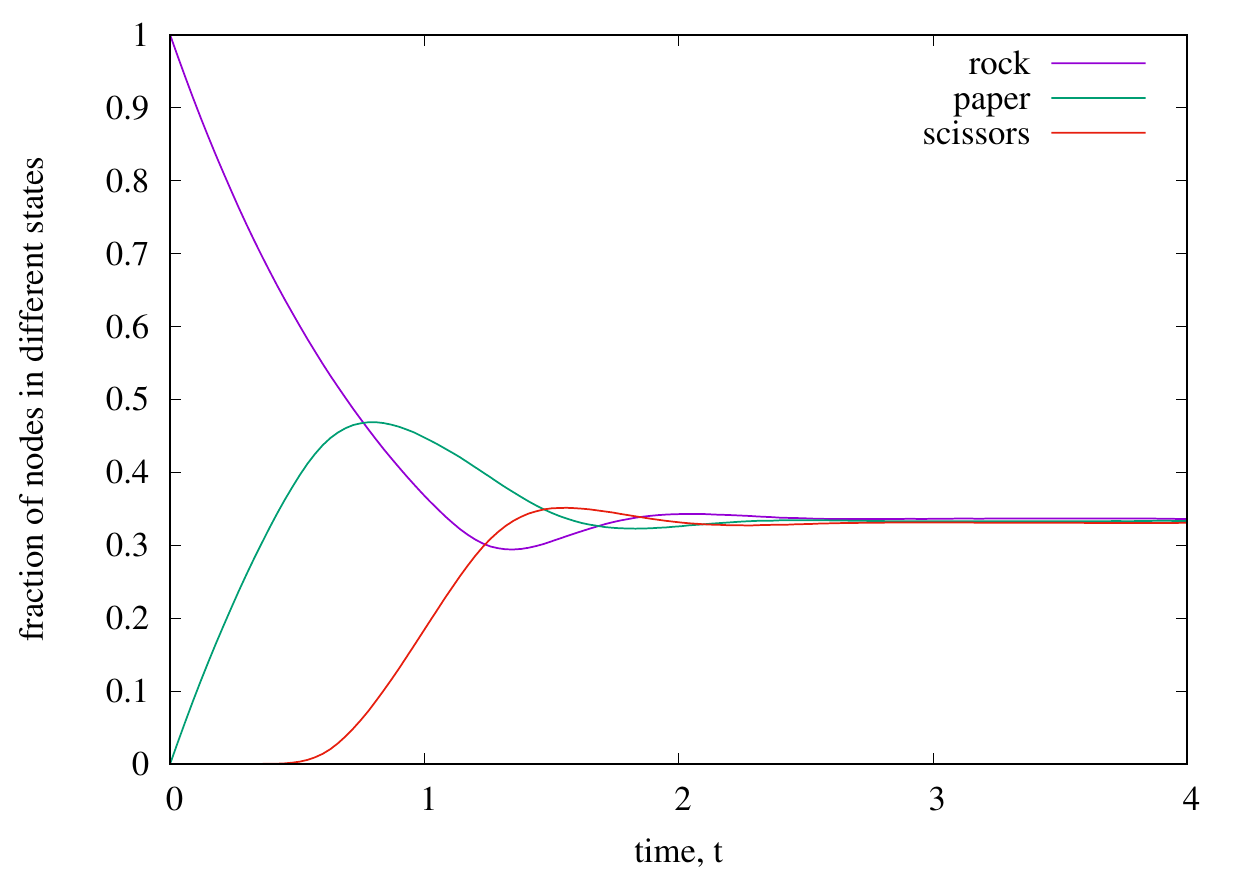}{Evolution over time of the fraction of nodes in states rock, paper, scissors, according 
to a single simulation run in a regular graph.} 
{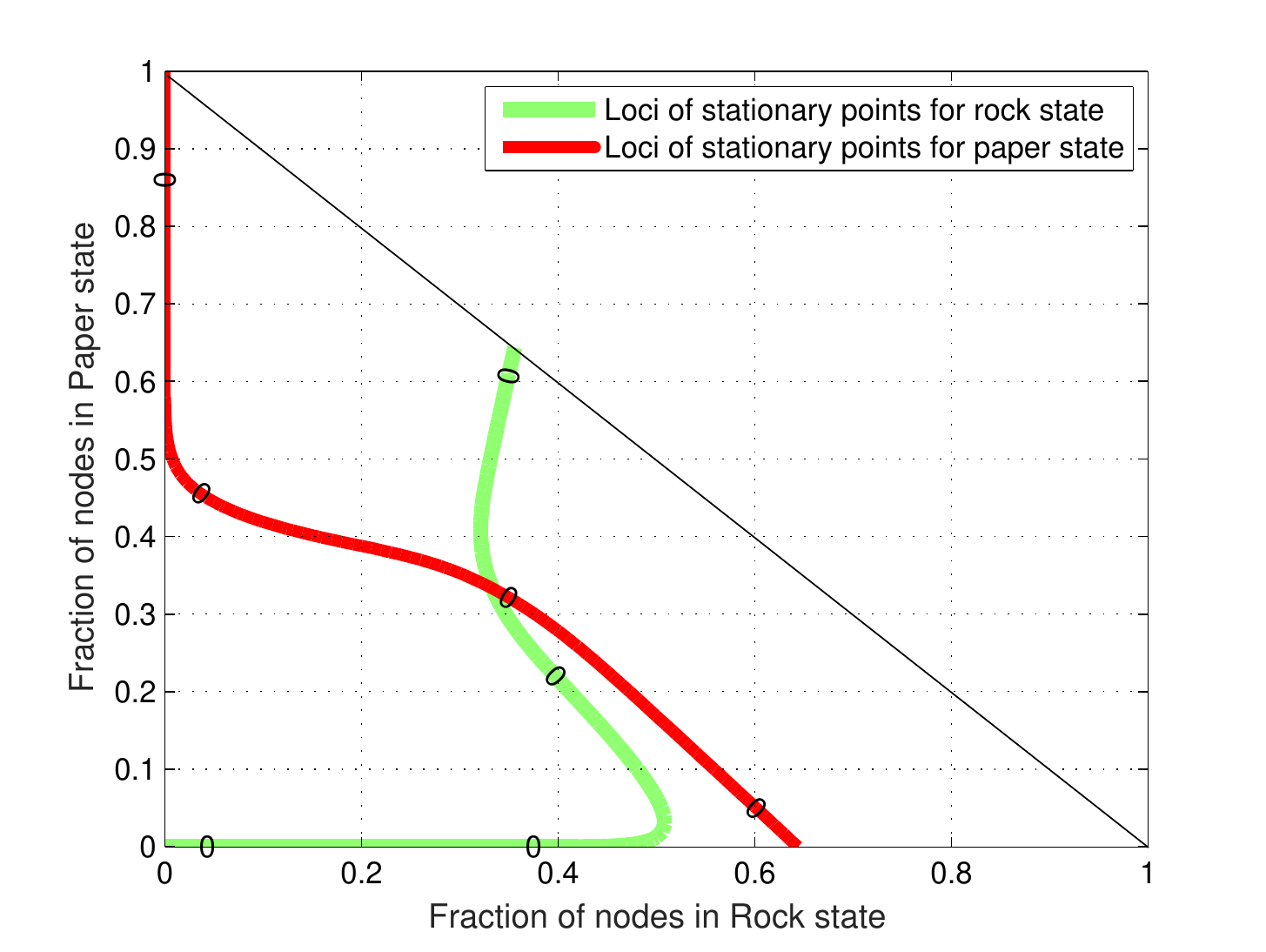}{Stationary solution of ODEs in \eqref{eq:zeta2} for 
the ERG dynamics in a regular graph.} 


\section{Conclusions}

In this paper we have proposed a mathematical framework showing that general semi-anonynoums dynamics 
in large scale random graphs converge to the solution of ordinary differential equations, allowing
fast numerical prediction of the transient behavior of many cascading processes in complex systems and,
in some cases, analytical estimation of their points of equilibrium.    
With respect to existing literature, we have extended the above mean-field approximation
along several directions: i) asynchronous node activation; ii) arbitrary semi-anonymous dynamics, 
including noisy best-response and class-dependent behavior; iii) general random graph exhibiting a 
local tree-structure, possibly mixing heterogeneous nodes and unbounded in/out degrees.
Our main contribution is a rigorous mathematical proof of convergence, which requires a careful 
combination of many independent results related to the different framework components.    
Despite the generality of our approach, we have not considered important variations 
of semi-anonynoums dynamics such as non-reversible transitions. 
Moreover, it remains still largely open how to analytically characterize in a tractable way 
the behavior of ASD in undirected network.

\appendix

\appendix
\section{Proofs of Section \ref{Sec:generalbounds}}

\subsection{Topological result: Proof of Theorem \ref{thm:top_res_modgraphs}}\label{app:Topological_result}
Let $\mathcal{N}_{t}$ be the relevant neighborhood of  a node  chosen uniformly at random from $\mathcal{V}$, respectively. Moreover, Let $\mathcal{T}_{t}$ be the truncated branching process at time $t$. 
In order to compare the distributions $\mu_{\mathcal{N}_t}$ and $\mu_{\mathcal{T}_t}$ of these two random variables, we define a  coupling between them.

\subsubsection*{Definition of coupling}

As a starting point, we  define  two different sequences of random variables.
For $a, a' \in \mathcal{A}$,  let  $(L^{a,a'}_{h} )_{h\in\mathbb{N}}$ be a sequence of i.i.d. random variables distributed according to a uniform distribution on the finite set $\mathcal{L}_{a,a'}$. Let $(M^{a,a'}_{h})_{h\in \mathcal{L}_{a,a'}}$ be a finite sequence of random variables such that
$$
\mathbb{P}(M^{a,a'}_{h}=L^{a,a'}_{h}|L^{a,a'}_{h}\notin\{M^{a,a'}_{1},\ldots,M^{a,a'}_{h-1}\})=1
$$
while $M^{a,a'}_{h}$  is uniformly distributed on the set $\mathcal{L}_{a,a'}\setminus \{M^{a,a'}_{1}, . . . , M^{a,a'}_{h-1}\}$, if $L^{a,a'}_{h} \in \{M^{a,a'}_{1},\ldots,M^{a,a'}_{h-1}\}$. 
Notice that the marginal distribution of the sequences $(L^{a,a'}_{h})_{h\in\mathbb{N}}$ and $(M^{a,a'}_{h})_{h\in \mathcal{L}_{a,a'}}$ are equivalent to sampling with replacement and sampling without replacement, respectively, from the set $\mathcal{L}_{a,a'}$.
We thus have
\begin{equation}\label{eq:eq_seq}
\mathbb{P}(L^{a,a'}_{h+1}\neq M^{a,a'}_{h+1}|(L^{a,a'}_{1},\ldots,L^{a,a'}_{h})=(M^{a,a'}_{1},\ldots,M^{a,a'}_{h}))=\frac{h}{l_{a,a'}}.
\end{equation}

We build  the  neighborhood  $\mathcal{N}_{t}$  with a dynamic exploration procedure  driven by  variables $(M^{a,a'}_{h})_{h\in\mathbb{N}}$. Our procedure   starts adding to $\mathcal{N}_t$
the root $v_0$ alone  (chosen uniformly at random over $\mathcal{V}$),   then $\mathcal{N}_t$ grows through the addition of new edges/nodes,  according to the following mechanism. Upon insertion, a newly added node is declared unexplored and its out-degree  is initialized  to $0$.  Then some of the unexplored nodes in $\mathcal{N}_t$ are  sequentially explored at random 
times (smaller than $t$). Upon its exploration, a node acquires a non-null out-degree and its out-neighbors are introduced in the network (if not already in the network).

Before describing in detail the algorithm, we introduce  the following variables:
\begin{itemize}
	\item $i$ denotes the iteration, which is determined by the number of explored  nodes in the network; 
	\item  $h_i^{a,a'}$  denotes  the number of $(a,a')$ edges at iteration  $i$;
	\item  $h_i$  denotes the  total number of edges in $\mathcal{N}_t$ at iteration $i$,  
	\item $\mathcal{S}_i$  denotes  the set of  nodes (except for the root $v_0$) in $\mathcal{N}_t$ at iteration $i$;
	\item $\mathcal{S}'_i$  denotes  the set of unexplored nodes in $\mathcal{N}_t$ at iteration $i$;
	\item  $T_i$ denotes  the time at which the $i$-th node activates (i.e., it is explored).
	\item  $\Gamma_i$  denotes  the time lag between the $(i-1)$-th and the $i$-th exploration.
	\item $V_i$  denotes the identity of the $i$-th explored node.
	\item $A_i$  denotes the  class of the $i$-th explored node.
	
\end{itemize}

Then $\mathcal{N}_{t}$ is generated  according to the following procedure:

\begin{enumerate} 
	\item  Set $i=0$. Start from the root $v_0$ (chosen uniformly at random from $\mathcal V$).  Assign it label $A_0=\lambda(v_0)$ and  out-degree vector ${\textbf{ J}}_0={\boldsymbol{ 0}}$, and declare it unexplored.
	Set $h_0^{a,a'}= 0$,  for all $a,a'\in\mathcal{A}$, $h_0=0$,  $\mathcal{S}_0=\mathcal{S}'_0=\emptyset$   and extract $T_0\sim\text{Exp}(1)$.
	
	\item If $T_0 > t$, stop the process. 
	If $T_0\le t$: declare $v_0$ explored, change its out-degree vector to ${\textbf{J}}_0={\textbf{K}}_{v_0}$, and add   $J_0=\sum_{a'\in\mathcal{A}}J_{0}^{a'}$ new edges  to  $\mathcal{N}_{t}$.
	In particular, for all $a'\in\mathcal{A}$, add exactly $J_{0}^{a'}$ out-going edges $M_{h'}^{A_0,a'}$, $h'\in\{1,\ldots J_{0}^{a'}\}$,  from $v_0$, pointing to nodes   $\nu_{A_0,a'}(M_{h'}^{A_0,a'})$. Note that such nodes are not necessarily different from each other. Newly introduced nodes  are declared unexplored and their
	 out-degree is set to $0$.
	 
	Then, counters are updated as:  $i=1$, $h_1^{A_0,a'}=J_{0}^{a'}$ for all $a'\in\mathcal{A}$ and $h_1^{a,a'}=0$ for $a\neq A_0$, $h_1=\sum_{a,a'} h_1^{a,a'}$,  $\mathcal{S}_1= \cup_{a', h'}\{ \nu_{A_0,a'}(M_{h'}^{A_0,a'}\}\}$ and $\mathcal{S}'_1 = \mathcal{S}_1$.   
	
	\item Let us define $T_i=T_{i-1}+\Gamma_i$ and let 
	$$\Gamma_i \sim \text{Exp}\left(|\mathcal{S}'_i|\right)$$
	
	where $|\mathcal{S}'_i|$ represents  the number of unexplored  nodes  in $\mathcal{N}_t$ at iteration $i$;
	and let $V_i$ be chosen uniformly at random from $\mathcal{S}'_i$.
	
	\item If $T_i > t$, stop the process. 
	If $T_i \le t$: declare the node  $V_i$ explored,  assign it a degree ${\textbf{J}}_{i}={\textbf{K}}_{V_i}$ 
	and add   $J_i=\sum_{a'\in\mathcal{A}}J_{i}^{a'}$ new edges  to  $\mathcal{N}_{t}$.
	In particular, add for all $a'\in\mathcal{A}$ exactly $J_{i}^{a'}$ out-going edges $M_{h'}^{A_i,a'}$  from $V_i$ pointing to nodes $\nu_{A_i,a'}(M_{h'}^{A_i,a'}) $
with $h'\in \{h^{A_i,a'}_i+1,\ldots,  h^{A_i,a'}_{i}+ J_{i}^{a'} \}$. Note that such nodes  are not necessarily all distinct, nor they are distinct from  other nodes already inserted in $\mathcal{N}_t$. Among them,
	 those that are not already in $\mathcal{N}_t$  are added and   declared unexplored and their out-degree is set to 0. 
	
	Then,  
	$h^{A_i,a'}_{i+1}= h_{i}^{A_i,a'}+J_{i}^{a'}$, 
	and, for all $a\neq A_i$,  $h^{a,a'}_{i+1}= h_{i}^{a,a'}$.
	Then  $h_{i+1}= h_{i}+\sum_{a'} J_{i}^{a'}$, $\mathcal{S}_{i+1}= \mathcal{S}_{i}\bigcup \{\cup_{a',h'}\{\nu_{A_i,a'}({\color{black}{M}}_{h'}^{A_i,a'})  \} $ and $\mathcal{S}'_{i+1}= \mathcal{S}'_{i} \backslash \{V_i\} \bigcup (\cup_{a',h'}\{\nu_{A_i,a'}({\color{black}{M}}_{h'}^{A_i,a'})  \}$.
	Finally, increment $i$ and go back to  Point 3.
	
\end{enumerate}

Note that by construction,  Point 4 is repeated for all $i \leq i_M$ where $i_M=\max\{i\geq 0| T_i \leq t\}$. 
Observe that the random network $\mathcal{N}_{t}$  generated in this way has, by construction,  the same structure 
and  desired distribution $\mu_{\mathcal{N}_{t}}$, of the relevant neighborhood of a random node in $\mathcal{G}$.
Let $E_t^{a,a'}$ be the total number of edges in $\mathcal{N}_{t}$ from nodes with label $a$ to nodes with label $a'.$
Notice that $E_t^{a,a'} = h_{i_M}^{a,a'}$.

Now we build  the tree $\mathcal{T}_{t}$  with a similar  dynamic exploration procedure  driven by  variables $(L^{a,a'}_{h})_{h\in\mathbb{N}}$,  which starts 
from a root $\widetilde{v}_0$ alone  and sequentially adds new edges/nodes $\widetilde{v}_h$,  $h=1,2,\dots$, to the tree.
Nodes $\{\widetilde{v}_h\}_{h\ge 0}$  are  assumed to be pairwise different. However a  correspondence  between    node $\{\widetilde{v}_h\}_h$  in the tree and  nodes   in the graph $\mathcal{G}$ is  dynamically established. We emphasize  that this correspondence  is, in general, non bijective: the same node in the network may correspond to several distinct nodes of $\mathcal{T}_t$, which are replicas of it.

More in detail, first we define  the following variables:
\begin{itemize}
	\item $i$ denotes the iteration, i.e., the number of activated nodes in the tree; 
	\item  $\widetilde{h}_i^{a,a'}$  denotes the number of $(a,a')$ edges at iteration  $i$;
	\item  $\widetilde{h}_i$  denotes the total number of edges in the tree at iteration $i$,   which, by construction, equals the total number of non-root nodes $S_i$ in the tree;
	\item $\widetilde{\mathcal{S}}_i$  denotes  the set of  nodes (except for the root $\widetilde{v}_0$) in $\mathcal{T}_t$ at iteration $i$;
	\item $\widetilde{\mathcal{S}}'_i$  denotes  the set of unexplored nodes in $\mathcal{T}_t$ at iteration $i$;
	\item  $\widetilde{T}_i$ denotes  the time at which the $i$-th node activates.
	\item  $\widetilde{\Gamma}_i$  denotes the time lag between the $(i-1)$-th and the $i$-th exploration.
	\item $\widetilde{V}_i$  denotes the identity of the $i$-th explored node.
	\item  $W(\cdot)$ represents the function that maps the nodes of the tree into $\mathcal V$.
    \item $\widetilde{A}_i$ is  the  class of the $i$-th activated node.	
	
\end{itemize}

Then $\mathcal{T}_{t}$ is generated  according to the following procedure:

\begin{enumerate} 
\item  Set $i=0$. Start from the root $\widetilde{v}_0$. Establish a correspondence between    $\widetilde{v}_0$ and $v_0$ (i.e., $W(\widetilde{v}_0)=v_0$).  Assign it label $\widetilde{A}_0=\lambda(v_0)$ and out-degree vector $\widetilde{{\textbf{J}}}_0={\boldsymbol{0}}$.
Set $\widetilde{h}_0^{a,a'}=\widetilde{h}_0=0$,  for all $a,a'\in\mathcal{A}$, $\widetilde{\mathcal{S}}_0 = \emptyset $ and $\widetilde{T}_0 = T_0$.

\item If $\widetilde{T}_0 > t$, stop the process. 
If $\widetilde{T}_0\le t$: change the out-degree vector of $\widetilde{v}_0$ to $\widetilde{{\textbf{J}}}_0={\textbf{K}}_{v_0}$, and add   $\widetilde{J}_0=\sum_{a'\in\mathcal{A}} \widetilde{J}_{0}^{a'}$ new nodes to  $\mathcal{T}_{t}$ (as children of $\widetilde{v}_0$).
In particular, for all $a'\in\mathcal{A}$, add exactly $\widetilde{J}_{0}^{a'}$ out-going edges from $\widetilde{v}_0$, pointing to different  nodes  $\widetilde{v}_h$, $h\in\{1,\ldots,\sum_{a'} \widetilde{J}_{0}^{a'}\}$. 
Establish a correspondence between newly inserted tree  nodes  and graph nodes as $W(\widetilde{v}_h)=\nu_{A_0,a'}(L_{h'}^{A_0,a'})$,  for $h= h'+\sum_{a'<a} \widetilde{J}_{0}^{a'}$. Set their out-degree to 0.
Then, counters are updated as:  $i=1$, $\widetilde{h}_1^{A_0,a'} = \widetilde{J}_{0}^{a'}$ for all $a'\in\mathcal{A}$ and $\widetilde{h}_1^{a,a'}=0$ for $a\neq A_0$, $\widetilde{h}_1=\sum_{a,a'} \widetilde{h}_1^{a,a'}$, $\widetilde{\mathcal{S}}_i = \cup_h \{\widetilde{v}_h\}$ and $\widetilde{\mathcal{S}}'_i = \widetilde{\mathcal{S}}_i$.   

\item Let us define $\widetilde{T}_i=\widetilde{T}_{i-1}+ \widetilde{\Gamma}_i$ with $\widetilde{\Gamma}_i = \Gamma_i$ if $| \widetilde{\mathcal{S}}'_i| = |\mathcal{S}'_i|$, otherwise 
 $$\Gamma'_i \sim \text{Exp}\left(|\widetilde{\mathcal{S}}'_i|\right);$$
moreover, let $\widetilde{V}_i = V_i$ if $\widetilde{\mathcal{S}}'_i  =\widetilde{\mathcal{S}}_i$, otherwise let $\widetilde{V}_i$ be chosen uniformly at random from $\widetilde{\mathcal{S}}'_i $. 

\item If $\widetilde{T}_i > t$, stop the process. If $\widetilde{T}_i \leq t$: declare the node  $\widetilde{V}_i$ explored,  assign it a degree $\widetilde{{\textbf{J}}}_{i}={\textbf{K}}_{W(\widetilde{V}_i)}$  
and  add $\widetilde{J}_i = \sum_{a'\in\mathcal{A}}\widetilde{J}_{i}^{a'}$ new nodes to $\mathcal{T}_{t}$   (as children of  $\widetilde{V}_i$). In particular, for all $a'\in\mathcal{A}$, add  exactly $\widetilde{J}_{i}^{a'}$ out-going edges from $\widetilde{V}_i$, pointing to different nodes $v_{\widetilde{h}_i+h}$, $h\in\{1,\ldots,\sum_{a'} \widetilde{J}_{i}^{a'}\}$. Establish a correspondence between newly inserted tree  nodes  and graph nodes as 
$W(\widetilde{v}_{\widetilde{h}_i+h})= \nu_{A_i, a'}(L_{\widetilde{h}_i^{A_i,a'}+h'}^{A_i,a'})$, for $h= h'+ \sum_{a'<a} \widetilde{J}_{i}^{a'}$. Set their out-degree to 0.
Then,   
$\widetilde{h}^{A_i,a'}_{i+1}=\widetilde{h}_{i}^{A_i,a'} + \widetilde{J}_{i}^{a'}$,
and for all $a\neq A_i$,  $\widetilde{h}^{a,a'}_{i+1}:=\overline{h}_{i}^{a,a'}$,  $\widetilde{h}_i= \widetilde{h}_{i-1}+\sum_{a'} \widetilde{J}_{i}^{a'}$, $\widetilde{\mathcal{S}}_{i+1}= \widetilde{\mathcal{S}}_{i}\bigcup (\cup_{h}\{\widetilde{v}_{\widetilde{h}_i+h}  \} $ and $\widetilde{\mathcal{S}}'_{i+1}= \widetilde{\mathcal{S}}'_{i} \backslash \widetilde{V}'_i \bigcup (\cup_{h}\{\widetilde{v}_{\widetilde{h}_i+h}  \})$.
Finally, increment $i$ and go back to  Point 3.
\end{enumerate}

\begin{figure}
\begin{center}\begin{tikzpicture} [every node/.style={circle,fill=gray!20,inner sep=4pt}]
  \node (n9) at (6.5,10) {\small$w_0$};
   \node (n18) at (7,8){\color{gray!20}\small$w_1$};
   \node (n20) at (6,6){\color{gray!20}\small$w_2$};
     \tikzset{every node/.style={}} 

    \node (n11) at (1,8) {\small$\nu_{++}(L_1^{++})$};
  \node (n12) at (3,8)   {$...$};
  \node (n13) at (4.5,8) {\small$\nu_{++}(L_{K^+_0}^{++})$};
   \node (n14) at (7,8){\small$\begin{array}{c}w_1=\nu_{+-}(L_1^{+-})\\
   A_1=-,{\textbf{K}_1}=(K_1^{+},K_1^{-})\end{array}$};
  \node (n15) at (9,8)   {$...$};
   \node (n17) at (11,8){\small$\nu_{+-}(L_{K_0^{-}}^{+-})$};

    \node (n1) at (2.5,6) {\small$\nu_{-+}(L_1^{-+})$};
  \node (n2) at (4,6)   {$...$};
  \node (n3) at (6,6) {\small$\begin{array}{c}w_2=\nu_{-+}(L_{K^+_1}^{-+})\\
  A_2=+,{\textbf{K}_2}=(K_2^{+},K_2^{-})
  \end{array}$};
   \node (n4) at (8,6){\small$\nu_{--}(L_1^{--})$};
  \node (n5) at (9.5,6)   {$...$};
   \node (n7) at (11,6){\small$\nu_{+-}(L_{K_1^{-}}^{--})$};

    \node (n1_last) at (2.5,4) {\small$\nu_{++}(L_{K_0^++1}^{++})$};
  \node (n2_last) at (4,4)   {$...$};
  \node (n3_last) at (6,4) {\small$\nu_{++}(L_{K^+_0+K_2^+}^{++})$};
   \node (n4_last) at (8.5,4){\small$\nu_{+-}(L_{K_0^-+1}^{+-})$};
  \node (n5_last) at (10,4)   {$...$};
   \node (n7_last) at (11.5,4){\small$\nu_{+-}(L_{K_0^{-}+K_2^-}^{+-})$};

\foreach \from/\to in {n9/n11,n9/n12,n9/n13,n9/n14,n9/n15,n9/n17}
\path (\from) edge[->,bend right=3] (\to);
\foreach \from/\to in {n14/n1,n14/n2,n14/n3,n14/n4,n14/n5,n14/n7}
\path (\from) edge[->,bend right=3] (\to);
\foreach \from/\to in {n3/n1_last,n3/n2_last,n3/n3_last,n3/n4_last,n3/n5_last,n3/n7_last}
\path (\from) edge[->,bend right=3] (\to);

\tikzset{mystyle/.style={->}} 
\tikzset{every node/.style={fill=cyan!10}} 

\tikzset{every node/.style={fill=cyan!10}} 
         \tikzset{every node/.style={}} 
  \node (n1) at (9.5,10) {$A_0=\lambda(w_0)=+, {\textbf{K}}_0=(K_0^+,K_0^{-})$};

     \tikzset{every node/.style={fill=cyan!10}} 
      \path (n9) edge[->,bend right=3, thick] node { $E_1$}(n14);
      \path (n14) edge[->,bend right=3, thick] node { $E_2$}(n3);

\end{tikzpicture}
\end{center}
\begin{caption}
{Notations: Example with two classes $\mathcal{A}=\{+,-\}$}
\end{caption}
\label{fig:Notation}
\end{figure}

Note that by construction,  Point 4 is repeated for all $i \leq \widetilde{\imath}_M$ where $\widetilde{\imath}_M=\max\{i \geq 0| \widetilde{T}_i \leq t\}$. 
Observe that the random tree $\mathcal{T}_{t}$  generated in this way has, by construction,  the same structure 
and  desired distribution $\mu_{\mathcal{T}_{t}}$, of the relevant neighborhood over a labeled branching process.
Let $\widetilde{E}_t^{a,a'}$ be the total number of edges in $\mathcal{T}_{t}$ from nodes with label $a$ to nodes with label $a'.$
Notice that $\widetilde{E}_t^{a,a'} = \widetilde{h}_{\widetilde{\imath}_M}^{a,a'}$.

The whole process of generation of both $\mathcal{N}_t$ and $\mathcal{T}_t$ is summarized in Table \ref{table:coupling}, where the parallelism of the generation of the two graphs is apparent.

\begin{table}
\begin{center}
\begin{tabular}{|p{6cm}|p{6cm}|}
\hline
Pick node $v_0$ uniformly at random from $\mathcal{V}$, and set its out-degree $\textbf{J}_0 = \boldsymbol{0}$. 
\newline Set $i := 0$. 
Set $h_0^{a,a'}:=0$  for all $a,a'\in\mathcal{A}$. 
Set $V_0 = v_0$.
& 
Let $\widetilde{v}_0 = v_0$ be the tree root and set its out-degree $\widetilde{\textbf{J}}_0 = \boldsymbol{0}$. 
\newline
Set $i := 0$. 
Set $\widetilde{h}_0^{a,a'}:=0$  for all $a,a'\in\mathcal{A}$. 
Set $\widetilde{V}_0 = \widetilde{v}_0$.  \\
\hline
Extract $T_0\sim\text{Exp}(1)$. 
& 
Set $\widetilde{T}_0 = T_0$. \\
\hline
For $i  =0,1,2,\dots $ & For $i  =0,1,2,\dots $ \\
\hline
IF $T_i \leq t$
\newline - assign $V_i$ the out-degree ${\textbf{J}}_{i}={\textbf{K}}_{V_i}$ and the class $A_i = \lambda(V_i)$; 
\newline - for all $a\in\mathcal{A}$, add $J_{i}^{a}$ out-going edges from $V_i$, pointing to nodes  
$\nu_{A_{i},a}(M_{h_i^{A_i,a}+j}^{A_i,a})$, with   $j \in\{ 1,\dots,  J_i^{a}\}$; 
\newline - assign all newly inserted nodes an out-degree equal to $\boldsymbol{0}$; 
\newline - update 
$h^{A_i,a'}_{i+1}:=h_{i}^{A_i,a'}+J_{i}^{a'}$,
and for all $a\neq A_i$ $h^{a,a'}_{i+1}:=h_{i}^{a,a'}$, $\mathcal{S}_{i+1}= \mathcal{S}_{i}\bigcup \{\cup_{a',h'}\{\nu_{A_i,a'}({\color{black}{M}}_{h'}^{A_i,a'})  \} $ and $\mathcal{S}'_{i+1}= \mathcal{S}'_{i} \backslash V_i \bigcup (\cup_{a',h'}\{\nu_{A_i,a'}({\color{black}{M}}_{h'}^{A_i,a'})  \}$.
\newline 
ELSE end.
 & 
IF $\widetilde{T}_i \leq t$ 
\newline 
- assign $\widetilde{V}_i$ the out-degree $\widetilde{{\textbf{J}}}_{i}={\textbf{K}}_{W(\widetilde{V}_i)}$ and the class $\widetilde{A}_i = \lambda(W(\widetilde{V}_i))$; 
\newline 
- for all $a\in\mathcal{A}$, add $\widetilde{J}_{i}^{a}$ out-going edges from $\widetilde{V}_i$, pointing to $\widetilde{J}_{i}^{a}$ new nodes, which are replicas of 
$\nu_{A_i,a}(L_{\widetilde{h}_i^{A_i,a}+j}^{A_i,a})$, with   $j \in \{1,\dots, \widetilde{J}_{i}^{a}\}$;
\newline  
- assign all newly inserted nodes an out-degree equal to $\boldsymbol{0}$;
\newline
- update 
$\widetilde{h}^{A_i,a'}_{i+1}:=\widetilde{h}_{i}^{A_i,a'}+\widetilde{J}_{i}^{a'}$,
and for all $a\neq A_i$ $\widetilde{h}^{a,a'}_{i+1}:=\widetilde{h}_{i}^{a,a'}$, $\widetilde{\mathcal{S}}_{i+1}= \widetilde{\mathcal{S}}_{i}\bigcup (\cup_{h}\{\widetilde{v}_{\widetilde{h}_i+h}  \} $ and $\widetilde{\mathcal{S}}'_{i+1}= \widetilde{\mathcal{S}}'_{i} \backslash \widetilde{V}'_i \bigcup (\cup_{h}\{\widetilde{v}_{\widetilde{h}_i+h}  \}$.
\newline
ELSE end. \\
\hline
Generate $T_{i+1} = T_{i} + \Gamma_{i+1}$, where
$\Gamma_{i+1} \sim \text{Exp} \left(|\mathcal{S}'_{i+1}|\right)$ 
\newline
Sample node $V_{i+1}$ uniformly at random from $\mathcal{S}'_{i+1}$.
& 
Generate $\widetilde{T}_{i+1} = \widetilde{T}_{i} + \widetilde{\Gamma}_{i+1}$, where
$\widetilde{\Gamma}_{i+1} =\Gamma_{i+1}$, if $|\widetilde{\mathcal{S}}'_{i+1}| = |\mathcal{S}'_{i+1}|$ and $\widetilde{\Gamma}_{i+1} \sim \text{Exp} \left(|\widetilde{\mathcal{S}}'_{i+1}|\right)$ otherwise. 
\newline
Sample node $\widetilde{V}_{i+1}$ as follows. If $\widetilde{\mathcal{S}}'_{i+1} = \mathcal{S}'_{i+1}$ then $\widetilde{V}_{i+1} = V_{i+1}$, else $\widetilde{V}_{i+1}$ is taken uniformly at random from $\widetilde{\mathcal{S}}'_{i+1}$. 
\\
\hline
\end{tabular}
\end{center}
\caption{Coupling between the generation processes of \protect$\mathcal{N}_{t}$ (left column) and \protect$\mathcal{T}_{t}$ (right column)}
\label{table:coupling}
\end{table}

\medskip

\subsubsection{Proof of Theorem \ref{thm:top_res_modgraphs}}
Using the coupling inequality (see Proposition 4.7 in \cite{LevinPeresWilmer2006}) we have 
\begin{equation}\label{eq:TV}\|\mu_{\mathcal{T}_{t}}-\mu_{\mathcal{N}_{t}}\|_{\mathrm{TV}}\leq \mathbb{P}(\mathcal{N}_t\neq \mathcal{T}_t).\end{equation}
Define the two events
\[
\mathcal{ B}_1=\bigcap_{a,a' \in \mathcal{A}}\left\{(L^{a,a'}_1,L^{a,a'}_2,\ldots,L^{a,a'}_{\widetilde{E}_t^{a,a'}})=(M^{a,a'}_1,M^{a,a'}_2,\ldots,M^{a,a'}_{\widetilde{E}_t^{a,a'}})\right\}
\]
and 
\begin{align*}\mathcal{ B}_2&=\bigcap_{a\in\mathcal{A}}\left\{\bigcap_{(b,h)\neq (b',h'):b,b'\in\mathcal{A},h\in\{1,\ldots,\widetilde{E}_t^{b,a}\},h'\in\{1,..., \widetilde{E}_t^{b',a}\}}\{\nu_{b,a}({\color{black}{L^{b',a}_h}}))\neq \nu_{b',a}({\color{black}{L^{b,a}_h}}))\}\right\}\bigcap\\
&\qquad\qquad\qquad\bigcap\left\{\bigcap_{b\in\mathcal{A}}\bigcap_{h\in\{1,\ldots,\widetilde{E}_t^{b,a_0}\}}\{\nu_{b,A_0}(W(\widetilde{v}_h))\neq v_0 \}\right\},\end{align*}
which are, in words, the event that there are no repeated edges in $\mathcal{T}_t$ and that the map $W(\cdot)$ is bijective (i.e., just a single node in 
$\mathcal{T}_t$ corresponds to every node in 
$\mathcal{N}_t$). 

We are going to show that  $\{\mathcal{ B}_1\cap \mathcal{ B}_2  \} \subseteq \{\mathcal{N}_t= \mathcal{T}_t \}$. 
The assertion, indeed,  can be  easily checked by induction over the iteration $i$. First observe that at the end of  iteration $0$,  by construction, under $\mathcal{ B}_1$ and $\mathcal{ B}_2$  the structure of $\mathcal{T}_t$ and $\mathcal{N}_t$ are necessarily equal. Indeed by construction they can be different only if either  some $L^{A_0,a'}_{h'}\neq ,M^{A_0,a'}_{h'}$
for $h'\le J_0^{A_0,a'}= \tilde{J}_0^{A_0,a'}$ or there are $h'$ and $h''$ such that $\nu_{A_0,a}(L^{A_0,a'}_{h'})=\nu_{A_0,a'}(M^{A_0,a'}_{h'})= \nu_{A_0,a'}(M^{A_0,a'}_{h''})= \nu_{A_0,a'}(L^{A_0,a'}_{h''})$.
Now suppose that the structure of $\mathcal N_{t}$ is equal to the structure of $\mathcal{T}_t$ at the end of iteration $i-1$ (for $i\ge 1$). Then, by construction $\mathcal{S}_i=\widetilde{\mathcal{S}}_i$ and   $\mathcal{S'}_i=\widetilde{\mathcal{ S}}'_i$  $ h_i^{a,a'}=\widetilde{h}_i^{a,a'}$; therefore $V_i=\widetilde{V}_i$ and  $\Gamma_i=\widetilde{\Gamma}_i$ and $A_i=\widetilde{A}_i$ $J_i^{A_i,a'}=\widetilde{J}_i^{A_i,a'}$ . During iteration $i$ we add to $\mathcal{N}_t$ 
nodes 
$\nu_{A_i,a'}(M^{A_i,a'}_{h_i^{A_i,a'}+h'})= \nu_{A_i,a'}(L^{A_i,a'}_{h_i^{A_i,a'}+h'})$ for $h'\in\{1,\ldots,  J_i^{a'} \}$  (where the equality descends from ${\mathcal{B}}_1$), which, from ${\mathcal{B}}_2$, are all different and different from nodes already in $\mathcal{N}_t$. In $\mathcal{T}_t$ we add  brand-new  replicas of 
nodes  $\nu_{A_i,a'}(L^{A_i,a'}_{h_i^{A_i,a'}+h'})= \nu_{A_i,a'}(M^{A_i,a'}_{h_i^{A_i,a'}+h'})$. Therefore the structures of $\mathcal{N}_t$ and 
$\mathcal{T}_t$ are  still equal at the end of iteration $i$.

Thus:
\begin{align} \label{eq:notequaltree}
\mathbb{P}\left(\mathcal{N}_t\neq \mathcal{T}_t\right)&\leq\mathbb{P}\left(\mathcal{ B}_1^C\cup\mathcal{ B}_2^C\bigg|\bigcap_{b,a}\left\{\widetilde{E}_t^{b,a} \leq x_{b,a}\right\}\right)+\mathbb{P}\left(\bigcup_{b,a}\left\{\widetilde{E}_t^{b,a} > x_{b,a}\right\}\right)\nonumber\\& \leq \mathbb{P}\left(\mathcal{ B}_1^C\bigg|\bigcap_{b,a}\left\{\widetilde{E}_t^{b,a} \leq x_{b,a}\right\}\right)+\mathbb{P}\left(\mathcal{\tilde B}_2^C\bigg|\mathcal{\tilde B}_1,\bigcap_{b,a}\left\{\widetilde{E}_t^{b,a} \leq x_{b,a}\right\}\right)\nonumber\\
&+\sum_{b,a}\mathbb{P}\left(\left\{\widetilde{E}_t^{b,a} > x_{b,a}\right\}\right)\nonumber\\
\end{align}
where in the first inequality we have let out the probability that the number of nodes exceeds a fixed threshold.

Define the event
\begin{equation}
\mathcal{E}_h^{b,a} = \left\{ (L^{b,a}_1,\ldots,L^{b,a}_{h}) = (M^{b,a}_1,\ldots,M^{b,a}_{h}))\right\}
\end{equation}

The first term of \eqref{eq:notequaltree} is upper bounded by
\begin{align}\mathbb{P} 
\left(\mathcal{ B}_1^C\bigg|\bigcap_{b,a}\left\{\widetilde{E}_t^{b,a} \leq x_{b,a}\right\}\right)
&\ =\mathbb{P}\left(\bigcup_{b,a\in\mathcal{A}} \left(\mathcal{E}_{\widetilde{E}_t^{b,a}}^{b,a}\right)^c 
\bigg|\bigcap_{b,a}\left\{\widetilde{E}_t^{b,a} \leq x_{b,a}\right\}\right)\nonumber\\
&\ \leq \sum_{b,a\in\mathcal{A}} \mathbb{P}\left(\left(\mathcal{E}_{\widetilde{E}_t^{b,a}}^{b,a}\right)^c \bigg| \widetilde{E}_t^{b,a} \leq x_{b,a}\right)\nonumber\\
&\ \leq\sum_{b,a\in\mathcal{A}}\sum_{h_{b,a}=0}^{x_{b,a}-1}\mathbb{P}(L^{b,a}_{h_{b,a}+1}\neq M^{b,a}_{h_{b,a}+1}| \mathcal{E}_{h_{b,a}}^{b,a})\nonumber\\
&=\sum_{ b,a\in\mathcal{A}}\sum_{h_{b,a}=0}^{x_{b,a}-1}\frac{h_{b,a}}{l_{b,a}}\label{eq:bound_A11}
\end{align}
where the first inequlity is the union bound, the second inequality is the chain rule, while the last equality comes from  \eqref{eq:eq_seq}.
Using the same arguments, the second term of \eqref{eq:notequaltree} becomes
\begin{align}\mathbb{P}
&\left(\mathcal{B}_2^C\bigg| \mathcal{B}_1,\bigcap_{b,a}\left\{\widetilde{E}_t^{b,a} \leq x_{b,a}\right\}\right)\nonumber\\
&\ \leq\sum_{a,b\in\mathcal{A}}\sum_{h_{b,a}=1}^{x_{b,a}}
\mathbb{P}(\nu_{b,a}(M^{b,a}_{h_{b,a}})\in\{w_0,\nu_{b,a}(M^{b,a}_1),...,\nu_{b,a}(M^{b,a}_{h_{b,a}-1})\}|\mathcal{E}_{h_{b,a}}^{b,a})\nonumber\\
& + \sum_{a,b\in\mathcal{A}} \sum_{b'\neq b} \sum_{h_{b,a}=1}^{x_{b,a}} \sum_{h_{b',a}=1}^{x_{b',a}} \mathbb{P}\left(\nu_{b,a}(M^{b,a}_{h_{b,a}}) = \nu_{b',a}(M^{b',a}_{h_{b',a}}) \bigg| \mathcal{E}_{h_{b,a}}^{b,a},\mathcal{E}_{h_{b',a}}^{b',a}\right)
\end{align}
where the first term above gives the probability that two edges from the same class point to the same node or that a given edge points to the root, while the second term computes the probability that two edges from two different classes point to the same node. From the definition in \eqref{eq:statistics_edges}, we have that the first term is upper bounded by
\begin{align}
&\sum_{a,b\in\mathcal{A}}\sum_{h_{b,a}=1}^{x_{b,a}}
\mathbb{P}(\nu_{b,a}(M^{b,a}_{h_{b,a}})\in\{w_0,\nu_{b,a}(M^{b,a}_1),...,\nu_{b,a}(M^{b,a}_{h_{b,a}-1})\}|\mathcal{E}_{h_{b,a}}^{b,a})\nonumber\\
&\leq\sum_{a,b\in\mathcal{A}}\sum_{h_{b,a}=1}^{x_{b,a}}\frac{(h_{b,a}-1)\sum_{{\textbf{d}},{\textbf{k}}}(d_b-1)q^{b}_{{\textbf{d}},{\textbf{k}}|a}+\sum_{{\textbf{d}},{\textbf{k}}}d_bq^{b}_{{\textbf{d}},{\textbf{k}}|a} }{l_{b,a}} \nonumber\\
&\leq\sum_{a,b\in\mathcal{A}} \frac{x_{b,a}(x_{b,a}+1)}{2}\frac{\sum_{{\textbf{d}},{\textbf{k}}}d_bq^{b}_{{\textbf{d}},{\textbf{k}}|a}}{l_{b,a}} -\sum_{ a,b\in\mathcal{A}}\sum_{h=0}^{x_{b,a}-1}\frac{h}{l_{b,a}}.
\end{align}
Using similar arguments, we get
\begin{align}&\sum_{a,b\in\mathcal{A}} \sum_{b'\neq b} \sum_{h_{b,a}=1}^{x_{b,a}} \sum_{h_{b',a}=1}^{x_{b',a}} \mathbb{P}\left(\nu_{b,a}(M^{b,a}_{h_{b,a}}) = \nu_{b',a}(M^{b',a}_{h_{b',a}}) \bigg| \mathcal{E}_{h_{b,a}}^{b,a},\mathcal{E}_{h_{b',a}}^{b',a}\right)\nonumber\\
& \leq \sum_{a,b\in\mathcal{A}} \sum_{b'\neq b} x_{b,a}x_{b',a} \frac{\sum_{{\textbf{d}},{\textbf{k}}}d_bq^{b'}_{{\textbf{d}},{\textbf{k}}|a} }{l_{b,a}}.
\label{eq:bound_A12}\end{align}

Combining these bounds and inequalities in \eqref{eq:TV} and \eqref{eq:bound_A11} we conclude the thesis.

\subsection{Concentration property: Proof of Theorem \ref{thm:concentration_modgraphs}}\label{appA_concentration}
Before presenting the proof of the main result we fix some notations and we state some preliminary lemmas.
%
%
First, we recall a simple variant of Azuma's inequality which will be useful in our arguments. 
Let $\{Y_k : k = 0, 1, 2, 3, \ldots \}$ be a martingale. The classical Azuma's inequality \cite[Theorem 7.2.1]{alon2004probabilistic} states that if $|Y_k-Y_{k-1}|\leq c_k$ with probability one, then
$$\mathbb{P} (|Y_N -Y_0 | \geq\eta  )\leq 2\e^{-\frac{\eta^2}{2\sum_{\ell=0}^Nc_\ell^2}}.
$$
The following martingale concentration result generalizes the Azuma inequality to the  case in which  $|Y_k-Y_{k-1}|$ is
not bounded.
\begin{lemma}[Lemma 1 in \cite{chalker_godbole_hitczenko_radcliff_ruehr_1999}]\label{lemma:Tao2015}Let $\{Y_k : k = 0, 1, 2, 3, \ldots \}$ be a martingale. Then for all sequences of positive numbers $(c_{\ell})$ and $\eta>0$, we have the following inequality
$$\mathbb{P} (|Y_N -Y_0 | \geq\eta  )\leq 2\e^{-\frac{\eta^2}{32\sum_{\ell=1}^Nc_\ell^2}}+\left(1+\frac{2\Delta^{\star}}{\eta}\right)\sum_{\ell=1}^n\mathbb{P}(|Y_\ell-Y_{\ell-1}|\geq c_\ell),
$$
with $\Delta^{\star}=\sup_i|Y_i-Y_{i-1}|$.
\end{lemma}

We recall that we consider three sources of randomness: the dynamics defined by $\Theta$ in \eqref{eq:Theta}, the activation process and the labeled network. The concentration property is proved in two steps. First, we study concentration by sequentially unveiling the edges in the labeled network (Lemma \ref{lemma:unveiling_edges}) and then we consider the other sources of randomness for a fixed graph (see Lemma \ref{lemma:unveiling_activation}).
\medskip


{\subsubsection{Unveiling the network}}
We recall that for any $\ell\in\N$ we use the notation $[\ell]=\{1,\ldots, \ell\}$. 
Let $\Pi_{a,a'}$ be the set of all permutations of $\mathcal{L}_{a,a'}=\{1,\ldots,l_{a,a'}\}$ for any $a,a'\in\mathcal{A}$ and denote by $\Pi=\times_{a,a'\in\mathcal{A}}\Pi_{a,a'}$. Since each of permutation $\pi_{a,a'}\in\Pi_{a,a'}$ defines a specific pairing of out-links from nodes with label $a$ and in-links of nodes with label $a'$, there are exactly $\prod_{a,a'\in\mathcal{A}}l_{a,a'}!$ distinct elements in $\Pi$. 
We define the following 
{ cylinder sets}: 
\begin{gather}
\mathcal{C}_{\ell}(\pi_{[\ell]})=\{ \upsilon\in\Pi:\upsilon_{[\ell]}=\pi_{[\ell]} \} ,\qquad \forall  \pi_{[\ell]} \label{eq:filtration_a}
\end{gather}
We notice that  $\mathcal{C}_{\ell}(\cdot)$ are disjoint and exhaustive events, i.e.,   $\mathcal{C}_{\ell}(\pi_{[\ell]}) \cap \mathcal{C}_{\ell}(\pi'_{[\ell]})=\emptyset$  if $ \pi_{[\ell]}\neq \pi'_{[\ell]}$  and $ \cup_{\pi_{[\ell]}} \mathcal{C}_{\ell}(\pi_{[\ell]})=\Pi$.  
 
{\color{black}{\begin{lemma}[Unveiling network] \label{lemma:unveiling_edges}
Let $\mathcal{N}=((\V,\E,\mathcal A,\lambda,\sigma,\tau))$ be a network sampled from the model ensemble $\mathfrak{C}_{n,p}$ of all labeled networks with given size $n$ and statistics $p$. 
We denote the induced graph obtained in the exploration process of the neighborhood of a node $v$ by $\mathcal{N}_{t}^v$, and with $V_t^{v}$ the number of nodes in it.  
For $t \geq 0$, let $Z(t)$ be the state vector of the ASD dynamics on $\mathcal{N}$, $b(t) =|\{v\in\V:Z_v(t)=\omega\}|$ be the number of state-$\omega$ adopters at time $t$. We denote the expectation over the ensemble of labeled graphs by $\widetilde{b}(t)$.
For any $s\geq 1$ we have
\begin{align*}
\mathbb{P}\left(|b(t)-\widetilde{b}(t)|\geq \eta n\right)\leq \inf_{x>0}\left\{2\e^{-\frac{\eta^2 n}{32\overline{d}x^2}}+\left(1+\frac{2}{\eta}\right)\frac{{{2^s}}}{x^s}\sum_{v\in\mathcal{V}}|\boldsymbol{\delta}_v|{\left[\mathbb{E}[|V_t^{v}|^s\right]}
\right\}
\end{align*}
\end{lemma}}}

\begin{proof} Let $\pi  \in\Pi$ 	 be the random element  of $\Pi$ (uniformly extracted by $\Pi$)  which  describes the network $\mathcal{N}=((\V,\E,\mathcal A,\lambda,\sigma,\tau))$.
	We denote  with $\mathcal{F}_\ell$ the natural filtration  generated by  $\pi _{[\ell]}$, with $\mathcal{F}_0$ equal to the 
trivial  $\sigma$-algebra.
		Let $\pi^{\ell} \in\Pi$,  for any given $\ell\in \{1\cdots |\mathcal{E}|\}$,  a  random  element  in $\Pi$ satisfying the following properties:
	i)   $\pi^\ell_{[\ell]}= \pi _{[\ell]}$; ii)  $\pi^\ell_{\ell+1}$ and  $\pi _{\ell+1}$  are conditionally independent given $ \pi _{[\ell]}$; iii) for any $i>1$,
   $\pi^\ell_{\ell+i}= \pi _{\ell+i}$   if $ \pi _{\ell+i}\neq \pi^\ell_{\ell+1}$  and  $\pi^\ell_{\ell+i}= \pi _{\ell+1}$  if   $ \pi _{\ell+i}= \pi^\ell_{\ell+1}$. Observe that by construction  $\pi^\ell$  is extracted uniformly from $\Pi$, as well. Furthermore the conditional law of both  $\pi $  and $\pi^{\ell}$, given  $\pi _{[\ell]}$ with   $\pi _{[\ell]}=\pi^\ell_{[\ell]}$ is uniform in   $C_\ell(\pi _{[\ell]})$.
  Let  $\mathcal{G}^{\ell}_{\ell'}$  be the  natural filtration  induced by  $\pi ^{\ell}_{[\ell']}$. Observe that by construction  $\mathcal{G}^{\ell}_{\ell'}=\mathcal{F_\ell} $  for $\ell'\le \ell$.
At last let $\hat{\mathcal{F}}_{\ell+1} $  and 	 $\hat{\mathcal{G}}^\ell_{\ell+1} $ the  natural filtration  induced by  $\pi _{\ell+1}$ and  $\pi ^\ell_{\ell+1}$. 
Of course   $\mathcal{F}_{\ell+1}= \sigma(\mathcal{F}_{\ell}\cup\hat{\mathcal{F}}_{\ell+1})$ and  $\mathcal{G}^\ell_{\ell+1}= \sigma(\mathcal{G}^\ell_{\ell}\cup\hat{\mathcal{G}}^\ell_{\ell+1}).$

We have
$$\mathbb{P}(|b_\pi (t)-\widetilde{b}(t)|\geq\eta n)= \mathbb{P}\left(\left|\mathbb{E}[b_\pi (t)|\mathcal{F}_{|\mathcal{E}|}]-\mathbb{E}[b_\pi (t)|\mathcal{F}_{0}]\right|\geq{\eta n}\right).
$$
Let us emphasize the dependence of the number of $\omega$-adopters $b(t)$ on a specific graph $\pi \in\Pi$ with notation $b_{\pi }(t)$
and define
{$A_{\ell}=\mathbb{E}[b_{\pi }(t)| \mathcal{F}_{\ell}]$.}  Note, indeed, that $\{A_\ell\}_\ell$  is a martingale.

{\color{black}
In order to estimate the above probability we apply Lemma \ref{lemma:Tao2015}.
First, we compute
\begin{align*}
\mathbb{P}\left(|A_{\ell+1}-A_{\ell}|\geq c_{\ell}\right)= \mathbb{P}\left(|A_{\ell+1}-A_{\ell}|^s\geq
c_{\ell}^s\right)\leq \frac{\mathbb{E}\left[|A_{\ell+1}-A_{\ell}|^s\right]}{c_{\ell}^s}
\end{align*}
Notice that by construction
$$
A_{\ell}=\mathbb{E}[b_{\pi }(t)|\mathcal{F}_{\ell}]= \mathbb{E}[b_{\pi^\ell}(t)|\mathcal{F}_{\ell}]= \mathbb{E}[b_{\pi^\ell}(t)|\mathcal{F}_{\ell+1}]
$$ 
where the first equation holds because $\pi $ and   $\pi^\ell$ are both uniform on $\mathcal{C}_\ell(\pi _{[\ell]}) $,
 and last equation descends from the fact that $\hat{G}^\ell_{\ell+1}$ and  $\hat{F}^\ell_{\ell+1}$
 are conditionally independent given $\mathcal{F}_\ell=\mathcal{G}^\ell_\ell$.
Furthermore we have:
$$
A_{\ell+1}=\mathbb{E}[b_{\pi }(t)|\mathcal{F}_{\ell+1}]. 
$$
Therefore  
\begin{align}
A_{\ell+1}-A_{{\ell}}&=\mathbb{E}[b_{\pi }(t)\mid\mathcal{F}_{\ell+1}]-\mathbb{E}[b_{\pi }(t)\mid\mathcal{F}_{\ell}],\nonumber\\
&=\mathbb{E}[b_{\pi }(t)\mid\mathcal{F}_{\ell+1}]-\mathbb{E}[b_{\pi^\ell}(t)\mid\mathcal{F}_{\ell+1}]\nonumber\\
&=\mathbb{E}[b_{\pi }(t)-b_{\pi^\ell}(t)\mid\mathcal{F}_{\ell+1}]\label{eq:bound_diff}
\end{align}
Now observing that by construction $\pi $ and $\pi^\ell$ differ in at most two positions, hence we have:
$$
\mathbb{E}[b_{\pi }(t)-b_{\pi^\ell}(t)\mid\mathcal{F}_{\ell+1}]\leq 2 \mathbb{E}[|\mathcal{N}_t^{\mathrm{v}(\pi _{\ell+1})}|\mid\mathcal{F}_{\ell+1}]
$$
and
$$
\mathbb{E}\left[|A_{\ell+1}-A_{\ell}|^s\right]\leq {2^s}\mathbb{E}\left[ \left(\mathbb{E}[\mid\mathcal{N}_t^{\mathrm{v}(\pi _{\ell+1})} \mid |\mathcal{F}_{\ell+1}]\right)^s\right] \leq {2^s}\mathbb{E}\left[|V_t^{\mathrm{v}(\pi _{\ell}+1)}|^s\right]
$$
where $\mathrm{v}({\pi _{\ell}})$ is the in-node of edge $\pi _{\ell}$.
We conclude that 
$$
\mathbb{P}\left(|A_{\ell+1}-A_{\ell}|\geq c_{\ell}\right)\leq {2^s}\frac{\left[\mathbb{E}[|V_t^{\mathrm{v}(\pi _{\ell+1})}|^s\right]}{c_{\ell}^s}
$$
For any $x>0$ let $c_{\ell}=x$ for all $\ell$. Then, by applying Lemma \ref{lemma:Tao2015} and observing that $\Delta^{\star}\leq n$, we obtain
\begin{align*}
\mathbb{P}\left(|b(t)-\widetilde{b}(t)|\geq \eta n\right)&\leq 2\e^{-\frac{\eta^2 n}{32\overline{d}x^2}}+\left(1+\frac{2}{\eta}\right)\frac{2^s}{x^s}\sum_{\ell=1}^{\overline d n}{\left[\mathbb{E}[|V_t^{v(\pi _{\ell})}|^s\right]}\\
&=2\e^{-\frac{\eta^2 n}{32\overline{d}x^2}}+\left(1+\frac{2}{\eta}\right)\frac{{{2^s}}}{x^s}\sum_{v\in\mathcal{V}}|\boldsymbol{\delta}_v|{\left[\mathbb{E}[|V_t^{v}|^s\right]}
\end{align*}
from which the thesis.
}
\end{proof}


	\begin{remark}
		The approach followed in Lemma \ref{lemma:unveiling_edges} 
		  can be potentially 
		extended to  more general  classes of random graphs, with a variable number of edges,   as long as: i) the number of edges  in the graph is sufficiently concentrated around its expectation;  ii)  random variables associated to the presence of different edges  in the graph are  sufficiently weakly correlated, so that we can effectively bound $\mathbb{E}[|A_{\ell+1}- A_{\ell}|]$, through a coupling argument similar to the one  established  between $\pi$ and $\pi^\ell$  in    Lemma \ref{lemma:unveiling_edges}.
	\end{remark}

Let $\mathcal{N}=((\V,\E,\mathcal A,\lambda,\sigma,\tau))$ be a labeled graph. We denote the random times at which the opinion update occurs, random node sequence activated, and the random state, by $\{{T}_{\ell}\}_{\ell\in\N},\{{w}_{\ell}\}_{\ell\in\N}$, $\{{z}_{\ell}\}_{\ell\in\N}$, respectively. 
For $t \geq 0$, let $Z(t)$ be the state vector of the ASD dynamics on $\mathcal{N}$ and $b(t) =|\{v\in\V:Z_v(t)=\omega\}|$ be the number of state-$\omega$ adopters at time $t$.

\begin{lemma}\label{lemma:Ex.1.13_Schwartz}
 Let $\{{T}_{\ell}\}_{\ell\in\N}$ be the random times at which the opinion update occurs. For $t>0$ define the random variable $$\iota(t)=\sup\{k\in\N:T_k\leq t\}.$$ Then for any $\epsilon>0$ and  $\Delta_n<tn$
the following bounds hold
\begin{gather*}
\mathbb{P}(\iota(t)\geq (1+\epsilon)tn)\leq\e^{-\frac{nt\epsilon^2}{2(1+\epsilon)}}\\
\mathbb{P}(|\iota(t)-tn|\geq \Delta_n)\leq2\e^{-\frac{\Delta_n^2}{2(tn+\Delta_n)}}.
\end{gather*}
\end{lemma}
\begin{proof}
This is a straightforward consequence of Chernoff bound \cite{feller1}. 
\end{proof}

\begin{lemma}[Unveiling dynamics]\label{lemma:unveiling_activation}
{\color{black}{Let $\mathcal{N}=((\V,\E,\mathcal A,\lambda,\sigma,\tau))$ be a labeled graph. Let $\{ {T}_{\ell}\}_{\ell\in\N},\{ {w}_{\ell}\}_{\ell\in\N}, \{{z}_{\ell}\}_{\ell\in\N}$ be the random times at which the opinion update occurs, random node sequence activated, and random state of activated sequence, respectively. We denote the size of the induced graph obtained in the exploration process of the neighborhood of a node $v$ with $V_t^v$ at time $t$.  }}
For $t \geq 0$, let $Z(t)$ be the state vector of the ASD dynamics on $\mathcal{N}$, $b(t) =|\{v\in\V:Z_v(t)=\omega\}|$ be the number of state-$\omega$ adopters at time $t$ conditioned to $\mathcal{N}$. We denote the expectation over the activation process by $\overline{b}(t)=\mathbb{E}[b(t)|\mathcal{N}]$. 
For any $\epsilon>0$ we have 
\begin{align*}
&\mathbb{P}(|b(t)-\overline{b}(t)|> \eta n)\\
&\quad\leq2 \inf_{x>0}\left\{2\e^{-\frac{\eta^2 n}{288(1+\epsilon)tx^2}} +\left(1+\frac{6}{\eta}\right)(1+\epsilon)tn\frac{2^s\mathbb{E}_{\color{black}{v}}\left[|V^{v}_t|^s\right]}{x^s}\right\}\\
&\quad+2\e^{-\frac{nt\epsilon^2}{2(1+\epsilon)}}+2\e^{-\frac{\eta^2 n}{72(t+\eta/6)}}
\end{align*}
with $v$ chosen uniformly at random in $\V$.
\end{lemma}

\begin{proof} For $t\in\R^{+}$ let $\iota(t)=\sup\{k\in\N:T_k\leq t\}$, $w$ and $z$ be the random sequences of activated nodes and the corresponding random state. We recall that for any $\ell>0$ the sequence $w_{[\ell]}$ is uniformly distributed over $\mathcal{V}^{[\ell]}$.
We denote by $\mathcal{F}_{\ell,s}$ the natural filtration generated by $w_{[\ell]}$ and $z_{[s]}$.
 Given  $\iota(t)$,  let ${w}$  be a random vector uniformly distributed in $\mathcal{V}^{[\iota(t)] }$ (let $ {w}_{\ell+1}=v$) and 
  $\hat{ {w}}^{\ell}$  be a random vector in 
$\mathcal{V}^{[\iota(t)]}$ which is obtained by choosing some $v'$ uniformly at random from the set of nodes $\V$ and putting $\hat{ {w}}^{\ell}_{\ell+1}=v'$ and $\hat{ {w}}^{\ell}_i =  {w}_i$ for all $i\in[\iota(t)]\setminus\{\ell+1\}$. 
It should be noticed that $w_{\ell+1}$ and $\hat w^{\ell}_{\ell+1}$ are conditionally independent given $w_{[\ell]}$. Furthermore, by construction  $\hat w^{\ell}$ is uniformly distributed over $\mathcal{V}^{[\iota(t)]}$. 

In an analogous way, recall that $ {z}_{s}=Z_{w_s}$ is a random variable distributed as defined in Definition \ref{def:ASD}.
 Given $\iota(t)$,  let $ {z}$ be a vector of length  $\iota(t)$, \textcolor{black}{whose components  are  independent with the $\ell$-component distributed  as $\Theta^{(\lambda(w_\ell))}$ in Definition \ref{def:ASD}}
 (let ${z}_{s+1}=\omega$) 
 and, let  $\tilde{ {z}}^{s}$  be a random vector which is obtained by choosing some $\omega'$ according to \textcolor{black}{$\Theta^{(\lambda(w_\ell))}$} in Definition \ref{def:ASD} and putting $\tilde{ {z}}^{s}_{s+1}=\omega'$ and $\tilde{ {z}}^{s}_i =  {z}_i$ for all $i\in[\iota(t)]\setminus\{s+1\}$. 
 

Let us emphasize the dependence of the number of $\omega$-adopters $b(t)$ on a specific sequence of activated nodes $w$ and states $z$ with notation $b_{w,z}(t)$.  Given $\iota(t)$, we define for any $(\ell,s)\in[\iota(t)]\times [\iota(t)]$
 $$B^{\iota(t),\mathcal{N}}_{\ell,s}=\mathbb{E}[b(t)|\iota(t), \mathcal{F}_{\ell,s},\mathcal{N}],$$ then 
\begin{align*}
\mathbb{P}\left(|b_{w,z}(t)-\overline{b}(t)|\geq \eta n\right)
&\leq\mathbb{P}\left(\left|\mathbb{E}[b_{w,z}(t)|\iota(t), \mathcal{F}_{\iota(t),\iota(t)},\mathcal{N}]-\mathbb{E}[b_{w,z}(t)|\iota(t),\mathcal{F}_{\iota(t),0},\mathcal{N}]\right|\geq\frac{\eta n}{3}\right)\nonumber\\
&+\mathbb{P}\left(\left|\mathbb{E}[b_{w,z}(t)|\iota(t), \mathcal{F}_{\iota(t),0},\mathcal{N}]-\mathbb{E}[b_{w,z}(t)|\iota(t), \mathcal{F}_{0,0},\mathcal{N}]\right|\geq \frac{\eta n}{3}\right)\nonumber\\
\label{eq:bound_martingale1}
&+\mathbb{P}\left(\left|\mathbb{E}[b_{w,z}(t)|\iota(t),\mathcal{N}]-\mathbb{E}[b_{w,z}(t)|\mathcal{F}_{0,0},\mathcal{N}]\right|\geq \frac{\eta n}{3}\right)\\
&= \underbrace{\mathbb{P}\left(|B^{\iota(t),\mathcal{N}}_{\iota(t),\iota(t)}-B^{\iota(t),\mathcal{N}}_{\iota(t),0}|> \frac{\eta n}{3}\right)}_{\text{(T1)}}+\underbrace{ \mathbb{P}\left(|B^{\iota(t),\mathcal{N}}_{\iota(t),0}-B^{\iota(t),\mathcal{N}}_{0,0}|> \frac{\eta n}{3}\right)}_{\text{(T2)}}\\
&\quad+ \underbrace{\mathbb{P}\left(\left|B^{\iota(t),\mathcal{N}}_{0,0}-\mathbb{E}[b_{w,z}(t)|\mathcal{F}_{0,0},\mathcal{N}]\right|\geq \frac{\eta n}{3}\right)}_{\text{(T3) }}
\end{align*}

We now evaluate (T1) and (T2)  by applying Lemma \ref{lemma:Tao2015} and (T3) using simple arguments.
\begin{itemize}\item
In order to estimate (T1) we first consider
\begin{align*}
B^{\iota(t),\mathcal{N}}_{\ell+1,0}-B^{\iota(t),\mathcal{N}}_{\ell,0}&=\mathbb{E}[b_{w,z}(t)|\iota(t), \mathcal{F}_{\ell+1,0},\mathcal{N}]-\mathbb{E}[b_{w,z}(t)|\iota(t),  \mathcal{F}_{\ell,0},\mathcal{N}]\\
&=\mathbb{E}[b_{w,z}(t)|\iota(t), \mathcal{F}_{\ell+1,0},\mathcal{N}]-\mathbb{E}[b_{\hat{w}^{\ell+1},z}(t)|\iota(t),  \mathcal{F}_{\ell+1,0},\mathcal{N}]\\
&=\mathbb{E}[b_{w,z}(t)-b_{\hat{w}^{\ell},z}(t)|\iota(t),  \mathcal{F}_{\ell+1,0},\mathcal{N}].
\end{align*}
By observing that, by construction, $w$ and $\hat{w}^{\ell-1}$ differ in at most one position, we get
$$
\mathbb{E}[b_{w,z}(t)-b_{\hat{w}^{\ell},z}(t)|\iota(t),  \mathcal{F}_{\ell+1,0},\mathcal{N}]\leq 2\mathbb{E}[|V^{v}_t||\mathcal{F}_{\ell,0}]
$$
where $v$ is chosen uniformly at random in $\V$.
We thus have for any $x>0$ 
\begin{align*}
\mathbb{P}\left(|B^{\iota(t),\mathcal{N}}_{\ell+1,0}-B^{\iota(t),\mathcal{N}}_{\ell,0}|>x\right)&= 
\mathbb{P}\left(|B^{\iota(t),\mathcal{N}}_{\ell+1,0}-B^{\iota(t),\mathcal{N}}_{\ell,0}|^m>x^m\right)\\
&\leq\frac{\mathbb{E}\left[|B^{\iota(t),\mathcal{N}}_{\ell+1,0}-B^{\iota(t),\mathcal{N}}_{\ell,0}|^m\right]}{x^m}\\
&\leq\frac{2^m\mathbb{E}\left[\left(\mathbb{E}[|V^{v}_t||\mathcal{F}_{\ell+1,0}]\right)^m\right]}{x^m}\\
& \leq\frac{2^m\mathbb{E}\left[|V^{v}_t|^m\right]}{x^m}
\end{align*}
where $v$ is chosen uniformly at random.
We thus have for any $\epsilon>0$
\begin{align}&\mathbb{P}\left(|B^{\iota(t),\mathcal{N}}_{\iota(t),0}-B^{\iota(t),\mathcal{N}}_{0,0}|> \frac{\eta n}{3}\right)\nonumber\\
 &\quad\leq\mathbb{P}\left(|B^{\iota(t),\mathcal{N}}_{\iota(t),0}-B^{\iota(t),\mathcal{N}}_{0,0}|> \frac{\eta n}{3}\big|\iota(t)<(1+\epsilon)tn\right)+\mathbb{P}\left(\iota(t)\geq(1+\epsilon)tn\right)\nonumber\\
&\quad\leq  \inf_{x>0}\left\{2\e^{-\frac{\eta^2 n}{288(1+\epsilon)tx^2}} +\left(1+\frac{6}{\eta}\right)\textcolor{black}{(1+\epsilon)}tn\frac{2^m\mathbb{E}\left[|V^{v}_t|^m\right]}{x^m}\right\}+\e^{-\frac{nt\epsilon^2}{2(1+\epsilon)}}
\end{align}
with $v$ chosen uniformly at random in $\V$.

\item (T2):
Following the same arguments used in the previous point and observing that $\tilde{z}_s^{s}$  differs from $z_{[\iota(t)]}$ only at position $s+1$, we get
 \begin{align}\begin{split}\label{eq:unv_state}
 &\mathbb{P}\left(|B^{\iota(t),\mathcal{N}}_{\iota(t),\iota(t)}-B^{\iota(t),\mathcal{N}}_{\iota(t),0}|> \frac{\eta n}{3}\right)\leq  \\
 &\inf_{x>0}\left\{2\e^{-\frac{\eta^2 n}{288(1+\epsilon)tx^2}} +\left(1+\frac{6}{\eta}\right)(1+\epsilon)tn\frac{2^m\mathbb{E}\left[|V^{v}_t|^m\right]}{x^m}\right\}+\e^{-\frac{nt\epsilon^2}{2(1+\epsilon)}}.\end{split}\end{align}
 
\item (T3):
Let now $\hat{\iota}(t)$ a random variable taking values in $\mathbb{Z}^+$, distributed as $\iota(t)$, and independent of it; let   $ {w}_{[\iota(t)]} \in  \V^{\iota(t)}$  and $\hat{ {w}}_{[\hat{\iota}(t)]}\in  \V^{\hat{\iota}(t)}$ random uniform sequences  of activation of length $\iota(t)$ and $\hat{\iota}(t)$ respectively and  
$b(t)$  and $\hat{b}_{\ell}(t)$ the corresponding 
the numbers of state-$\omega$ adopters at time $t$ on the network $\mathcal{N}$:

For any $\Delta_n$ we have
\begin{align}
&\mathbb{P}\left(\left|\mathbb{E}[b(t)|\iota(t),\mathcal{F}_{0,0},\mathcal{N}]-\mathbb{E}[b(t)|\mathcal{F}_{0,0},\mathcal{N}]\right|
\geq \frac{\eta n}{3}\right)\\
&  =  \mathbb{P}\left(\left|\mathbb{E}[b(t)|\iota(t),\mathcal{F}_{0,0},\mathcal{N}]-\mathbb{E}[\hat{b}_{\ell}(t)\mid \iota(t),\mathcal{F}_{0,0},\mathcal{N}]\right|
\geq \frac{\eta n}{3}\right)\nonumber\\
&\leq\mathbb{P}\left(|\iota(t)-\hat{\iota}(t)|\geq \frac{\eta n}{3}\right)\nonumber\\
&\leq\mathbb{P}\left(|\iota(t)-\hat{\iota}(t)|\geq \frac{\eta n}{3}\Big| |\hat{\iota}(t)-tn|\leq \Delta_n,|{\iota}(t)-tn|\leq \Delta_n\right)\nonumber\\
&+2\mathbb{P}\left(|{\iota}(t)-tn|> \Delta_n\right)\leq2\e^{-\frac{\Delta_n^2}{2(tn+\Delta_n)}}
\label{eq:bound_iota}\end{align}
where  the second inequality  descend from the fact that we can establish a coupling between   the sequences of activation 
$ {w}_{[\iota(t)]}$ and $\hat{ {w}}_{[\hat{\iota}(t)]}$
 by forcing them to have  common initial part  of length $ \min(\iota(t),\hat{\iota}(t))$   while
the last inequality follows from Lemma \ref{lemma:Ex.1.13_Schwartz}. {\color{black}Choosing $\Delta_n=\frac{\eta n}{6}$ we get 
$$\mathbb{P}\left(|\iota(t)-\hat{\iota}(t)|\geq \frac{\eta n}{3}\Big| |\hat{\iota}(t)-tn|\leq \Delta_n,|{\iota}(t)-tn|\leq \Delta_n\right)=0$$ and we conclude the proof combining with \eqref{eq:unv_state}.}
\end{itemize}
\end{proof}

\textbf{Proof of Theorem \ref{thm:concentration_modgraphs}}
For any $\epsilon>0$ we have 
\begin{align*}
&\mathbb{P}(|b(t)-\mathbb{E}[{b}(t)]|> \eta n)\\
&\quad\leq\mathbb{P}(|b(t)-\mathbb{E}[{b}(t)|\mathcal{N}]|> \eta n/2)+\mathbb{P}(|\mathbb{E}[{b}(t)|\mathcal{N}]-\mathbb{E}[{b}(t)]|> \eta n/2)
\end{align*}
Let $v$ is sampled with a probability proportional with its in-degree.  
Combining Lemma \ref{lemma:unveiling_activation} with Lemma \ref{lemma:unveiling_edges}
 we get that for any $\epsilon>0$, $\eta>0$ and $x>0$ we have 
{\color{black}{\begin{align*}
\mathbb{P}(|b(t)-\mathbb{E}[{b}(t)]|> \eta n)&\leq
4\e^{-\frac{\eta^2 n}{1152(1+\epsilon)tx^2}} +\left(1+\frac{12}{\eta}\right)(1+\epsilon){tn}\frac{2^s\mathbb{E}_v\left[|V^{v}_t|^s\right]}{x^s}\\
&+2\e^{-\frac{nt\epsilon^2}{2(1+\epsilon)}}+2\e^{-\frac{\eta^2 n}{288(t+\eta/12)}}\\
&+\left(1+\frac{4}{\eta}\right)\frac{{{2^s}}}{x^s}\sum_{w\in\mathcal{V}}|\boldsymbol{\delta}_w|{\left[\mathbb{E}[|V_t^{w}|^s\right]}
+2\e^{-\frac{\eta^2 n}{128\overline{d}x^2}}.
\end{align*}
}}

\section{Convergence to ODE solution with asymptotic degree distribution}
\label{Convergence_ODE}

\textbf{Proof of Proposition \ref{prop:asympdist1}.}\quad
Let ${ \textbf f}^{(n)}({\textbf z})={\boldsymbol \phi}^{(n)}({\textbf z})-{\textbf z}$ and ${\textbf f}({\textbf z})={\boldsymbol \phi}({\textbf z})-{\textbf z}$.  
For any $\Delta>0$ and $t \in[0,m\Delta]$ we have
\begin{align*}
\boldsymbol{\zeta}^{(n)}(t)&=\boldsymbol{\zeta}^{(n)}(0)+\int_0^{(m-1)\Delta} {\textbf f}^{(n)}(\boldsymbol{\zeta}^{(n)}(s))\textrm{d}s+\int_{(m-1)\Delta}^{t}{\textbf f}^{(n)}(\boldsymbol{\zeta}^{(n)}(s))\textrm{d}s\\
\boldsymbol{\zeta}(t)&=\boldsymbol{\zeta}(0)+\int_0^{(m-1)\Delta}{\textbf f}(\boldsymbol{\zeta}(s))\textrm{d}s+\int_{(m-1)\Delta}^{t}{\textbf f}(\boldsymbol{\zeta}(s))\textrm{d}s
\end{align*}
from which
\begin{align*}
\|\boldsymbol{\zeta}^{(n)}(t)-\boldsymbol{\zeta}(t)\|_{m\Delta}:&=\sup_{t\in[0,m\Delta]}\|\boldsymbol{\zeta}^{(n)}(t)-\boldsymbol{\zeta}(t)\|_{\infty}\\
&\leq\sup_{t\in[0,(m-1)\Delta]}\|\boldsymbol{\zeta}^{(n)}(t)-\boldsymbol{\zeta}(t)\|_{\infty}+\Delta \left\|{\textbf f}^{(n)}-{\textbf f}\right\|_{\infty}\\
&+\sup_{t\in[(m-1)\Delta,m\Delta]}  \int_{(m-1)\Delta}^{t}\left\|{\textbf f}(\boldsymbol{\zeta}^{(n)}(s))-{\textbf f}(\boldsymbol{\zeta}(s))\right\|_{\infty}\textrm{d}s.
\end{align*}
Since ${\textbf f}$ is Lipschitz on a compact and invariant set $B=[0,1]^{|\mathcal A|}$, then there exists $L>0$ such that for any $\boldsymbol{\zeta}_1, \boldsymbol{\zeta}_2\in B$
it holds that $
\left\|f(\boldsymbol{\zeta}_1)-f(\boldsymbol{\zeta}_2)\right\|_{\infty}\leq L\|\boldsymbol{\zeta}_1- \boldsymbol{\zeta}_2\|_{\infty}$.
Then for any $\Delta<1/L$ we have
\begin{align*}
\|\boldsymbol{\zeta}^{(n)}(t)-\boldsymbol{\zeta}(t)\|_{m\Delta}
&\leq\frac{\|\boldsymbol{\zeta}^{(n)}(t)-\boldsymbol{\zeta}(t)\|_{(m-1)\Delta}}{1-L\Delta}+\frac{\Delta \left\|{\textbf f}^{(n)}-{\textbf f}\right\|_{\infty}}{1-L\Delta}\\
& \leq\frac{\|\boldsymbol{\zeta}^{(n)}(0)-\boldsymbol{\zeta}(0)\|_{\infty}}{(1-L\Delta)^m}+\frac{1}{1-L\Delta}\sum_{j=0}^{m-1}\frac{\Delta \left\|{\textbf f}^{(n)}-{\textbf f}\right\|_{\infty}
}{(1-L\Delta)^{j}}\\
&=\frac{\|\boldsymbol{\zeta}^{(n)}(0)-\boldsymbol{\zeta}(0)\|_{\infty}}{(1-L\Delta)^m}+\frac{\left\|{\textbf f}^{(n)}-{\textbf f}\right\|_{\infty}}{L}\left(\frac{1}{(1-\Delta L)^m}-1 \right)\\
&\leq\frac{\|\boldsymbol{\zeta}^{(n)}_0-\boldsymbol{\zeta}_0\|_{\infty}}{(1-L\Delta)^m}+\frac{\|q^{(n)}_{\textbf{k}|a}-q_{\textbf{k}|a}\|_{\mathrm{TV}}}{L}\left(\frac{1}{(1-\Delta L)^m}-1\right).
\end{align*}
with a  similar procedure we obtain the bound on $\|\boldsymbol{y}^{(n)}(t) - \boldsymbol{y}(t) \| $.

\section{Proofs of Sections \ref{sec:asdconfig} and \ref{section:CBM2}}\label{app:B}

\textbf{Proof of Lemma \ref{lemma:F_branching}}
Let $\mathrm{d}(w_1,w_2)$ be the geodesic distance (i.e. the number of edges in a shortest path) between nodes $w_1$ and $w_2$. We denote the maximum number of hops traversed from the root $v$ to a node $w$ in 
$\mathcal{T}_t$ with $H_v(t)=\max_{w\in \mathcal{T}_t}\mathrm{d}(v,w).$ Equivalently, $H_v(t)$ is the depth of the tree $\mathcal{T}_t$. Let us fix $h_n=c\log n$ with $c>0$ then
\begin{align*}
\mathsf{F}_{\widetilde{W}_t}(x_n)&\leq\mathbb{P}(\widetilde{W}_t >x_n|H_v(t)< h_n)+\mathbb{P}(H_v(t)\geq h_n)\\
& =\mathbb{P}(\widetilde{W}_t >x_n|H_v(t)< h_n)+\mathbb{P}\left(\max_{w\in \mathcal{T}_t}\mathrm{d}(v,w)\geq h_n\right)
\\
&= \mathbb{P}(\widetilde{W}_t >x_n|H_v(t)< h_n)+\mathbb{P}\left(\exists {w\in \mathcal{T}_t}:\mathrm{d}(v,w)= h_n\right)
\\
&\leq \mathbb{P}(N_{h_n} >x_n)+n\mathbb{P}(\widetilde{P}(t)\geq h_n)
\end{align*}
where $\{N_{h}\}_{h\in\N}$ is a truncated  GW process of maximum depth $h$, in which the offspring distribution of the root follows law $p$,  while the degree of remaining nodes follow law $q$, and
$\widetilde{P}(t)$ is a variable representing the number of points falling in $[0,t)$ according to a homogeneous Poisson  process with constant parameter $\gamma=1$. 
Note that  $n\mathbb{P}(\widetilde{P}(t)\geq h_n)$  represents an obvious upper-bound  to the probability that the depth of $\mathcal{T}_t$ 
exceeds $h_n$, 
since by conduction   $\mathbb{P}(\widetilde{P}(t)\geq h_n)$  is equal to the probability that  a given branch of  $\mathcal{T}_t$  has depth larger than $h_n$.

 We have
\begin{align*}
\mathbb{P}(\widetilde{P}(t)\geq h_n)&\leq \sum_{h=h_n}^{\infty}\frac{\e^{-  t}t^h}{h!}= \frac{\e^{-  t}t^{h_n}}{h_n!}\sum_{h\geq h_n}\frac{t^{h-h_n}}{h(h-1)\ldots (h_n+1)}\\
&= \frac{\e^{-  t}t^{h_n}}{h_n!}\sum_{s\geq 0}\frac{t^{s}}{(h_n+s)(h_n+s-1)\ldots (h_n+1)}\\
&\leq \frac{\e^{-  t}t^{h_n}}{h_n!}\sum_{s\geq 0}\left(\frac{  t}{h_n}\right)^s= \frac{\e^{-  t}t^{h_n}}{h_n!}\left(1-\frac{  t}{h_n}\right)^{-1}\quad\text{for } h_n\rightarrow\infty
\end{align*}
where the last inequality follows from $  t/h_n<1$ definitely, being $  t=o(h_n)$  for $h_n\rightarrow\infty$. Using Stirling's approximation \cite{feller1} $$h_n!\geq \sqrt{2\pi}h_n^{h_n+1/2}\e^{-h_n}$$ we obtain
\begin{align}\label{eq:Stirling}
\mathbb{P}(\widetilde{P}(t)\geq h_n)&\leq  \e^{-t+h_n\log t-h_n(\log h_n)+h_n-\frac{1}{2}\log(2\pi h_n)-\log(1-t/h_n)}.
\end{align}
Using bound in  \eqref{eq:Stirling}, we obtain for any $s>0$ 
\begin{align*}
\mathsf{F}_{\widetilde{W}_t}(x_n)&\leq \mathbb{P}(N_{h_n} >x_n)+n\e^{-h_n\log h_n+o(h_n\log h_n)}\leq \frac{\mathbb{E}[N_{h_n}^s]}{x_n^s}+o(1/n)\quad n\rightarrow\infty,
\end{align*}
where the second last inequality follows from the Markov inequality \cite{feller1}.
At last, we emphasize that the an analogous bound holds for the number $W_t$ of edges,
since $\widetilde{W}_t=W_t+1$.
\qed

\begin{lemma}\label{bound_number_nodes_conf}
Let $\mathcal{N}=(\V,\E)$ be a network sampled from the configuration model ensemble $\mathfrak{C}_{n,p}$ of compatible size $n$ and statistics $p$ and $q$, $\mathcal{N}^{w_0}_{h}$ be the induced graph obtained by the exploration process of the $h$-depth  neighborhood of a node $w_0$ chosen uniformly at random from the node set $\mathcal{V}$.
Let $\dot{q}$ be the distribution defined as follows:   
$\sum_{h=0} ^k \dot{q}_h =\min (\sum_{h=0} ^k p_h, \sum_{h=0} ^k q_h)$, $\forall k$.
Note that $\dot{q}$ stochastically dominates both $p$ and $q$.
Moreover, let $\widehat{q}_k$ be the distribution related to $\dot{q}_k$ 
as follows: $\widehat{q}_{k+k_0}=\dot{q}_k$, with 
$k_0= \min_k: \sum_{j=1}^k \dot{q}_j> \epsilon$.
Let ${N}^{w_0}_{h}$  be the number of nodes  in $\mathcal{N}^{w_0}_h$.
We have that for every $x_n\le \lfloor \epsilon n\rfloor$:
\[
\mathbb{P}({N}^{w_0}_{h}> x_n) \le\mathbb{P}( {\overline N}^{w_0}_{h}> x_n)
\]
where  ${\overline N}_{h_0}^{w_0}$ is  the total number of nodes over  a tree of  depth $h_0$, in which  the 
degree of all the nodes  follow law $\widehat{q}$. \end{lemma}


\begin{proof}
First note that, as long as $k_0$ is bounded, the order of magnitude of the 
moments of $\dot{q}$ and $\widehat{q}$ is the same.
We prove  the assertion through  coupling arguments. First we show that $N_h^{w_0}\le _{st} \widehat{N}_h^{w_0}$
where   $ \widehat{N}_h^{w_0}$ is the number of nodes in a tree in which the root has a degree distributed as $p$ and the degree of the other nodes are obtained by extractions without  repetitions from an empirical distribution matching $q $. Then we show that  
$\mathbb{P}(\widehat{N}^{w_0}_{h_0}> x_n) \le\mathbb{P}( {\overline N}^{w_0}_{h_0}> x_n)$ under the assumption that 
 $x_n\le \lfloor \epsilon n\rfloor$.  
 
 To show  that  $N_h^{w_0}\le _{st} \widehat{N}_h^{w_0}$, we start  performing a breadth-first exploration of $\mathcal{N}_h$
 and we denote with  $\widetilde{\mathcal{T}}_h$  the induced  spanning tree of $\mathcal{N}_h$.  Note that this spanning tree is obtained as result of the exploration of the  edges  
 $\{M_{i,j}\}_{i,j}$ of $\mathcal{N}_t$ (we use a double index, where $i$  represents the level, i.e., the distance from the root of the tree, 
 while $j$ is an ordinal number induced by the breadth-first exploration among edges of the same level)
 by retaining only those $M_{i,j}$ satisfying the following condition:   $\{ \nu(M_{i, j})\neq w_0 ,\nu(M_{i,j})\neq \nu(M_{k,q}),\forall (k,q) \text { with }  k < i   \text{ or }  k=i \text{ and  } q<j\}$.   Let $\mathcal{M}$ be the set of  explored
 edges at the end of the exploration process and $\widetilde{\mathcal{M}}$ be the corresponding  set of retained edges.  
 
 Now, let's consider  a tree  $\widehat{\mathcal{T}}_h$
 which is obtained from $\widetilde{\mathcal{T}}_h$  by adding to it  edges and nodes according to the following
 procedure.   First we \lq\lq re-add"  all edges in $\mathcal{M} \setminus \widetilde{\mathcal{M}}$ and we append to  each of them 
 a new node with an output degree equal to  $\delta(\nu(M_{i,j}))$. Then we append to
 every new added node whose distance from the  root is less that $h$ a number of edges $\widetilde{M}_{i,j}$ equal to the degree of the node, which are extracted uniformly without repetition from  $\mathcal{L}\setminus \mathcal{M}$. A node with degree $\delta(\nu(\widetilde{M}_{i,j}))$ is then appended to every newly introduced edge. The procedure is iterated until  level $h$ is reached.  By construction $ N_h^{w_0}=|\mathcal{N}_t| =| \widetilde{\mathcal{T}}_h| \le 
  |\widehat{\mathcal{T}}_h|= \widehat{N}_h^{w_0}$. Furthermore  $\widehat{\mathcal{T}}_h$ has by construction the desired structure,  so we have proved that $N_h^{w_0}\le _{st} \widehat{N}_h^{w_0}$.
  
Now through coupling we  show that,  conditionally over the event $|{\mathcal{T}}_h |\le \lfloor \epsilon n\rfloor $, we have 
$ |\widehat {\mathcal{T}}_h |\le_{st} |\mathcal{T}_h|$.  We couple the process of generation of the two trees in the following 
way: \textcolor{black}{first we assign to the root of $\widehat {\mathcal{T}}_h $ a degree extracted from $p$, and to the 
the root of $\mathcal{T}_h$ a degree extracted from $\widehat{q}$. Since $p\le_{st}\widehat{q}$  we can couple the two extractions so as to guarantee that the latter is non smaller than the former.}
Then,  let $\{M_i\}_i$ be a sequence of links extracted without repetition from $\{1, \ldots,  l \}$; without loss of generality we assume that 
$M_i<M_j$ implies $\delta(\nu(M_i))\le \delta(\nu(M_j))$.  Similarly let ${L}_i$ be a sequence of i.i.d. random variables
taking value in $\{1, \ldots, l \}$. Without loss of generality we assume that 
$L_i<L_j$ implies $\delta'(\nu(L_i))\le \delta'(\nu(L_j))$. By construction $\delta(\mathcal{L})$ follows  the empirical distribution $q$, while 
$\delta'(\mathcal{L})$ follows the distribution $\widehat{q}$ and therefore $\delta(i+k)\le \delta'(i)$ for every $i\le l- k$ and
$k<k_0:=\lfloor \epsilon n\rfloor$.  

 We generate $ \widehat {\mathcal{T}}_h$ and $\mathcal{T}_h$ by coupling the sequences $\{L_i\}_i$ and $\{M_i\}_i$.  Let 
   $\mathcal{M}_i= \mathcal{L}\setminus \{M_j\}_{j< i}$  and $F(\mathcal{M}_i,k)$  a function that  for any $k\le l -i+1$ returns the $k$-th   ordered element in  $\mathcal{M}_i$.  Now we couple   $\{L_i\}_i$ and $\{M_i\}_i$ in the following way:
   if $L_i\le  l -i+1$ then   $M_i=F(\mathcal{M}_i,k)$, otherwise $M_i$ is chosen uniformly in $\mathcal{M}_i$.
   It is of immediate verification that for any $i \le  k_0 $  we have $\delta'(\nu(L_i))\ge \delta(\nu(M_i))$, from which it immediately  descends that $ |\widehat {\mathcal{T}}_h|\le |\mathcal{T}_h|$.

 \end{proof}

\begin{corollary}\label{cor:mom}
For any $h_{0}$ and  for any $s\ge 1$ we have:
\[
\mathbb{E}[({N}^{w_0}_{h_{0}})^s]=O( \mathbb{E}[({\overline N}^{w_0}_{h_{0}})^s])
\]
\end{corollary}
\begin{proof}
note that by construction:

\[
 \overline{N}^{w_0}_{h_0}\ind_{\{\overline{N}^{w_0}_{h_0}< \lfloor \epsilon n\rfloor \}}+ \lfloor \epsilon n\rfloor \ind_{\{\overline{N}^{w_0}_{h_0} \geq \lfloor \epsilon n\rfloor\}} 
 \le \overline{N}^{w_0}_{h_0}\le \overline{N}^{w_0}_{h_0}\ind_{\{\overline{N}^{w_0}_{h_0}< \lfloor \epsilon n\rfloor \}}+ n\ind_{\{\overline{N}^{w_0}_{h_0}\ge  \lfloor \epsilon n\rfloor\}} 
\]
Moreover note that from Lemma \ref{bound_number_nodes_conf} we have:
\[
{N}^{w_0}_{h_0}\le_{st}  \overline{N}^{w_0}_{h_0}\ind_{\{\overline{N}^{w_0}_{h_0}< \lfloor \epsilon n\rfloor \}}+ n\ind_{\{\overline{N}^{w_0}_{h_0}\ge \lfloor \epsilon n\rfloor\}} 
\]
The assertion follows  immediately.
\end{proof}


Now we introduce a technical result that characterizes the moments of the total number of nodes 
generated in a GW process in which the offspring distribution
follows a generic law $\widehat{q}$. Such result will later on be used to
prove more specific results valid when $\widehat{q}$ either follows a truncated 
power law distribution (Corollary \ref{coro:powerlaw})
or it has all finite moments (Corollary \ref{coro:finitemom}).

\begin{lemma}\label{conjmic}
Let  $\{N_h\}_{h\geq0}$ be a supercritical GW process with power-law degree distribution $\widehat{q} = \{\widehat{q}_k\}_{k\geq0}$ (with $\sum_{k=0}^\infty k\widehat{q}_k>1$). Let $n_h$ be the number of nodes at depth $h$, and $N_h$ be the total number of nodes
generated up to generation $h$.  
These quantities are defined recursively as follows:
$$n_{h+1}=\sum_{i=1}^{n_h}D_i \qquad ; \qquad N_{h+1}= N_h + n_{h+1} $$ 
where $D_i$ are i.i.d. according to $\widehat{q}$, and we start with $n_0 = N_0 = 1$.
Let \mbox{$\widetilde{\mu}_j=\mathbb{E}[(D+1)^j]$}, with $D$ distributed according to $\widehat{q}$. 
We have 
\begin{align*}
\mathbb{E}[N_h^s] = O \left( \sum_{(k_1,...,k_s)\in\mathcal{K}_s}{\widetilde{\mu}}_{\textcolor{black}{k}} {\widetilde{\mu}}_1^{k_1} {\widetilde{\mu}}_2^{k_2} \ldots {\widetilde{\mu}}_s^{k_s}{\widetilde{\mu}}_1^{s(h-2)}\right).
\end{align*} 
where $k= \sum_{j=1}^sk_j$, and the summation is over the following set
$$\mathcal{K}_s=\{(k_1,\ldots,k_s):\sum_{j=1}^sjk_j=s\}.$$
\end{lemma}

\begin{proof}
We first show through a coupling argument that the total number of nodes
$N_h$ is stochastically dominated by the number $Z_h$ of nodes
at depth $h$ in a modified branching process in which the outdegree 
is augmented by one.
Indeed, in such modified branching process we have, recursively:
$$ Z_{h+1}=\sum_{i=1}^{Z_h}(D_i + 1) $$ 
starting again from $Z_0 = 1$.
We prove by induction that $N_{h}\leq Z_{h}$. For $h = 0$ we have
$N_1 = 1 + D = Z_1$, and the assertion is verified.  
We prove now that, if the assertion is true for $h$, then it is 
true for $h+1$. By inductive hypothesis $N_h \leq Z_h$ and
$$
N_{h+1} = N_h + n_{h+1} = N_h + \sum_{i=1}^{n_h}D_i \leq 
N_h + \sum_{i=1}^{N_h}D_i \leq Z_h + \sum_{i=1}^{Z_h}D_i = Z_{h+1}  
$$ 
and thus the assertion is verified. Therefore we will analyze in the
following how the number $Z_h$ of nodes at depth $h$ in the modified
branching process depends on the moments 
$\widetilde{\mu}_j=\mathbb{E}[(D+1)^j]$ with $D\sim \widehat{q}$.

Denote by 
$\Phi_h(t)=\mathbb{E}[ t^{Z_h}]$ 
 the probability generating function of the random variable $Z_h$, and by $\Psi_{\widetilde{q}}$ the probability generating function 
of the modified offspring distribution (i.e., of random variable $D+1$). 
Then the following recursion holds:
$\Phi_0(t)=t$ and 
\begin{equation}\label{eq:rec_prob_gen_function_pl}
\Phi_h(t)=\Phi_{h-1}(\Psi_{\widetilde{q}}(t))=\Psi_{\widetilde{q}}(\Phi_{h-1}(t))\quad\forall h\geq1.
\end{equation}

Let $F_h^s=\mathbb{E}[Z_h(Z_h-1)\cdots (Z_h-s +1)]$ be the factorial moments of $Z_h$. 
$F_h^s$ is obtained by evaluating at $t = 1$ the $s$-th derivative of $\Phi_h(t)$, i.e., we have $F_h^s = \Phi^{(s)}_h(1)$. 

We thus have
$$ F_h^1 =\Psi^{(1)}_{\widetilde{q}}(\Phi_{h-1}(1))\Phi^{(1)}_{h-1}(1)=\widetilde{\mu}_1 F_{h-1}^1 $$
Since $F_1^1 =  \widetilde{\mu}_1$, we obtain by induction $ F_h^1 =  \widetilde{\mu}_1^h$, hence the result is true
for the first moment of $Z_h$.
To show that the result holds for a generic moment $s$, we first show that it is true for the
factorial moments of $Z_h$. 
First, notice that $F_1^{s} = O(\widetilde{\mu}_s)$, $\forall s \geq 1$.
We use the strong induction on $h$, showing that, if the relation holds for $F_{h-1}^s$,  it holds also for  $F_h^s$. 

To differentiate $s$-times $\Phi_{h}(t)$, we apply Fa\`{a} di Bruno formula \cite{Comtet} to \eqref{eq:rec_prob_gen_function_pl} \begin{align*}
\Phi^{(s)}_h(t)=\sum_{\overline k=(k_1,\ldots,k_s) \in\mathcal{K}_s}b_{\overline k}\Psi_{\widetilde{q}}^{(k_1+k_2\ldots+k_s)}(t)\prod_{j=1}^s\left(\Phi_{h-1}^{(j)}(t)\right)^{k_j}
\end{align*}
where $$b_{\overline k}=\frac{s!}{k_1!k_2!\cdots k_s!1!^{k_1}\cdots s!^{k_s}}$$  and the summation is over the following set
$$\mathcal{K}_s=\{(k_1,\ldots,k_s):\sum_{j=1}^sjk_j=s\}.$$

Neglecting factorial terms, which are just some constants for any finite $s$, and using
the inductive step $F_{h-1}^s = O({\widetilde{\mu}}_s \cdot {\widetilde{\mu}}_1^{s(h-2)})$,
we obtain
\begin{align}\label{eq:mic1}
&F_h^s = O \left( \sum_{(k_1,...,k_s)\in\mathcal{K}_s} {\widetilde{\mu}}_{k_1  + \ldots + k_s} \cdot ({\widetilde{\mu}}_1 {\widetilde{\mu}}_1^{h-2})^{k_1} \cdot ({\widetilde{\mu}}_2 {\widetilde{\mu}}_1^{2(h-2)})^{k_2}  
 \cdot \ldots \cdot ({\widetilde{\mu}}_s {\widetilde{\mu}}_1^{s(h-2)})^{k_s} \right) \nonumber\\
 &=O \left( \sum_{(k_1,...,k_s)\in\mathcal{K}_s}{\widetilde{\mu}}_{\textcolor{black}{k}} {\widetilde{\mu}}_1^{k_1} {\widetilde{\mu}}_2^{k_2} \ldots {\widetilde{\mu}}_s^{k_s}{\widetilde{\mu}}_1^{s(h-2)}\right).
\end{align} 
where $k= \sum_{j=1}^sk_j$.
At last we observe that an analogous result holds for the moments $\mathbb{E}[N_h^s]$, whose scaling order is the same as $F_h^s$.
\end{proof}

\begin{corollary}[\textbf{Proof of Lemma \ref{lemma:powerlaw}}]\label{coro:powerlaw}
Let $\{N_h\}_{h\geq0}$ be a supercritical GW process with  
power-law degree distribution $\widehat{q} = \{\widehat{q}_k\}_{k\geq0}$ of exponent $\beta > 1$, truncated
at $\widehat{k}_{\max} = \Theta(n^\zeta)$, $\zeta > 0$. 
We have:
$\mathbb{E}[N_h^s] = O(\widehat{\mu}_s \cdot \widehat{\mu}_1^{s(h-1)})$, $\forall \beta > 1$, 
where $\widehat{\mu}_j$ is the $j$-th moment of $q$. 
\end{corollary}   
\begin{proof}
Let's first consider the extreme case in which all moments $\widehat{\mu}_j$ of $\widehat{q}$
are infinite, including the first one, which happens for $1 < \beta < 2$.
From \equaref{mus1} we have ${\widetilde{\mu}}_j = \Theta(\widehat{\mu}_j) = \Theta(n^{\zeta (j+1-\beta)})$, 
$\forall j \geq 1$.
Plugging the above expression of  ${\widetilde{\mu}}_s$ into \eqref{eq:mic1}, we obtain:
\begin{align}\label{eq:mic2}
 & F_h^s = O \left( \sum_{(k_1,...,k_s)\in\mathcal{K}_s}  {\mu}_k n^{\zeta (k_1(1+1-\beta) + k_2(2+1-\beta) + \ldots k_s(s+1-\beta))} {\mu}_1^{s(h-2)}\right) \nonumber \\
&= O \left( {\mu}_1^{s(h-2)}  \sum_{(k_1,...,k_s)\in\mathcal{K}_s}  n^{\zeta [ s+ 1 - \beta + k (2-\beta)]}  \right) \nonumber \\
 &= O \left( \widehat{\mu}_s \widehat{\mu}_1^{s(h-2)}  \sum_{(k_1,...,k_s)\in\mathcal{K}_s}  n^{\zeta k (2-\beta)}  \right) \nonumber \\
 &= O \left( \widehat{\mu}_s \widehat{\mu}_1^{s(h-2)}   n^{\zeta s (2-\beta)}  \right) = 
O \left(  \widehat{\mu}_s  \widehat{\mu}_1^{s(h-2)}  \widehat{\mu}_1^s  \right) =  O \left(  \widehat{\mu}_s  \widehat{\mu}_1^{s(h-1)} \right)
\end{align}
where we have used the fact that, since we are assuming $\beta < 2$, the dominant term in the summation is the one associated 
to the largest possible value of $k$, obtained when $k_1 = s$, while all others $k_i = 0$, $i > 1$.

Let us now assume that all moments of the degree distribution are finite up to moment $j-1$, whereas
moments of order $j$ or higher are infinite. This happens when $\beta > j$.
Repeating the same passages as before, adding and subtracting the \lq missing' terms corresponding
to finite moments, we get:
\begin{align}\label{eq:mic3}
& F_h^s  = O \left(  \widehat{\mu}_s \widehat{\mu}_1^{s(h-2)}  \sum_{(k_1,...,k_s)\in\mathcal{K}_s}  n^{\zeta ((k-k_1)(2-\beta) - k_2(3-\beta) - \ldots - k_{j-1}(j-\beta))} \right) = \nonumber \\
& O \left(  \widehat{\mu}_s \widehat{\mu}_1^{s(h-2)} \right) = O \left( \widehat{\mu}_s \widehat{\mu}_1^{s(h-1)} \right)
\end{align}
where we have used the fact that, since $\beta > j$, the dominant term is obtained again by choosing
$k = k_1 = s$, while all others $k_i = 0$, $i > 1$. 
\end{proof}
\begin{remark}
In our application to the single-class configuration-model 
with (truncated) power law distribution (Section \ref{sec:asdconfig}),
we are only interested to the case $\beta > 2$ (so that the average degree is finite),
for which we could obtain the  stricter bound $F_h^s = O \left(\widehat{\mu}_s \widehat{\mu}_1^{s(h-2)} \right)$.
However, since we take $h = c \log{n}$,  we are not penalized  by using  looser bound stated in Corollary
\ref{coro:powerlaw}. 
\end{remark}    
\begin{remark}
To apply Corollary \ref{coro:powerlaw} to the single-class configuration-model 
with (truncated) power law distribution (Section \ref{sec:asdconfig}),
one should also consider the fact that the first generation of nodes
in the GW process follows law $p^{(n)}_{k}$, while
the following generations follow law $q^{(n)}_{k}$.
However, by Assumption \ref{ass1}, we have that $p^{(n)}_{k}$ and $q^{(n)}_{k}$
are both $O(k^{-\beta})$, hence $p^{(n)}_{k}$ and $q^{(n)}_{k}$ are 
both stochastically dominated by a power law distribution $\widehat{q}$
of exponent $\beta$, which allows us to apply Corollary \ref{coro:powerlaw}
and obtain a valid bound for our configuration-model.  In particular we can define   $\widehat{q}$ as follows:   
$\sum_{h=0} ^k \widehat{q}_h  =\min (\sum_{h=0} ^k p_h, \sum_{h=0} ^k q_h)$.
\end{remark}

\begin{corollary}\label{coro:finitemom}
Let  $\{N_h\}_{h\geq0}$ be a supercritical GW process with 
degree distribution $\widehat{q} = \{\widehat{q}_k\}_{k\geq0}$ (with $\sum_{k=0}^\infty k\widehat{q}_k>1$)
having all finite moments. 
Let $N_h$ be the total number of nodes
generated up to generation $h$.
We have $\mathbb{E}[N_h^s] = O(\widehat{\mu}_1^{s(h-2)})$, 
where $\widehat{\mu}_1$ is the first moment of $\widehat{q}$. 
\end{corollary}
\begin{proof}
By assumption we have ${\widetilde{\mu}}_j = \Theta(\widehat{\mu}_j) = \Theta(1)$ for any $j \geq 1$.
As direct application of Lemma \ref{conjmic} we get 
\begin{align}\label{eq:mic4}
& \mathbb{E}[N_h^s] =  O \left(\widehat{\mu}_1^{s(h-2)} \right).
\end{align}
\end{proof}

{\color{black}
\textbf{Proof of Theorem \ref{top_res_sbm}}

First observe that the number of edges $W^{b,a}_{t}$   in $\mathcal{T}_t$ between any pair of classes $(a, b)$ 
con be upper bounded by the total number of edges in $\mathcal{T}_t$, which itself can be 
bounded by the total number of nodes in  $\mathcal{T}_t$.

Recall that the number of edges between a node of community $a$ and nodes of community $j$
conforms to the empirical distribution $p^{\texttt{out}}_{a,j}[k]$.

Let $P^{\texttt{out}}_{a,j}$ be the cumulant of $p^{\texttt{out}}_{a,j}$,
and define $P^{\texttt{out}}_*[k]= \min_{a,j} P^{\texttt{out}}_{a,j}[k]$.
Then, let $p^{\texttt{out}}_*[k]$ be the distribution whose cumulant is $P^{\texttt{out}}_*[k]$.


The  outgoing degree  of a generic node  is then 
stochastically dominated by a r.v. distributed as $p^{\texttt{out}}_*$ which has, 
by construction, all finite moments. Indeed note that, since all $p^{\texttt{out}}_{a,j}$
have, by assumption, an exponential tail, $p^{\texttt{out}}_*$ has an exponential tail as well. 

Therefore, we can bound the number of 
edges in $\mathcal{T}_t$  with those in  
a truncated GW tree $\overline{\mathcal{T}}_t$ in which the outgoing degree 
distribution of nodes is $p^{\texttt{out}}_*$.

Then, setting  $x_{b,a} = x_n=n^{\frac{1}{2}-\gamma}$, for any $(a,b) \in \mathcal{A}\times\mathcal{A}$,
from \equaref{bound_topologico} and Lemma \ref{lemma:F_branching} 
we have, for any $s>1$:


\begin{align}\label{eq:bound_topologico2}
&&\|\mu_{\mathcal{T}_{t}}-\mu_{\mathcal{N}_{t}}\|_{\mathrm{TV}}
 \leq  S\, \frac{x_n(x_n+1)}{2 n \bar{d}} + \textcolor{black}{|\mathcal A|^2} 
 \frac{\mathbb{E}[\overline{N}_{h_n}^s]}{x_n^s} +o(1/n)\quad n\rightarrow\infty
\end{align}

where $\overline{N}_{h_n}^s$  is the number of nodes in a truncated  GW process 
of maximum depth $h_n$, in which the offspring distribution of every node is 
$p^{\texttt{out}}_*$.
Moreover, we have introduced 
$S = \sum_{a,b\in\mathcal{A}} \sum_{d} d\,p^{\texttt{in}}_{a,b}[d]$
which is by assumption a finite constant since the number of classes is finite
and the average number of edges going into a node of class $b$ from nodes of class $a$ is finite.  


Now we can apply Corollary \ref{coro:finitemom} to bound $\mathbb{E}[\overline{N}_{h_n}^s] $, 
since all moments of
$p^{\texttt{out}}_*$ are finite. 

In particular, from Corollary \ref{coro:finitemom} we 
have $\mathbb{E}[\overline N_{h_n}^s]=\Theta(n^{{cs}\log\mu_1})$, therefore: 
\begin{align*}
\frac{\mathbb{E}[\overline N_{h_n}^s]}{x_n^s}=\Theta(n^{-s\left(\frac{1}{2}-\gamma-c\log\mu_1\right)})=o(1/n)\qquad n\rightarrow\infty
\end{align*}
and we can conclude  
\begin{align*}
\|\mu_{\mathcal{T}_{t}}-\mu_{\mathcal{N}_{t}}\|_{\mathrm{TV}}&\leq\frac{|{\mathcal A}|^2S}{2\overline{d}n^{2\gamma}}+o(1/n)=o(1)\qquad n\rightarrow\infty.
\end{align*}
\qed
} 

\end{document}